\documentclass[11pt]{article}
\pdfoutput=1

\usepackage[T1]{fontenc}
\usepackage{longtable}
\usepackage[margin=1in]{geometry}
\usepackage[left]{lineno}
\usepackage{amsmath,amsthm,amssymb,amsfonts,titling,bbm,titlesec,bm,physics,enumitem,accents,setspace,fancyhdr,pdflscape}
\usepackage{graphicx}
\usepackage{float}
\usepackage{placeins}
\usepackage[font=footnotesize,labelfont=bf]{caption}
\usepackage{array}
\usepackage[title]{appendix}
\usepackage{afterpage}
\usepackage{listings}
\usepackage{prodint}
\usepackage{tocloft}
\usepackage{algorithm}
\usepackage{booktabs}
\usepackage{multirow}
\usepackage{subcaption}

\newlength{\continueindent}
\usepackage{etoolbox}
\newcommand{\slightspacing}{\setstretch{1.15}}
\slightspacing

\usepackage[dvipsnames]{xcolor}
\usepackage{siunitx}
\usepackage{soul}
\definecolor{Bleu}{RGB}{0,0,204}

\usepackage{tikz}
\usetikzlibrary{tikzmark,arrows.meta,calc}
\definecolor{CB5red}{HTML}{DA0000}
\definecolor{CB5purple}{HTML}{8000B3}

\usepackage[hypertexnames=false]{hyperref}
\hypersetup{
  colorlinks,
  citecolor=Bleu,
  linkcolor=Bleu,
  urlcolor=Bleu}

\usepackage{amsthm}

\makeatletter
\def\@endtheorem{\endtrivlist}
\makeatother

\usepackage{upgreek}
\usepackage{mathrsfs}
\usepackage[noend]{algpseudocode}
\usepackage[round]{natbib}
\usepackage{calc}
\usepackage{authblk}

\theoremstyle{plain}

\newtheorem{theorem}{Theorem}
\newtheorem{corollary}[theorem]{Corollary}

\newtheorem{lemma}[theorem]{Lemma}
\newtheorem{proposition}[theorem]{Proposition}
\newtheorem{assumption}{Assumption}

\DeclareMathOperator*{\argmin}{argmin}

\usepackage[capitalize,noabbrev]{cleveref}
\newcommand\independent{\protect\mathpalette{\protect\independenT}{\perp}}
\def\independenT#1#2{\mathrel{\rlap{$#1#2$}\mkern2mu{#1#2}}}

\DeclareFontFamily{U}{jkpmia}{}
\DeclareFontShape{U}{jkpmia}{m}{it}{<->s*jkpmia}{}
\DeclareFontShape{U}{jkpmia}{bx}{it}{<->s*jkpbmia}{}
\DeclareMathAlphabet{\mathfrak}{U}{jkpmia}{m}{it}
\SetMathAlphabet{\mathfrak}{bold}{U}{jkpmia}{bx}{it}

\DeclareMathOperator*{\Var}{Var}

\newcommand{\E}[2]{\ensuremath{{\mathbb E}_{#1}\left[#2\right]}}

\usepackage[hang,flushmargin]{footmisc}

\definecolor{codegreen}{rgb}{0,0.6,0}
\definecolor{codegray}{rgb}{0.5,0.5,0.5}
\definecolor{codepurple}{rgb}{0.58,0,0.82}
\definecolor{backcolour}{rgb}{0.95,0.95,0.92}

\lstdefinestyle{unadjusted}{
    backgroundcolor=\color{backcolour},
    commentstyle=\color{codegreen},
    keywordstyle=\color{blue},
    keywords={},
    numberstyle=\tiny\color{codegray},
    stringstyle=\color{codepurple},
    basicstyle=\fontfamily{zi4}\selectfont\footnotesize,
    breakatwhitespace=false,
    breaklines=true,
    captionpos=b,
    keepspaces=true,
    numbers=left,
    numbersep=5pt,
    showspaces=false,
    showstringspaces=false,
    showtabs=false,
    tabsize=2
    }

\lstdefinestyle{adjusted}{
    backgroundcolor=\color{backcolour},
    commentstyle=\color{codegreen},
    keywordstyle=\color{blue},
    morekeywords={w, calculate_weights, bootstrap_standard_error, weighted, T, weights},
    numberstyle=\tiny\color{codegray},
    stringstyle=\color{codepurple},
    basicstyle=\ttfamily\footnotesize,
    breakatwhitespace=false,
    breaklines=true,
    captionpos=b,
    keepspaces=true,
    numbers=left,
    numbersep=5pt,
    showspaces=false,
    showstringspaces=false,
    showtabs=false,
    tabsize=2
}

\lstdefinestyle{adjusted_cox}{
    backgroundcolor=\color{backcolour},
    commentstyle=\color{codegreen},
    moredelim=**[is][\bfseries\color{Bleu}]{|}{|},
    numberstyle=\tiny\color{codegray},
    stringstyle=\color{codepurple},
    basicstyle=\fontfamily{zi4}\selectfont\footnotesize,
    breakatwhitespace=false,
    breaklines=true,
    captionpos=b,
    keepspaces=true,
    numbers=left,
    numbersep=5pt,
    showspaces=false,
    showstringspaces=false,
    showtabs=false,
    tabsize=2
    }

\usepackage{titletoc}
\newcommand\DoToC{%
  \startcontents
  \printcontents{}{1}{\vskip 1.5em\hrule\vskip .75em}
  \vskip .75em\hrule\vskip 2em
}

\title{Simple Covariate Adjustment for Many Estimands\\
Using Stable Balancing Weights}
\date{\today}

\author[1]{Kayla Irish, José Zubizarreta, and Alex Luedtke}

\begin{document}

\allowdisplaybreaks
\maketitle

\begin{abstract}
    \noindent Covariate adjustment can improve precision in estimating treatment effects in clinical trials, and regulatory guidance increasingly encourages its implementation. Despite its benefits, covariate adjustment can require advanced statistical techniques or tailored approaches for different estimands, which can deter its use in practice. To address this issue, we introduce a simplified approach to covariate adjustment using stable balancing weights that applies directly across many clinical trial estimands, including the average treatment effect, relative risk, Mann–Whitney estimand, and survival ratio. More precisely, our results cover any estimand that is a Hadamard differentiable functional of the arm-specific distributions. Once the weights are obtained, our adjusted estimator can be implemented with the same software used for an unadjusted analysis, provided that software can take observation weights. This construction improves asymptotic efficiency relative to unadjusted estimation and preserves the plug-in relationship between arm-specific marginal and subgroup-specific summaries.
\end{abstract}

\section{Introduction}

We introduce a straightforward and interpretable approach to covariate adjustment. Our method leverages variables measured at baseline to balance the treatment and control groups to gain precision in estimating treatment effects, yielding narrower confidence intervals and more powerful hypothesis tests.

In clinical trials, covariate adjustment involves the pre-planned use of prognostic baseline variables, such as demographic factors, disease characteristics, or other patient information collected at the time of randomization, to estimate a treatment's effect. The ICH E9 Guidance on Statistical Methods for Analyzing Clinical Trials (\citeyear{international1998statistical}) emphasizes the importance of adjusting for covariates that are measured before randomization and are expected to be correlated with the primary outcomes of the trial. This adjustment improves precision and accuracy by correcting for any covariate imbalances between treatment groups. As a result, statistical power often increases, which reduces the required sample size needed to determine the effect of the treatment. Smaller trials are valuable both practically and ethically, consuming fewer resources and exposing fewer participants to a treatment with unproven benefit. The FDA's recent guidance further supports this approach, noting that incorporating prognostic baseline factors in the primary analysis can improve precision with minimal impact on bias and Type I error rate \citep{fda2023adjusting}.

Despite these recommendations, covariate adjustment is not yet routine in practice \citep{kahan2014risks, van2024covariate}. Many randomized controlled trials do not fully exploit available baseline covariates, in part because implementation requires prespecifying covariates, functional forms, and estimand-specific adjustment procedures \citep{fda2023adjusting}. Consequently, the primary analysis often remains unadjusted, meaning that it does not incorporate baseline covariates. This approach is appealing because it is simple, yields an unbiased point estimate, and supports standard valid confidence intervals. Its main drawback is that it leaves precision gains unused. A central goal of this paper is therefore to lower the practical barriers to covariate adjustment by providing a simple, transparent approach that can be used across a wide range of estimands.

Extensive literature has shown that covariate adjustment can improve precision in randomized trials \citep{tsiatis2008covariate,lin2013agnostic}. Recent work develops adjustment procedures for particular endpoints, including survival and log-rank summaries \citep{diaz2019improved,ye2024covariate,zhang2025unified}, restricted mean survival time \citep{li2022restricted}, and win-based estimands \citep{cao2025covariate,scheidegger2026covariate}; \citet{benkeser2021improving} span binary, ordinal, and time-to-event outcomes, with a separate estimator for each.
Although these methods can yield efficiency gains, their implementation depends on the estimand and may require selecting a tailored procedure. This motivates covariate-adjustment approaches that are prespecifiable, transparent, and reusable across estimands. Our goal is to provide such a template: a single adjustment strategy that can be paired with plug-in estimation for a broad class of treatment-effect summaries.

Weighting methods offer one route toward more transparent covariate adjustment. Inverse probability weighting (IPW) based on an estimated treatment-assignment model can improve precision by using the fitted model to account for covariate imbalance \citep{shen2014inverse,williamson2014variance}. Recent work by \citet{shao2026inverse} extends this idea to right-censored survival endpoints, showing that weighted Cox and Kaplan--Meier analyses can improve efficiency and are asymptotically equivalent to optimal linear augmentation estimators. Our approach differs: rather than estimating a treatment-assignment model, we choose weights by solving a covariate-balancing problem. This connects the method to the broader balancing-weights literature, where weights are designed to improve sample balance through optimization \citep{zubizarreta2015stable,ben2021balancing}.

Linear regression provides a useful benchmark for our weighting approach. For the average treatment effect, ordinary least squares adjustment with treatment-by-covariate interactions is asymptotically no less efficient than the unadjusted difference in means under randomization, without requiring the linear model to be correctly specified \citep{lin2013agnostic}. Recent work shows that linear regression estimators can be written as weighted averages of observed outcomes, with weights determined implicitly by the regression fit \citep{chattopadhyay2023implied}, and estimators that combine outcome modeling with balancing weights can reduce to a single regression estimator \citep{bruns2026augmented}. These works show that weighting and regression are different representations of the same estimator under linear models, for estimands that are linear functionals of a regression. In particular, for the average treatment effect, the stable balancing weighted estimator coincides numerically with Lin's \citeyearpar{lin2013agnostic} estimator.

This class of estimands, however, does not capture many of those commonly targeted in primary analyses of clinical trials, such as risk ratios, Mann--Whitney effect measures, and Kaplan--Meier contrasts. We therefore provide a weighting method and uncertainty quantification for the class of Hadamard differentiable functionals of the arm-specific distributions \citep{gill2006lectures}, which contains these primary-analysis estimands along with many commonly used linear functionals of regressions, including the average treatment effect. For linear functionals of regressions, \citet{kong2025asymptotics} recently established asymptotic normality and confidence intervals for a related minimax-weighting estimator, using a tailored variance estimator built from a preliminary regression fit. We establish asymptotic normality for our estimator across the class of Hadamard differentiable estimands and show that valid confidence intervals instead follow from a familiar nonparametric bootstrap. Our estimator also never loses asymptotic precision relative to the unadjusted estimator, without requiring a correctly specified outcome model.

Our approach builds on stable balancing weights (SBWs) \citep{zubizarreta2015stable} to develop a general template for covariate-adjusted plug-in estimation in randomized trials. Stable balancing weights, obtained through convex optimization, are minimum-variance observation weights that balance covariates between treated and control groups. In their original observational-study setting, these weights are used in causal inference problems where identification depends on no unmeasured confounding. In randomized trials, identification is resolved by design, so the weights can be used purely to improve efficiency. Moreover, the weights depend only on the baseline covariates and treatment assignment, and not on the outcome, so a single prespecified weight vector can be computed before outcome unblinding. The resulting estimator retains the simple form of the unadjusted plug-in estimator, with each arm's empirical distribution replaced by a weighted one. The estimand is therefore unchanged: adjustment does not convert an unconditional treatment effect into a conditional one, a distinction that matters for non-collapsible summaries such as odds ratios and hazard ratios \citep{greenland1999confounding, fda2023adjusting}.

The contributions of this work are as follows:
\begin{enumerate}
    \item We present a unified weighting approach for covariate adjustment that is straightforward to implement across many estimands of practical interest in randomized trials.
    \item We establish that our method improves efficiency compared to unadjusted estimation.
    \item We derive conditions for valid bootstrap confidence intervals and hypothesis tests that, up to mild regularity, match those for unadjusted Wald-type inference.
    \item We show that, since SBW adjustment respects the plug-in principle, it preserves the same relationship between overall and subgroup-specific estimates as unadjusted estimation, keeping subgroup analyses consistent with the overall result.

    \item We illustrate the generality of our approach by applying it to the average treatment effect, the relative risk, the survival ratio, and the Mann--Whitney estimand.
\end{enumerate}
We also apply our method to evaluate the efficacy of a broadly neutralizing antibody for the prevention of HIV, using data from the Antibody-Mediated Prevention (AMP) trials \citep{corey2021two}. There, the estimand is a ratio of cumulative incidences under right censoring.

\section{Setup}

\subsection{Sample and estimands}\label{sec:sample_and_estimands}

Consider a randomized controlled trial in which participants are assigned to a treatment or control group. For each participant, a vector of baseline variables $X \in \mathbb{R}^q$ is recorded at randomization; these may include measurements assumed to be prognostic of the post-treatment observation(s). Let $A$ denote randomized treatment assignment that is independent of $X$, with $A=0$ for control and $A=1$ for treatment. After treatment administration, a post-treatment observation (or set of observations) $Y \in \mathcal{Y}$ is measured. This variable may represent any post-treatment quantity of interest: for example, an outcome or set of outcomes that is continuous, ordinal, or categorical. The participant data ${(X_i, A_i, Y_i)}_{i=1}^n$ are assumed to be $n$ independent draws from $P$, with the treatment assignment probability bounded away from 0 and 1, and with treated and control group sizes given by $N_1=\sum_{i=1}^n A_i$ and $N_0=\sum_{i=1}^n (1-A_i)$. Define $P_a$ as the distribution of $(X,Y)$ given $A=a$.

\begin{table}[t]
\centering
\begingroup
\scriptsize
\setlength{\tabcolsep}{3.55pt}
\renewcommand{\arraystretch}{1.4}
\begin{tabular}{l@{\hspace{1.5em}}l}
\toprule
\textbf{Estimand} & $\boldsymbol{\Psi(P_0,P_1)}$ \\
\midrule
Average treatment effect & $\E{P_1}{Y}-\E{P_0}{Y}$ \\
Relative risk            & $\E{P_1}{Y}\,/\,\E{P_0}{Y}$ \\
Mann--Whitney estimand   & $P(Y_1>Y_0)+\tfrac12 P(Y_1=Y_0)$ \\
Survival ratio           & $S_1(t_0)\,/\,S_0(t_0)$ \\
\bottomrule
\end{tabular}
\endgroup
\caption{Examples of the estimands $\Psi(P_0,P_1)$ considered in this paper, where $Y_a\sim P_a$ are independent and $S_a(t)=P_a\{T>t\}$.}
\label{tab:estimand_examples}
\end{table}

Our objective is to obtain inference on a scalar summary $\psi := \Psi(P_0,P_1)$ of the population distributions $P_0$ and $P_1$, where the parameter $\Psi$ maps a pair of arm-specific distributions to the real line. Note that $\Psi$ need not be a linear functional of a regression. In many cases, $\psi$ is a contrast; see Table~\ref{tab:estimand_examples} for examples. The framework also extends to estimands such as odds ratios and CDFs at fixed thresholds. Adjusted estimation for several of these estimands, such as the Mann--Whitney estimand and the survival ratio, currently requires tailored, estimand-specific procedures. By contrast, a straightforward approach that applies uniformly across these estimands is the unadjusted plug-in estimator $\widehat\psi^{\,\mathrm{unadj}} := \Psi(P_{n,0},P_{n,1})$, where $P_{n,0}$ is the empirical distribution of the control group data $\{(X_i,Y_i) : A_i=0, i\in [n]\}$ and $P_{n,1}$ is the empirical distribution of the treatment group data $\{(X_i,Y_i) : A_i=1, i\in [n]\}$. For the average treatment effect, the plug-in estimator returns the difference in empirical means between the treatment and control groups: $\widehat\psi^{\,\mathrm{unadj}} = \E{P_{n,1}}{Y} - \E{P_{n,0}}{Y}$; for the relative risk, it returns the ratio of empirical means: $\widehat\psi^{\,\mathrm{unadj}} = \E{P_{n,1}}{Y}/\E{P_{n,0}}{Y}$. While the plug-in estimator is unbiased, consistent, and easy to compute, it is statistically inefficient. This arises because the marginal distribution of the covariates may differ under $P_{n,0}$ and $P_{n,1}$, even though randomization ensures it is the same under $P_0$ and $P_1$ \citep{van2024covariate}.

\subsection{Illustration of the approach}
\label{section:illustration}

To illustrate the ideas introduced above, we consider a concrete example structured to resemble ``Table 1'' of a clinical trial primary publication, which summarizes baseline covariate distributions by treatment arm. We adapt the data-generating process of \citet{bannick2023general} with modifications to the model coefficients. We simulate a randomized trial with 350 participants assigned to treatment or control at a 1:1 ratio, and two baseline covariates representing age and sex.

We generate three post-treatment outcomes to illustrate that the same weighting approach can be used across data types and treatment-effect summaries. First, we generate a binary outcome and estimate both the average treatment effect (ATE) and relative risk (RR). Second, we generate an ordinal outcome and estimate the Mann--Whitney (MW) estimand, which compares randomly selected treated and control participants, counting tied pairs with weight one half. Third, we generate a survival outcome and estimate the survival ratio (SR) at a fixed time point $t_0$, defined as the ratio of the marginal survival probabilities in the two treatment arms. In each case, the outcome depends on treatment assignment, baseline covariates, and treatment-covariate interactions. Full details of the data-generating mechanisms and inference procedures are given in Appendix~\ref{appendix:illustration_details}.

In this randomized setting, covariate imbalances can still occur by chance. The top half of Table~\ref{tab:combined_table} shows the covariate means before and after applying SBWs: the weights are chosen so that, within each treatment arm, the weighted covariate means match the unweighted covariate means of the entire study population. The bottom half reports the resulting treatment-effect estimates and their precision.

\begin{table}[t]
    \centering
    \renewcommand{\arraystretch}{1.4}
    \begingroup
    \setlength{\tabcolsep}{3.45pt}
    \scriptsize
        \begin{tabular}{c c cc cc}
            \toprule
            \multicolumn{2}{c}{} & \multicolumn{2}{c}{\textbf{Unadjusted}} & \multicolumn{2}{c}{\textbf{Adjusted}} \\
            \cmidrule(lr){3-4} \cmidrule(lr){5-6}
            \textbf{Covariate} & & \textbf{Treated} & \textbf{Control} & \textbf{Treated} & \textbf{Control} \\
            \midrule
            Age (mean) & & \tikzmark{uaTL}39.87 & 40.41 & \tikzmark{adTL}40.16 & 40.16 \\
            Female (\%) & & 0.494 & 0.452\tikzmark{uaBR} & 0.471 & 0.471\tikzmark{adBR} \\
            \bottomrule
        \end{tabular}

    \vspace{1em}

    \begin{tabular}{c cc cc}
    \toprule
    & \multicolumn{2}{c}{\textbf{Unadjusted}} & \multicolumn{2}{c}{\textbf{Adjusted}} \\
    \cmidrule(lr){2-3} \cmidrule(lr){4-5}
    \textbf{Estimand} & \textbf{Estimate} & \textbf{95\% CI Width} & \textbf{Estimate} & \textbf{95\% CI Width} \\
    \midrule
    ATE & $-0.19$ & 0.19 & $-0.16$ & \tikzmark{ciTL}0.17 \\
    RR & 0.76 & 0.22 & 0.80 & 0.19 \\
    MW & 0.50 & 0.12 & 0.53 & 0.09 \\
    SR & 1.38 & 0.49 & 1.29 & 0.41\tikzmark{ciBR} \\
    \bottomrule
    \end{tabular}
    \endgroup

\begin{tikzpicture}[remember picture, overlay,
    annot/.style={font=\bfseries\footnotesize, inner sep=1pt},
    abox/.style={rounded corners=1pt, line width=0.8pt},
    arr/.style={-{Stealth[length=5pt,width=4pt]}, line width=1.2pt,
                shorten <=2.5pt, shorten >=2.5pt}]

    \coordinate (ua1) at ($(pic cs:uaTL)+(-0.12cm, 0.30cm)$);
    \coordinate (ua2) at ($(pic cs:uaBR)+( 0.12cm,-0.15cm)$);
    \draw[abox, CB5red] (ua1) rectangle (ua2);

    \coordinate (ad1) at ($(pic cs:adTL)+(-0.12cm, 0.30cm)$);
    \coordinate (ad2) at ($(pic cs:adBR)+( 0.12cm,-0.15cm)$);
    \draw[abox, CB5purple] (ad1) rectangle (ad2);

    \coordinate (ci1) at ($(pic cs:ciTL)+(-0.16cm, 0.30cm)$);
    \coordinate (ci2) at ($(pic cs:ciBR)+( 0.16cm,-0.15cm)$);
    \draw[abox, CB5purple] (ci1) rectangle (ci2);

    \coordinate (adM) at ($(ad1)!0.5!(ad2)$);   \coordinate (adR) at (adM -| ad2);
    \coordinate (ciM) at ($(ci1)!0.5!(ci2)$);   \coordinate (ciR) at (ciM -| ci2);

    \node[annot, CB5red, anchor=east, align=right]
        (olab) at ($(ua1|-ua2)+(-3.4cm,-0.34cm)$) {Chance imbalance\\at baseline};
    \draw[arr, CB5red] (olab.east) -- ($(ua1|-ua2)+(-0.05cm,0.06cm)$);

    \node[annot, CB5purple, anchor=west, align=left]
        (tlab) at ($(adR)+(0.75cm,0)$) {Exact mean balance\\in both arms};
    \draw[arr, CB5purple] (tlab.west) -- ($(adR)+(0.06cm,0)$);

    \node[annot, CB5purple, anchor=west, align=left]
        (clab) at ($(ciR)+(0.75cm,0)$) {Equivalent to enrolling\\23--78\% more\\participants};
    \draw[arr, CB5purple] (clab.west) -- ($(ciR)+(0.06cm,0)$);
\end{tikzpicture}
\caption{Baseline characteristics before and after stable balancing weighting (top), and unadjusted and adjusted estimates for the average treatment effect (ATE), relative risk (RR), Mann--Whitney (MW) estimand, and survival ratio (SR) at $t_0=3$ (bottom). Treatment and control sample sizes are 162 and 188, respectively. 95\% confidence interval (CI) widths are reported; details on interval construction are given in Appendix~\ref{appendix:illustration_details}. The true values are $-0.14$ for the ATE, $0.81$ for the RR, $0.53$ for the MW estimand, and $1.16$ for the SR.}
\label{tab:combined_table}
\end{table}

The following R code illustrates how SBWs can be incorporated into standard estimation pipelines. The unadjusted estimator computes Kaplan--Meier curves within each arm and takes the ratio of the estimated survival probabilities at $t_0$:
\begin{lstlisting}[style=unadjusted]
fit = survfit(Surv(time, event) ~ A, data = data)
S0 = summary(fit, times = t0)$surv[1]                    # A = 0 survival at t0
S1 = summary(fit, times = t0)$surv[2]                    # A = 1 survival at t0
unadjusted_SR = S1 / S0
\end{lstlisting}
\noindent The code for the SBW-adjusted estimator is displayed below. \textbf{\color{Bleu}The only change} relative to the block above is the use of weights in the Kaplan--Meier fit on lines 1--2.
\begin{lstlisting}[style=adjusted_cox]
|w = calculate_weights(X, A)|
fit = survfit(Surv(time, event) ~ A, data = data, |weights = w|)
S0 = summary(fit, times = t0)$surv[1]
S1 = summary(fit, times = t0)$surv[2]
adjusted_SR = S1 / S0
\end{lstlisting}
The weights \texttt{w} are obtained by solving the optimization problem described at the beginning of Section~\ref{section:algorithm}. The use of SBWs does not require specialized estimation procedures; any software that accepts observation weights can be applied directly by supplying \texttt{w} as input. 

The bottom half of Table~\ref{tab:combined_table} compares treatment-effect estimates and corresponding measures of uncertainty for unweighted and SBW analyses. Across all four estimands in this illustrative dataset, SBW adjustment reduces the width of the 95\% confidence interval while preserving the same plug-in structure as the unadjusted analysis. To quantify the resulting precision improvement, we report the estimated relative efficiency gain,
$
(\widehat{\mathrm{SE}}_{\mathrm{unadj}}/\widehat{\mathrm{SE}}_{\mathrm{sbw}})^2 - 1,
$
computed on the scale used for inference. In this illustrative dataset, the use of SBWs corresponds to estimated relative efficiency gains of 25\% for the ATE, 45\% for the relative risk, 78\% for the MW estimand, and 23\% for the SR. The relative efficiency gain can be interpreted as the approximate increase in sample size an unadjusted analysis would need to attain the same precision as our SBW estimator.

\section{Algorithm}\label{section:algorithm}

We now describe the steps for estimating $\psi$ using SBWs. This approach involves solving a convex optimization problem that minimizes the variance of the SBWs, subject to three constraints: exact balance on the weighted covariate means (equivalently, an imbalance tolerance of $\delta=0$ in \citet{zubizarreta2015stable}'s formulation), non-negativity of the weights, and the weights having a mean of 1. If the non-negativity constraint is not active, the SBWs have the closed-form solution 
\begin{equation}\label{eqn:w}
    w_a = \mathbf 1_a-\boldsymbol{X}_a (\boldsymbol{X}_a^\top \boldsymbol{X}_a)^{-1}(\boldsymbol{X}_a^\top \mathbf 1_a-N_a \bar{X}_n),
\end{equation}
where, for treatment group $a$ with $q$ baseline covariates, $\mathbf 1_a$ is a vector of ones of length $N_a$, and $\boldsymbol{X}_a \in \mathbb{R}^{N_a \times (q+1)}$ denotes the arm-$a$ covariate matrix augmented with an intercept column of ones. Accordingly, $\bar{X}_n \in \mathbb{R}^{q+1}$ includes the full-sample mean of the intercept (equal to 1) and the $q$ covariate means. In Appendix~\ref{sec:SBW-Optimization-Problem}, we give a full specification of the optimization problem and show that the non-negativity constraint is inactive with probability tending to one under mild conditions.

Algorithm~\ref{alg:gsbw} presents the steps for obtaining the covariate-adjusted estimator using SBWs. Because the SBWs balance each treatment arm to the full-sample covariate distribution, they use baseline information while preserving the arm-specific plug-in structure of the estimator. The weights are used to define the arm-specific distribution estimates $\widehat P_{a}^{\,\mathrm{sbw}}$; importantly, the weights do not change the estimand encoded by $\Psi$.
We refer to $\widehat{\psi}^{\,\mathrm{sbw}}:=\Psi(\widehat P_{0}^{\,\mathrm{sbw}},
      \widehat P_{1}^{\,\mathrm{sbw}})$ as the SBW estimator of $\psi$.

\begin{algorithm}[t]
    \caption{Construction of the SBW plug-in estimator}
    \begin{algorithmic}[1]
    \Require Data $\mathcal D_n=\{(X_i,A_i,Y_i)\}_{i=1}^n$ and full-sample covariate mean $\bar X_n = n^{-1}\sum_{i=1}^n X_i$.
    \State \textbf{For each} treatment arm $a\in\{0,1\}$:
    \State \hspace{\algorithmicindent}\textbf{Weights.} Find minimum-variance SBWs $\{w_{a,i}\}$ satisfying $N_a^{-1}\sum_{i:A_i=a}w_{a,i}X_i=\bar X_n$.
    \State \hspace{\algorithmicindent}\textbf{Weighted empirical distribution.} Define, for $\mathcal B\subseteq\mathcal X\times\mathcal Y$,
        \[
        \widehat P_{a}^{\,\mathrm{sbw}}(\mathcal B)
        :=
        \frac{1}{N_a}\sum_{i:A_i=a}
        w_{a,i}\mathbbm{1}\{(X_i,Y_i)\in \mathcal B\}.
        \]
    \State \textbf{Return} Plug-in estimator $\widehat{\psi}^{\,\mathrm{sbw}} := \Psi(\widehat P_{0}^{\,\mathrm{sbw}},
      \widehat P_{1}^{\,\mathrm{sbw}}).$
    \end{algorithmic}
    \label{alg:gsbw}
\end{algorithm}

To quantify uncertainty in this estimator, we implement a nonparametric bootstrap procedure in Algorithm~\ref{alg:bootstrap-z}. This method recalculates the SBWs for each bootstrap resample and then re-estimates the treatment effect using the updated weights. Here, $z_{1-\alpha/2}$ denotes the $(1-\alpha/2)$ quantile of the standard normal distribution. Because the weights are estimated from the data, inference should account for this estimation step, rather than relying on standard software's model-based weighted standard errors.

\begin{algorithm}[t]
    \caption{Bootstrap inference for the SBW plug-in estimator}
    \begin{algorithmic}[1]
    \Require Data $\mathcal D_n=\{(X_i,A_i,Y_i)\}_{i=1}^n$; bootstrap replicates $B$; significance level $\alpha$.
    \State \textbf{Point estimate.} Compute $\widehat{\psi}^{\,\mathrm{sbw}}$ from $\mathcal D_n$ using Algorithm~\ref{alg:gsbw}.
    \State \textbf{For each} bootstrap replicate $b \in \{1,\ldots,B\}$:
    \State \hspace{\algorithmicindent}\textbf{Resample.} Draw $\mathcal D_n^{*(b)}$ by sampling $n$ observations with replacement from $\mathcal D_n$.
    \State \hspace{\algorithmicindent}\textbf{Re-estimate.} Apply Algorithm~\ref{alg:gsbw} to $\mathcal D_n^{*(b)}$ to obtain $\widehat{\psi}^{\,\mathrm{sbw},*(b)}$.
    \State \textbf{Standard error.} Compute
    $
    \widehat{\mathrm{SE}}_{\mathrm{boot}}
    :=
    \mathrm{SD}\{
    \widehat{\psi}^{\,\mathrm{sbw},*(1)},\ldots,
    \widehat{\psi}^{\,\mathrm{sbw},*(B)}
    \}
    $.
    \State \textbf{Return} Wald interval
    $
    \left(
    \widehat{\psi}^{\,\mathrm{sbw}}
    -
    z_{1-\alpha/2}\widehat{\mathrm{SE}}_{\mathrm{boot}},
    \,
    \widehat{\psi}^{\,\mathrm{sbw}}
    +
    z_{1-\alpha/2}\widehat{\mathrm{SE}}_{\mathrm{boot}}
    \right)
    $. 
    \end{algorithmic}
    \label{alg:bootstrap-z}
\end{algorithm}

Two features of this construction deserve comment. First, the nonnegativity and mean-one constraints make $\widehat P_a^{\,\mathrm{sbw}}$ a probability distribution in every finite sample, not just asymptotically. This matters because many of our target functionals are undefined on signed measures: ratios need nonzero denominators, log scales need positive arguments, and weighted Kaplan--Meier needs positive at-risk denominators. Second, for the average treatment effect, our SBW estimator is already familiar: when the nonnegativity constraint is inactive, our estimator is numerically identical to the fully interacted ANCOVA estimator of \citet{lin2013agnostic}. Our work makes the same weights usable for generic Hadamard differentiable estimands.

\section{Theoretical properties of the SBW estimator}
\label{section:theory}

\subsection{Large-sample inference and efficiency}
\label{subsection:large_sample_theory}

We now present the main theoretical properties of the generalized SBW estimator, with full details provided in Appendix~\ref{section:theoretical_results}. These results rely on the assumption that the functional of interest $(P_0,P_1)\mapsto\Psi(P_0,P_1)$ is sufficiently smooth, formalized through Hadamard differentiability. This condition permits use of the functional delta method to establish asymptotic normality of $\widehat{\psi}^{\,\mathrm{sbw}}$ and consistency of the bootstrap standard error \citep[Theorems~20.8, 23.9]{van2000asymptotic}.

Throughout, we assume: (i) $P(A=1)$ is bounded away from $0$ and $1$; (ii) $X$ has bounded support; (iii) $\mathbb E_{P_a}[XX^\top]$ is nonsingular for $a\in\{0,1\}$, meaning that the balancing variables contain no exact linear redundancies within either treatment arm; and (iv) $\Psi$ is Hadamard differentiable at $(P_0,P_1)$ in the directions generated by the SBW construction, as formalized in Appendix~\ref{sec:differentiability-Phi}. Appendix~\ref{sec:estimand-regularity} verifies Hadamard differentiability for several common clinical-trial estimands, including those in Section~\ref{section:illustration}. Conditions (i) and (iv) alone already give asymptotic normality and bootstrap validity for the unadjusted estimator \citep[Theorems~20.8, 23.9]{van2000asymptotic}; SBW adds only the mild conditions (ii)--(iii) on the covariate distribution.

\begin{theorem}[Asymptotic normality and bootstrap consistency]
\label{thm:asymp_boot}
Under the regularity conditions above, the SBW estimator
$\widehat{\psi}^{\,\mathrm{sbw}}$ computed from
Algorithm~\ref{alg:gsbw} satisfies
\[
\sqrt n(\widehat{\psi}^{\,\mathrm{sbw}}-\psi)
\rightsquigarrow
N(0,\sigma_{\mathrm{sbw}}^2)
\]
for $\sigma_{\mathrm{sbw}}^2<\infty$. The $\sqrt n$-scaled bootstrap standard error in Algorithm~\ref{alg:bootstrap-z} is consistent for $\sigma_{\mathrm{sbw}}$.
\end{theorem}
Theorem~\ref{thm:asymp_boot} justifies Wald-type inference using the bootstrap standard error from Algorithm~\ref{alg:bootstrap-z}. Alternatively, percentile bootstrap intervals can be formed from the empirical quantiles of these estimates. Both approaches follow from bootstrap distributional consistency, with intervals constructed on the appropriate scale for the estimand. In Appendix~\ref{section:theoretical_results}, we give a multivariate generalization of Theorem~\ref{thm:asymp_boot}. We also show that, under a standard estimand-specific continuity condition, the SBW estimator has an asymptotically linear representation with a closed-form influence function \citep[Theorem~20.8]{van2000asymptotic}. This representation also facilitates standard error estimation via a plug-in estimate of the influence function, as an alternative to the bootstrap.

Briefly, the proof represents the SBW estimator as the plug-in evaluation of a functional $\Phi$. Under randomization, $\Phi(P)$ coincides with the target estimand $\Psi(P_0,P_1)$. The functional $\Phi$ is built by chaining together the maps that define the arm-specific laws, the SBWs, the weighted arm-specific summaries, and the target functional $\Psi$. We verify Hadamard differentiability of each component and apply the chain rule
\citep[Theorem~20.9]{van2000asymptotic}. The functional delta method then yields a mean-zero normal limit, and bootstrap consistency follows from the bootstrap delta method \citep[Theorems 20.8 and 23.9]{van2000asymptotic}. The full proof is given in Appendix~\ref{section:theoretical_results}.

We next compare the large-sample precision of the SBW and unadjusted plug-in estimators introduced in Section~\ref{sec:sample_and_estimands}. Let $\sigma_{\mathrm{unadj}}^2$ denote the asymptotic variance of $\sqrt n(\widehat\psi^{\,\mathrm{unadj}}-\psi)$. This variance comparison requires an additional estimand-specific condition: locally, the derivative of $\Psi$ must be the sum of arm-specific linear components.
This condition is stated formally in Appendix~\ref{sec:variance_reduction} (Assumption~\ref{assumption:arm_specific_derivative_decomposition}) and verified for common estimands in Appendix~\ref{sec:estimand-regularity}.
It holds for all estimands considered in Section~\ref{section:illustration}, and for other common estimands, such as odds ratios, CDFs at fixed thresholds, and restricted mean survival times, under the regularity conditions discussed there.

\begin{theorem}[Asymptotic variance reduction]
\label{thm:variance_reduction}
Under the regularity conditions of Theorem~\ref{thm:asymp_boot} and the arm-specific derivative decomposition condition stated above,
\[
\sigma_{\mathrm{sbw}}^2
\le
\sigma_{\mathrm{unadj}}^2.
\]
\end{theorem}
The inequality is strict when the balanced covariates have any nonzero linear association with the arm-specific quantities entering the estimand. The corresponding statement for vector-valued estimands follows coordinatewise, or by replacing variances with covariance matrices. This property is especially useful when a consistent, asymptotically normal unadjusted estimator is known but no covariate-adjustment method is available or intuitive. Theorem~\ref{thm:variance_reduction} shows that replacing the unadjusted plug-in estimator with its SBW analogue weakly improves large-sample precision.

\subsection{Arm-specific plug-in coherence for subgroup analyses}
\label{subsection:plugin_coherence}

Beyond asymptotic precision, SBW has a useful structural property for clinical trial reporting: it preserves the plug-in relationship between arm-specific marginal and subgroup-specific summaries. This is relevant when investigators report a primary treatment effect estimate alongside subgroup analyses defined by baseline covariates, such as sex, age, or disease severity. In these settings, practitioners often expect the primary and subgroup analyses to be internally coherent: the subgroup summaries should be interpretable as components of the overall analysis, rather than as estimates produced by a separate and potentially incompatible procedure.

We formalize this property through a simple plug-in identity, which we refer to as \emph{arm-specific plug-in coherence}. Let $V=v(X)$ denote a discrete subgroup variable defined from baseline information, such as age category, disease stage, or geographic region. For an arm-specific distribution $P_a$, the law of total expectation implies
\[
\mathbb E_{P_a}[Y]
=
\sum_v \mathbb E_{P_a}[Y\mid V=v] P_a(V=v).
\]
Thus, within each treatment arm, the marginal mean decomposes as the subgroup-weighted average of the subgroup-specific means. The following proposition states that the same identity is preserved when the arm-specific distribution is replaced by an empirical or weighted empirical distribution. This plug-in identity does not require $V$ to be included among the SBW balancing variables. Its proof is given in Appendix~\ref{appendix:subgroup_coherence}.

\FloatBarrier

\begin{proposition}[Arm-specific plug-in coherence]
\label{prop:plugin_coherence}
Let $\widehat P_{a}^{g}$ be a (possibly weighted) empirical distribution for treatment arm $a$, and write
$
\widehat{\mathbb E}_a^g[\cdot] := \mathbb E_{\widehat P_a^g}[\cdot].
$
Then, for any discrete subgroup variable $V=v(X)$,
\[
\widehat{\mathbb E}_{a}^{g}[Y]
=
\sum_v
\widehat{\mathbb E}_{a}^{g}[Y\mid V=v]\,
\widehat P_{a}^{g}(V=v).
\]
\end{proposition}
Proposition~\ref{prop:plugin_coherence} applies directly to both the unadjusted plug-in estimator and the SBW estimator. For the unadjusted estimator, $\widehat P_{a}^{g}=P_{n,a}$, and for SBW, $\widehat P_{a}^{g}=\widehat P_{a}^{\,\mathrm{sbw}}$. Thus, for each method and treatment arm, the overall arm-specific mean can be reconstructed exactly from that arm's subgroup means and subgroup proportions. This is an arm-specific identity, not a statement that the overall treatment effect can always be reconstructed from subgroup treatment effects; we discuss this stronger aggregation property in Appendix~\ref{appendix:subgroup_coherence}.

Model-based adjusted estimators do not automatically share this property. For example, covariate adjustment via augmented inverse-probability weighting (AIPW) has increasingly been recommended as a general approach for randomized trials \citep{fda2023adjusting, bannick2023general}; such estimators can be constructed to respect subgroup decompositions when the overall and subgroup summaries are derived from a common fitted object. If instead they are fit separately, however, the resulting summaries may not arise from a single empirical, weighted empirical, or fitted distribution, and arm-specific plug-in coherence can fail.

We illustrate this property in a simulation where subgroup membership is both prognostic and
effect-modifying. We generate randomized trials with $n=500$ participants and 1000 Monte Carlo replicates. For each participant, we draw $X\sim N(0,1)$, define the subgroup indicator $V=\mathbbm 1\{X>0\}$, randomize treatment as $A\sim\mathrm{Bernoulli}(1/2)$, and generate $Y \sim N(2\sin\{2X\}
+
A\{-0.5+2.2V\}, 0.75^2).$

For each simulated dataset, we computed overall and subgroup-specific arm-specific means using the unadjusted plug-in estimator, the SBW plug-in estimator balancing $X$, and AIPW. For AIPW, the propensity score was estimated by the empirical treatment probability, and separate linear outcome regressions using $X$ were fit for the overall and subgroup analyses. Table~\ref{tab:subgroup_coherence_simulation} reports the maximum absolute arm-specific coherence gap,
\begin{equation} \label{eq:coherence-gap}
    \max_{a\in\{0,1\}} 
\left|
\widehat{\mathbb E}_a[Y]
-
\sum_v
\widehat{\mathbb E}_a[Y\mid V=v]\,
\widehat P_a(V=v)
\right|.
\end{equation}

\begin{table}[t]
\centering
\begingroup
\scriptsize
\setlength{\tabcolsep}{3.55pt}
\begin{tabular}{lcc}
\toprule
\textbf{Estimator}
&
\textbf{Median gap}
&
\textbf{95th percentile gap}
\\
\midrule
Unadjusted
&
$1\times 10^{-16}$
&
$3\times 10^{-16}$
\\
SBW
&
$8\times 10^{-16}$
&
$2\times 10^{-15}$
\\
AIPW, separate fits
&
$0.04$
&
$0.11$
\\
\bottomrule
\end{tabular}
\endgroup
\caption{Simulation illustration of arm-specific plug-in coherence. Values are summaries across 1000 Monte Carlo replicates. The arm-specific coherence gap is calculated in \eqref{eq:coherence-gap}.}
\label{tab:subgroup_coherence_simulation}
\end{table}

The arm-specific coherence gap is zero up to numerical precision for the unadjusted and SBW plug-in estimators. In contrast, the separately fit AIPW analyses yield nonzero arm-specific gaps, because the overall and subgroup summaries are not constrained to arise from a single fitted distribution. A separate treatment-effect aggregation diagnostic is considered in Appendix~\ref{appendix:subgroup_coherence}, where we show that SBW balancing $X$ reduces the median common-weight ATE aggregation gap relative to the unadjusted estimator. Stable balancing weighting therefore preserves the arm-specific plug-in structure of the unadjusted estimator, so marginal and subgroup summaries remain internally consistent. Appendix~\ref{appendix:subgroup_coherence} reports an analogous consistency property for the relative risk, where SBW's violation rate is roughly half to two-thirds that of AIPW and the unadjusted estimator.

\section{Simulations}

We use simulations to assess whether SBW, despite its simple and estimand-agnostic form, remains competitive with more tailored covariate-adjusted estimators. We consider two outcome settings: a binary-outcome simulation targeting the ATE and a time-to-event simulation targeting the survival ratio. To isolate modeling choices from data structure, both simulations use the same underlying covariate structure and vary only which transformations of the covariates are available to the estimators. This design introduces controlled misspecification without altering the true outcome-generating features.

\subsection{Average treatment effect}

We first consider randomized trials with binary outcomes. The estimand is the ATE, and power is evaluated using two-sided tests at level $\alpha = 0.05$.

In both simulations, covariates are sampled with replacement from a cleaned NHANES dataset containing ten baseline variables, including demographic and laboratory measurements. The data-generating process applies oscillatory and nonlinear transformations to seven covariates while leaving three covariates untransformed. Denoting the transformed covariates by $X_{\text{true}}$, treatment is randomized as $A\sim\mathrm{Bernoulli}(1/2)$, and outcomes $Y\in\{0,1\}$ follow
\[
P(Y=1\mid X_{\text{true}},A)=\mathrm{expit}\!\big(X_{\text{true}}^\top\beta+\tau A\big),
\]
where $\beta$ is fixed and $\tau$ controls the marginal treatment effect.

To assess sensitivity to model misspecification, we vary which covariate representation is available for adjustment. The approximately nonlinear view exposes the original, untransformed NHANES variables and is intentionally misspecified relative to the data-generating process. The moderately nonlinear view includes engineered transformations designed to resemble, but not exactly match, the data-generating process's features. The approximately linear view exposes the exact data-generating process's features $X_{\text{true}}$, so that the logistic regression is linear in the observed features. Appendix~\ref{sec:appendix-simulation-results} quantifies the linear approximation of these views using $R^2$ from regressions of the true conditional mean on treatment and the observed features.

Within each covariate view, we vary the adjustment set to study the effects of dimension and prognostic strength. Adjustment sets either include all ten covariates, three selected covariates, or a trigonometric basis expansion of three selected covariates. For the three-covariate settings, we consider low, medium, and high prognostic strength. A complete specification of the data-generating process, including coefficients, covariate views, and adjustment-set construction, is given in Appendix~\ref{sec:appendix-simulation-results}.

We compare four estimators. The unadjusted estimator is the raw difference in sample means, with standard error based on the arm-specific sample variances. The SBW estimator constructs SBWs separately within each treatment arm by targeting the overall covariate means of the chosen adjustment set. We also consider two AIPW estimators with the propensity score estimated by the empirical treatment probability: AIPW (linear), which uses a linear model for the outcome regression, and AIPW (RF), which uses cross-fitted random forests fit separately by arm for the outcome regression. For AIPW (linear) and AIPW (RF), inference uses Wald standard errors. For SBW, standard errors are obtained via a nonparametric bootstrap with $1500$ resamples. Numerical failures are recorded and summarized in Appendix~\ref{sec:appendix-simulation-results}.

Sample sizes and effect sizes are calibrated to achieve approximately 90\% power at $\alpha=0.05$ for the unadjusted estimator. The small, medium, and large sample-size settings correspond to total trial sample sizes of 288, 1134, and 7086, respectively.

\begin{figure}[t]
    \centering
    \includegraphics[width=1\linewidth]{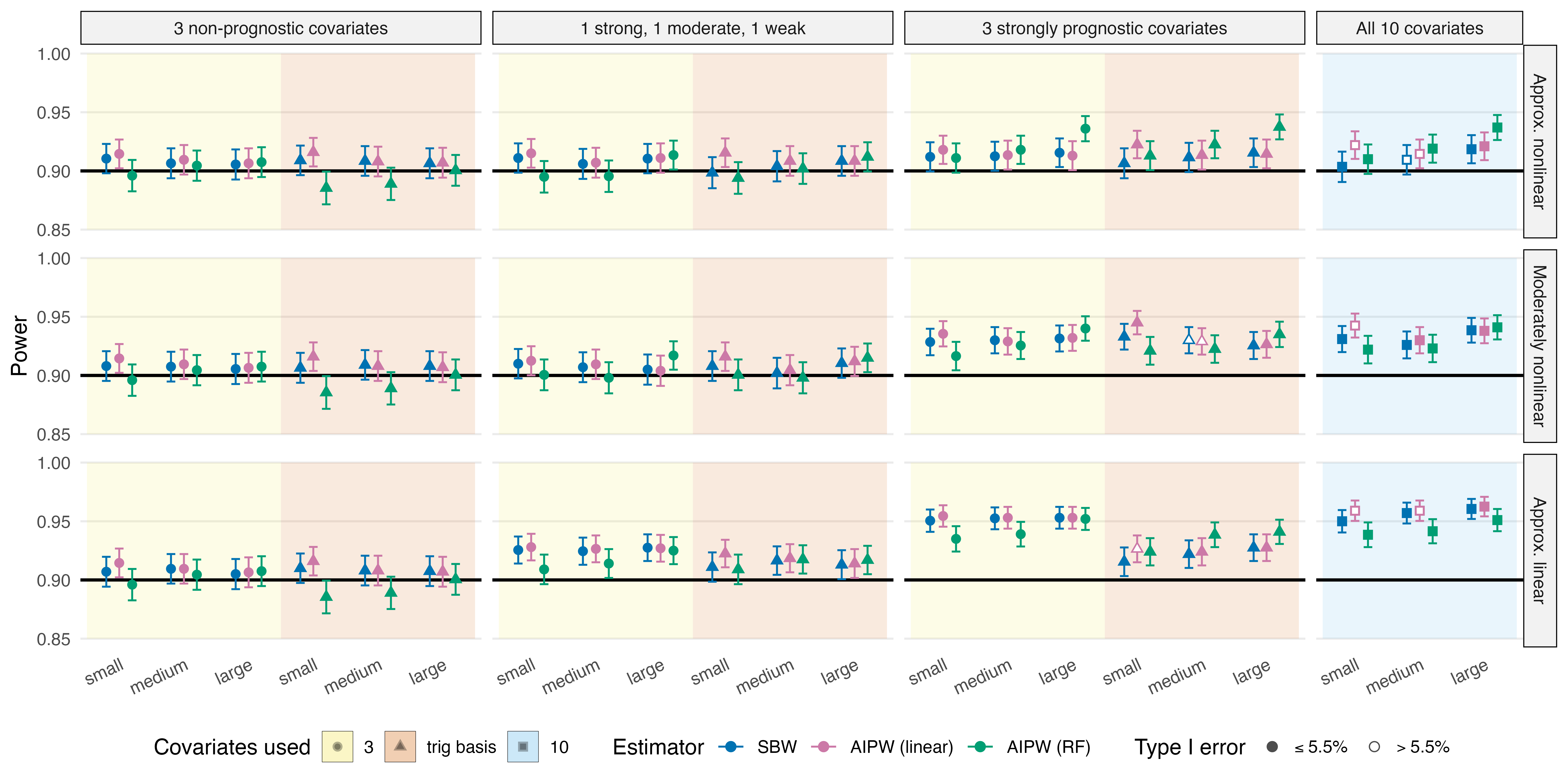}
    \caption{Power for the ATE (95\% Wald CIs) across sample sizes, covariate views, and covariate subsets.}
    \label{fig:binarysim_power}
\end{figure}

Each factorial cell, defined by sample size, covariate view, dimension, and prognostic level, is evaluated using 2000 Monte Carlo replicates. In each replicate we resample covariates, randomize treatment, generate outcomes from the fixed data-generating process, construct the chosen view, select the relevant covariate subset, and compute all four estimators. Figure~\ref{fig:binarysim_power} summarizes empirical power with 95\% Wald confidence intervals for SBW, AIPW (linear), and AIPW (RF) across views, sample sizes, and covariate subsets.

Across most settings, all three covariate-adjusted estimators achieve power near or above the 90\% design target. Power gains are largest when the available covariates are closely aligned with the outcome-generating mechanism and when the adjustment set includes strongly prognostic covariates. This pattern is most apparent in the approximately linear view and, to a lesser extent, in the moderately nonlinear view. In the approximately nonlinear view, adjustment provides smaller gains unless the available covariates are strongly prognostic.

The covariate-subset comparisons suggest that using all ten covariates performs similarly to, and often slightly better than, restricting to only the three most prognostic variables. We do not observe evidence of instability from this modestly higher-dimensional adjustment for SBW, AIPW (linear), or AIPW (RF). In contrast, the trig-basis setting does not consistently improve performance relative to using the corresponding three covariates alone, and in some settings slightly attenuates the gains from adjustment.

Type I error control is generally close to nominal across estimators. The largest upward deviations occur for AIPW (linear), and to a lesser extent SBW, most often in higher-dimensional small- and medium-sample settings. The AIPW (RF) does not exceed the 5.5\% threshold in these results, although it is occasionally slightly conservative. Consequently, power comparisons involving AIPW (linear) should be interpreted alongside the corresponding Type I error results, since some of its power gains occur in settings with mild Type I error inflation. Additional details are reported in Appendix~\ref{sec:appendix-simulation-results}.

\subsection{Survival ratio}\label{subsection:SR}

We next study covariate adjustment for time-to-event outcomes, focusing on a survival ratio estimand evaluated at a fixed analysis time $t_0$. This simulation uses the same NHANES covariate views and covariate-subset regimes as in the binary-outcome simulation, but replaces the binary outcome model with a time-to-event data-generating process.

For each replicate, event times are generated from an exponential proportional-hazards model with covariate-dependent hazard
\[
\lambda(t \mid X_{\text{true}}, A) = \lambda_0 \exp\!\big(X_{\text{true}}^\top\beta + \tau A\big),
\]
where $\lambda_0$ is fixed and $\tau$ governs the marginal treatment effect on the hazard scale. Independent censoring times are generated as $C\sim \mathrm{Exp}(\lambda_c)$, with $\lambda_c$ chosen to obtain approximately 15\% of observations censored across each sample size. We observe $\tilde T = \min(T,C)$. The estimand is the survival ratio at time $t_0$,
\[
\mathrm{SR}(t_0) = \frac{S_1(t_0)}{S_0(t_0)},
\qquad
S_a(t) = P_a\{T > t\},
\]
Neither $\mathrm{SR}(t_0)$ nor its arm-specific numerator and denominator is a linear functional of a regression; each is a Kaplan--Meier functional, which is still Hadamard differentiable under standard censoring conditions.

Inference is based on two-sided tests at level $\alpha=0.05$ using log-scale intervals. We set $t_0=5$ throughout. Calibration of $(n,\tau,\lambda_c)$ values and the corresponding true $\mathrm{SR}(t_0)$ for coverage calculations follows the same strategy as in the binary-outcome setting and is reported in Appendix~\ref{sec:appendix-simulation-results}. The small, medium, and large sample-size settings correspond to total trial sample sizes of 340, 1294, and 7100, respectively.

We compare four estimators of $\mathrm{SR}(t_0)$. The unadjusted estimator uses the Kaplan--Meier estimator within each arm and reports $\widehat S_1(t_0)/\widehat S_0(t_0)$ as a point estimate with a confidence interval based on Greenwood standard errors on the log--log scale. The SBW and IPW estimators compute weighted Kaplan--Meier curves using, respectively, SBWs and stabilized IPW based on a logistic propensity score model for $A$ given the available covariates \citep{shao2026inverse}. For both SBW and IPW, uncertainty is quantified using a nonparametric bootstrap with $1500$ resamples, with a Wald interval on the log scale. If a bootstrap replicate fails, we substitute the unadjusted estimator for that replicate. Bootstrap failures are rare in these simulations and are summarized in Appendix~\ref{sec:appendix-simulation-results}.

Finally, we include the \texttt{CFsurvival} estimator of \citet{westling2024inference}, a cross-fitted estimator of covariate-adjusted treatment-specific survival curves. Unlike the Kaplan--Meier-based unadjusted, SBW, and IPW estimators, which rely on independent censoring, \texttt{CFsurvival} allows censoring to depend on baseline covariates through nuisance estimation. In the present simulations, censoring is independent of both event times and covariates, so the corresponding estimands coincide. Inference uses the package's influence-function-based standard error to construct log-scale Wald intervals; numerical failures are recorded and summarized in Appendix~\ref{sec:appendix-simulation-results}.

\begin{figure}[t]
    \centering
    \includegraphics[width=1\linewidth]{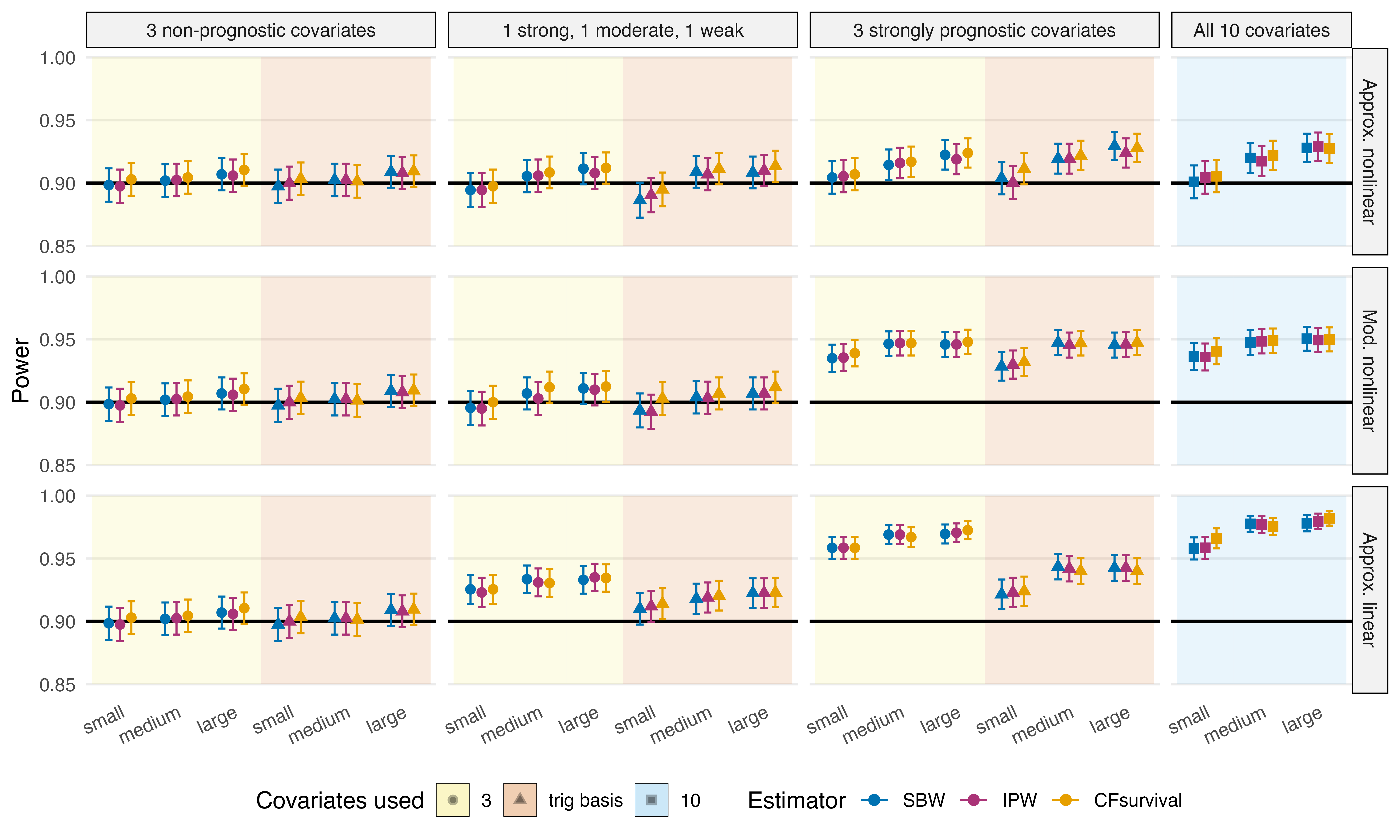}
    \caption{Power for the SR (95\% Wald CIs) across sample sizes, covariate views, and covariate subsets. No empirical Type I error simulation result exceeded 5.5\% in these settings.}
    \label{fig:survsim_power}
\end{figure}

Figure~\ref{fig:survsim_power} summarizes empirical power across views, sample sizes, and covariate subsets. Across most settings, the results are similar to those for the binary outcome: covariate adjustment is most beneficial when the available features are more closely aligned with the true outcome mechanism and when the adjusted covariates are more prognostic. The largest gains occur in the approximately linear view, especially when all ten covariates are used or when the low-dimensional adjustment set contains the most prognostic variables. Similar, though smaller, gains appear in the moderately nonlinear view under strongly prognostic adjustment. When only weakly prognostic covariates are available, adjusted and unadjusted methods perform similarly, indicating that adjustment offers limited efficiency gains in these settings but does not lead to substantial power loss.

The approximately nonlinear view remains the most challenging, with smaller gains from adjustment than in the moderately nonlinear and approximately linear views. Adjustment with strongly prognostic covariates or all ten covariates still improves power in some settings, but the gains are modest. The trig-basis setting does not consistently improve performance relative to using the corresponding three covariates alone; in the approximately linear view, it often attenuates the gains obtained from the original selected covariates, but without producing large losses relative to the unadjusted analysis. Type I error is uniformly close to nominal across all methods and settings. Overall, the survival simulations show that the simple, interpretable, and multipurpose SBW procedure in Algorithm~\ref{alg:gsbw} captures efficiency gains when prognostic covariates are available and remains competitive with more tailored survival-specific methods.

\section{Application to Antibody-Mediated Prevention trials}

We illustrate the proposed covariate-adjusted estimator using data from the Antibody-Mediated Prevention (AMP) trials, two parallel randomized, double-blind, placebo-controlled phase~2b studies evaluating the broadly neutralizing monoclonal antibody VRC01 for prevention of HIV-1 acquisition \citep{corey2021two}. The trials enrolled populations at elevated risk for HIV infection, including cisgender men and transgender persons in the Americas and Europe (HVTN~704/HPTN~085) and heterosexual women in sub-Saharan Africa (HVTN~703/HPTN~081). Our analysis includes $n=3039$ participants with complete time-to-event and treatment information. Although the primary analysis did not demonstrate statistically significant prevention efficacy against HIV-1 acquisition, the AMP trials are important proof-of-concept studies for antibody-based HIV prevention and continue to inform the design of subsequent efficacy trials.

In these trials, we consider the marginal prevention efficacy at a fixed late time point $\tau$, defined as
\[
\text{PE}(\tau) = 1 - \frac{P_1\{T \le \tau\}}{P_0\{T \le \tau\}},
\]
with $a=1$ corresponding to high-dose VRC01 and $a=0$ to placebo, and $T$ denotes the time to HIV-1 diagnosis. In words, $\text{PE}(\tau)$ is the proportional reduction in $\tau$-time risk of HIV-1 acquisition attributable to VRC01. As with the survival ratio in Section~\ref{subsection:SR}, neither $\mathrm{PE}(\tau)$ nor its arm-specific numerator and denominator, $P_1\{T\le\tau\}$ and $P_0\{T\le\tau\}$, is a linear functional of a regression. Like the odds and hazard ratios noted in the introduction, the risk ratio underlying $\mathrm{PE}(\tau)$ is a non-collapsible summary \citep{greenland1999confounding}, which complicates covariate-adjusted estimation. We take $\tau = 86$ weeks post-enrollment and pool the two trials; the published primary analyses instead report each trial separately with the two antibody dose groups pooled. Inference is carried out on the log cumulative-incidence-ratio scale. The unadjusted Kaplan-Meier estimator of $\text{PE}(\tau)$ yields $28.8\%$ with a Wald-based 95\% confidence interval of $(-4.1\%,\;51.3\%)$ and a two-sided Wald p-value of $0.079$, indicating suggestive but not significant evidence of a protective effect.

We next apply the proposed covariate-adjusted estimator, incorporating baseline covariates that are plausibly prognostic for HIV acquisition and available across the pooled trials. We adjust for continuous age at enrollment and country of enrollment using SBWs and conduct inference using a bootstrap-based Wald test. The covariate-adjusted estimator yields a similar point estimate of $29.2\%$, with a bootstrapped 95\% Wald confidence interval of $(-3.1\%,\;51.4\%)$ and a corresponding p-value of $0.072$. Relative to the unadjusted analysis, covariate adjustment yields a relative efficiency gain of $2.1\%$ and an effective sample size increase of approximately $65$ participants.
This modest gain is expected when only a small set of moderately prognostic baseline covariates is available; larger gains have been observed when adjustment uses strongly prognostic baseline variables \citep{kahan2014risks, thompson2015covariate}.

\section{Discussion}

Covariate adjustment improves efficiency in randomized trials, yet its practical adoption can be limited by methodological complexity and uncertainty about model specification. We have proposed a weighting-based adjustment approach that directly balances prognostic covariates while minimizing weight variability. The resulting estimators are simple to implement, broadly applicable across estimands, and are asymptotically at least as efficient as their unadjusted counterparts under randomization. These properties make SBWs a practical default strategy when conventional model-based adjustments are unavailable or difficult to justify.

Several directions for future work remain. First, many modern trials use more complex randomization or monitoring schemes, including covariate-adaptive randomization, group sequential designs, and adaptive trial designs. Extending the present framework to explicitly incorporate these settings would be valuable, particularly for understanding how weight construction or recalculation should interact with the randomization procedure and interim analyses.  Second, the choice of covariates and transformations used for balance deserves further study. While balancing baseline prognostic variables improves precision, determining which moments, interactions, or nonlinear transformations to include could potentially be informed by variable selection or data-adaptive procedures, as has been explored for other covariate-adjusted estimators \citep{balzer2024adaptive, liu2025coadvise}. Finally, although the original framework of \citet{zubizarreta2015stable} permits approximate balance through a user-specified imbalance tolerance, our formulation sets this tolerance to zero and therefore enforces exact covariate balance. Allowing controlled imbalance in randomized trials may reduce weight variability in settings with extreme or highly collinear covariates and could improve finite-sample performance. Characterizing this trade-off formally is an important area for future research.

\section*{Acknowledgments}

We thank Peter Gilbert for his assistance in accessing the AMP trial data and for his guidance in interpreting the results. This work was supported by the Patient-Centered Outcomes Research Institute (PCORI, ME-2024C2-39990, ME-2024C2-40180) and the National Institute on Aging (P01 AG032952); the content is solely the responsibility of the authors and does not necessarily represent the official views of the funding agency.

\bibliography{citations}

\appendix

\setcounter{equation}{0}
\renewcommand{\theequation}{S\arabic{equation}}
\setcounter{theorem}{0}
\setcounter{figure}{0}
\setcounter{table}{0}
\renewcommand{\thetheorem}{S\arabic{theorem}}
\renewcommand{\thefigure}{S\arabic{figure}}
\renewcommand{\thetable}{S\arabic{table}}
\renewcommand{\thealgorithm}{S\arabic{algorithm}}

\section*{\LARGE Appendices}

\DoToC

\section{Stable balancing weights optimization problem}\label{sec:SBW-Optimization-Problem}

We now provide the full specification of the optimization problem used to construct the SBWs within each treatment group $a\in\{0,1\}$. Our formulation follows the convex optimization framework of \citet{zubizarreta2015stable}, adapted to our notation and normalization.

We retain the notation from Section~\ref{section:algorithm}. For treatment group $a\in\{0,1\}$, let $\mathbf X_a\in\mathbb R^{N_a\times(q+1)}$ denote the covariate matrix for units in arm $a$, augmented with an intercept column, and let
$
\bar X_n := \frac1n\sum_{i=1}^n X_i
$
denote the full-sample covariate mean, including the intercept mean equal to $1$. We seek weights $w_a\in\mathbb R^{N_a}$ that reweight arm $a$ to match these full-sample means.

\subsection{Quadratic program}

The SBWs are defined as the solution to
\begin{equation}
\label{eq:sbw_primal}
\begin{aligned}
\min_{w_a \in \mathbb{R}^{N_a}}
& \quad \|w_a - \mathbf 1_a\|_2^2 \\
\text{subject to}
& \quad \mathbf{X}_a^\top w_a = N_a \bar X_n, \\
& \quad w_a \ge 0.
\end{aligned}
\end{equation}
The objective minimizes the squared $\ell_2$ distance between $w_a$ and the uniform weights $\mathbf 1_a$, subject to exact mean balance.

In \citet{zubizarreta2015stable}, covariate balance is imposed through inequality constraints of the form
\[
\left|
\frac{1}{N_a}\mathbf X_a^\top w_a-\bar X_n
\right|
\le \delta,
\]
where $\delta$ is a user-specified imbalance tolerance. In our formulation, we set $\delta=0$, thereby enforcing exact equality of weighted covariate means with the full-sample means. This yields the constraint $\mathbf X_a^\top w_a=N_a\bar X_n.$

Because $\mathbf X_a$ includes an intercept column, this condition also implies
$
\frac{1}{N_a}\sum_{i:A_i=a} w_{a,i}=1,
$
that is, the weights have mean one within each treatment group. This normalization preserves the scale of standard plug-in estimators: for example, an unweighted sample mean can be converted to its SBW-adjusted analogue by incorporating $w_{a,i}$ without introducing an additional normalization factor. Apart from this scaling and the choice $\delta=0$, the formulation is equivalent to that of \citet{zubizarreta2015stable}. The nonnegativity constraint $w_a\ge0$ ensures that the adjusted estimator remains a convex combination of observed outcomes and avoids extrapolation beyond the empirical support.

\subsection{Closed-form solution without the nonnegativity constraint}

If the nonnegativity constraint is omitted, the problem reduces to a linearly constrained quadratic program with equality constraints only. Let $\tilde w_a^{\mathrm{cf}}$ denote this equality-constrained solution, where the superscript ``cf'' denotes closed form. The tilde distinguishes the closed-form weights from the constrained quadratic program weights used in practice.

The Lagrangian for the equality-constrained problem is
\[
\mathcal L(w_a,\lambda)
=
\frac12\|w_a-\mathbf 1_a\|_2^2
-
\lambda^\top\bigl(\mathbf X_a^\top w_a-N_a\bar X_n\bigr).
\]
The first-order condition with respect to $w_a$ gives $w_a-\mathbf 1_a-\mathbf X_a\lambda=0,$ so $w_a=\mathbf 1_a+\mathbf X_a\lambda.$ Substituting this expression into the balance constraint gives $ \mathbf X_a^\top \mathbf 1_a+\mathbf X_a^\top\mathbf X_a\lambda = N_a\bar X_n, $ and therefore we have $ \lambda = -(\mathbf X_a^\top\mathbf X_a)^{-1} \bigl(\mathbf X_a^\top \mathbf 1_a-N_a\bar X_n\bigr). $ Thus the equality-constrained, closed-form solution is
\begin{equation}
\label{eq:sbw_closed_form}
\tilde w_a^{\mathrm{cf}}
=
\mathbf 1_a
-
\mathbf X_a(\mathbf X_a^\top\mathbf X_a)^{-1}
\bigl(\mathbf X_a^\top \mathbf 1_a-N_a\bar X_n\bigr),
\end{equation}
which matches equation \eqref{eqn:w} in the main text.

If the closed-form solution $\tilde w_a^{\mathrm{cf}}$ is nonnegative, then it satisfies all constraints in \eqref{eq:sbw_primal}: it satisfies exact balance by construction, and it satisfies nonnegativity by assumption. Since $\tilde w_a^{\mathrm{cf}}$ was obtained by solving the equality-constrained problem, it minimizes the objective among all weights satisfying exact balance. Therefore, no weight satisfying both exact balance and nonnegativity can have a smaller objective value. Hence $\tilde w_a^{\mathrm{cf}}$ also solves \eqref{eq:sbw_primal}. By strict convexity of the objective, this solution is unique. When $\tilde w_a^{\mathrm{cf}}$ has negative components, it is not feasible for \eqref{eq:sbw_primal}; in that case, the nonnegativity constraints affect the solution, and the weights must be computed numerically by solving the quadratic program.

In Appendix~\ref{sec:equivalence-constrained-sbw}, we show that under mild conditions, the closed-form weights are nonnegative with probability tending to one. Consequently, the estimator based on the closed-form weights coincides with the practical SBW estimator computed from \eqref{eq:sbw_primal} with probability tending to one.

\section{Common estimands compatible with SBW inference}
\label{sec:estimand-regularity}

This section is meant as a quick reference for common estimands that satisfy the smoothness requirements used by the SBW theory. We say below that an estimand falls under the SBW inference theory if the map $(P_0,P_1)\mapsto\Psi(P_0,P_1)$ satisfies the Hadamard differentiability condition required for asymptotic normality and bootstrap consistency, stated formally at the beginning of Appendix~\ref{sec:differentiability-Phi}. For the variance-reduction result, we additionally need the derivative of $\Psi$ to separate into arm-specific pieces along the directions used in the variance comparison, stated formally in Assumption~\ref{assumption:arm_specific_derivative_decomposition}.

For the examples below, this second condition is usually automatic, since many common estimands have the same structure: first compute one or more regular summaries within each treatment arm, and then combine the arm-specific summaries using a differentiable formula. In that case, perturbing $P_0$ only changes the arm-$0$ summary, perturbing $P_1$ only changes the arm-$1$ summary, and the derivative separates into the required arm-specific components.

Standard Hadamard differentiability results for arm-specific plug-in estimands apply to the local SBW directions used in the theory. The Hadamard differentiability condition is imposed on the estimand as a function of the arm-specific laws, and although the SBW estimator uses reweighted empirical distributions, the relevant local perturbations can be treated as perturbations of valid arm-specific probability laws. Under randomization, the population SBW coefficient is zero at the target law, so the population closed-form SBW weight function is identically one. By a similar continuity argument to the one used to show asymptotic inactivity of the nonnegativity constraint in Appendix~\ref{sec:equivalence-constrained-sbw}, the population closed-form weights remain strictly positive in a neighborhood of the target law.

In practice, existing bootstrap theory for the corresponding unadjusted plug-in estimator can be a useful guide. If the unadjusted estimator is known to admit a valid nonparametric bootstrap confidence interval, this often points to the Hadamard differentiability needed for SBW inference. The remaining SBW-specific details must still be checked, including compatibility with the arm-specific perturbations and, for the variance-reduction result, the derivative decomposition in Assumption~\ref{assumption:arm_specific_derivative_decomposition}.

Throughout this section, smooth transformations are understood to be evaluated away from their singularities. For example, ratios require denominators bounded away from zero, odds and logit transformations require probabilities bounded away from zero and one, standard deviations require positive variances, and linear-projection coefficients require nonsingular moment matrices. We refer to these below as the relevant nondegeneracy conditions.

\begin{enumerate}
\item \textbf{Means, moments, response probabilities, and CDF values.}
Arm-specific means, bounded moments, response probabilities, and CDF values at fixed thresholds fall directly under Lemmas~\ref{lemma:expectation} and~\ref{lemma:vector_expectation}, provided the relevant integrands or threshold indicators satisfy the bounded Hardy--Krause variation condition. Differences of these quantities, including average treatment effects, risk differences, survival probabilities with fully observed event times, and differences in CDF values at fixed thresholds, therefore fall under the SBW inference theory.

\item \textbf{Smooth contrasts and moment-based summaries.}
Ratios of means, relative risks, marginal odds ratios, log relative risks, log odds ratios, variances, standard deviations, correlations, restricted means with fully observed event times, restricted mean time lost, and number needed to treat or harm fall under the SBW theory by Lemmas~\ref{lemma:expectation} and~\ref{lemma:vector_expectation} plus the chain rule, under the relevant nondegeneracy conditions summarized in the paragraph above this list. Each of these summaries is a smooth function of the arm-specific summaries in item 1, so the chain rule preserves their additive arm-specific split.

\item \textbf{Quantiles and medians.}
Quantiles require an additional inverse-map condition. Let $q_a(u)$ be the arm-$a$ $u$th quantile. A sufficient condition is that the arm-specific CDF is continuous at $q_a(u)$ and crosses the level $u$ there: for small $\epsilon>0$, $F_a\{q_a(u)-\epsilon\}<u<F_a\{q_a(u)+\epsilon\}$. This is implied, for example, by a positive continuous density at $q_a(u)$. Under this condition, the quantile map is Hadamard differentiable \citep[Lemma~3.9.23]{van1996weak}. Differences of arm-specific quantiles follow by the chain rule; ratios follow when the denominator quantile is bounded away from zero; and the variance-reduction condition holds by the same chain-rule argument as in item 2. Quantiles should not be claimed to fall under this theory when the target distribution has a flat region, a nonunique crossing, or a boundary value at the quantile of interest.

\item \textbf{Right-censored survival estimands based on Kaplan--Meier.}
For right-censored outcomes, the input is the distribution of the observed follow-up time and event indicator. The estimands in this item require independent censoring and sufficient follow-up through the relevant fixed time range. More precisely, if $C$ denotes the censoring time, the censoring survival function $P(C\ge t)$ should be bounded away from zero over the time range being analyzed. Under independent censoring and sufficient follow-up through the fixed time range, the Kaplan--Meier estimator is obtained by composing the empirical observed-data law with product-integral and ordinary integration maps. This composition is Hadamard differentiable \citep[Sections 4 and 6]{gill2006lectures}, so the following estimands fall under the SBW inference theory:

\begin{enumerate}
\item \textbf{Survival probabilities and fixed-time contrasts.} An arm-specific survival probability $S_a(t_0)$ at a fixed time $t_0$, the survival difference $S_1(t_0)-S_0(t_0)$, the survival ratio $S_1(t_0)/S_0(t_0)$, and the risk ratio $F_1(t_0)/F_0(t_0)$, where $F_a=1-S_a$ is the arm-specific cumulative incidence, follow from differentiability of the Kaplan--Meier map and the chain rule.

\item \textbf{Restricted mean survival time (RMST).} The estimand $\mathrm{RMST}_a(\tau)=\int_0^\tau S_a(t)\,dt$ follows because integration over the fixed interval $[0,\tau]$ is a continuous linear operation. Restricted mean survival time differences and ratios then follow by the chain rule. The same reasoning applies to restricted mean event-free time when the event-free survival curve is estimated by Kaplan--Meier.

\item \textbf{Survival medians and other survival quantiles.} These summaries combine the Kaplan--Meier map with the quantile map. The same inverse-map condition applies with $S_a$ in place of $F_a$: $S_a$ must be continuous at the quantile level and cross it there, e.g. a positive local slope of $S_a$ at $1/2$ for the median.
Classical inference for median survival uses this type of inversion condition \citep{brookmeyer1982confidence}. If the median is not reached during follow-up, or if the curve is locally flat at the crossing, this appendix does not justify SBW inference for the median.
\end{enumerate}

\item \textbf{Mann--Whitney, win probability, and win ratio estimands.}
For finite ordinal outcomes, the Mann--Whitney estimand $P(Y_1>Y_0)+\tfrac12P(Y_1=Y_0)$ is a finite sum of products of arm-specific outcome probabilities, so it follows from Lemmas~\ref{lemma:expectation} and~\ref{lemma:vector_expectation} plus finite-dimensional differentiability of products. More generally, pairwise-comparison estimands can often be written as $\iint \phi(y_1,y_0)\,dP_1(y_1)\,dP_0(y_0)$ for a fixed comparison rule $\phi$. Their first-order derivatives separate into the term from perturbing $P_1$ with $P_0$ fixed and the term from perturbing $P_0$ with $P_1$ fixed, so the variance-reduction condition holds. Win probabilities have this same form, and win odds or win ratios follow by the chain rule when the corresponding denominator is bounded away from zero; see \citet{buyse2010generalized} and \citet{pocock2012win} for pairwise-comparison and win-ratio estimands. For right-censored or tied outcomes, additional survival-specific arguments are needed; see \citet{dobler2016mann}.
\end{enumerate}

\section{Additional details on illustration of the approach}
\label{appendix:illustration_details}

This appendix gives the data-generating mechanisms and inference details for the illustrative example in Section~\ref{section:illustration}. Treatment is assigned by simple randomization at a 1:1 ratio. The baseline covariates are $X_b\sim\mathrm{Bernoulli}(1/2)$ and $X_c\sim\mathrm{Uniform}(35,45)$.

For the binary outcome, $Y$ is drawn from a Bernoulli distribution with success probability
\[
p =
\operatorname{logit}^{-1}\left[
0.5 - 0.3 A + 2 X_b - A X_b + 1.5 (X_c-40) - A (X_c-40) + 0.3 (1-A) (X_c-40)^2
\right].
\]
Under this data-generating process, the outcome has a true mean of about $0.76$ in the control group and $0.62$ in the treated group.

For the ordinal outcome, we first generate
$
Z = 0.6(X_c-40) + 0.3X_b + 0.35A + \varepsilon,
$
where $\varepsilon$ follows a standard logistic distribution, and then discretize $Z$ into three ordered categories, with larger values taken to be better.

For the survival outcome, event times are drawn from an Exponential distribution with individual hazard rate
\[
\lambda =
\exp\left[
-2 - 0.5 A + 0.4 X_b - 0.3 A X_b + 0.3 (X_c - 40) + 0.2 A (X_c - 40)
\right].
\]
Censoring times are independently drawn from an Exponential$(0.05)$ distribution.

For the unadjusted estimators, Wald confidence intervals are constructed for the ATE and relative risk using standard error-based normal approximations. Relative risk inference is performed on the log scale. For the MW estimand, inference is performed on the logit scale, following recommendations for Mann--Whitney-type effect measures \citep{perme2019confidence}; standard errors are computed using the nonparametric bootstrap. For the SR at $t_0=3$, uncertainty is quantified on the log scale using a log-log Greenwood approximation based on the Kaplan--Meier estimators.

The estimated relative efficiency gain reported in Section~\ref{section:illustration} is
$
(\widehat{\mathrm{SE}}_{\mathrm{unadj}}/\widehat{\mathrm{SE}}_{\mathrm{sbw}})^2 - 1.
$
This quantity is computed on the scale used for inference: the natural scale for the ATE, the log scale for RR and SR, and the logit scale for MW. The relative efficiency gain can be interpreted as the approximate increase in sample size an unadjusted analysis would need to attain the same precision as our SBW estimator.

\section{Additional details on subgroup coherence}
\label{appendix:subgroup_coherence}

Here, we provide additional details for the discussion of arm-specific plug-in coherence in Section~\ref{subsection:plugin_coherence}. We first prove Proposition~\ref{prop:plugin_coherence}. We then briefly discuss how this property differs from stronger notions of treatment-effect aggregation across subgroups, and how balance constraints can be used to align subgroup composition across treatment arms.

\subsection{Proof of arm-specific plug-in coherence}

\begin{proof}[Proof of Proposition~\ref{prop:plugin_coherence}]
Fix treatment arm $a\in\{0,1\}$ and let $\widehat P_a$ be an empirical or weighted empirical distribution. For any subgroup level $v$ such that $\widehat P_a(V=v)>0$, the plug-in conditional mean is
\[
\widehat{\mathbb E}_a[Y\mid V=v]
=
\frac{\int y\,\mathbbm{1}\{V=v\}\,d\widehat P_a(x,y)}
{\widehat P_a(V=v)}.
\]
Multiplying both sides by $\widehat P_a(V=v)$ and summing over the possible values of $V$, we have
\[
\sum_v
\widehat{\mathbb E}_a[Y\mid V=v]\widehat P_a(V=v)
=
\sum_v
\int y\,\mathbbm{1}\{V=v\}\,d\widehat P_a(x,y).
\]
Since $V$ is discrete and its levels form a partition of the sample space,
$
\sum_v \mathbbm{1}\{V=v\}=1.
$
Therefore,
\[
\sum_v
\widehat{\mathbb E}_a[Y\mid V=v]\widehat P_a(V=v)
=
\int y
\left\{
\sum_v \mathbbm{1}\{V=v\}
\right\}
d\widehat P_a(x,y)
=
\int y\,d\widehat P_a(x,y)
=
\widehat{\mathbb E}_a[Y].
\]
\end{proof}

\subsection{Distinction from aggregation of treatment effects}
\label{appendix:strong_compatibility}

Arm-specific plug-in coherence is an accounting identity for arm-specific summaries. It should not be confused with the stronger requirement that an overall treatment effect equal a fixed weighted average of subgroup-specific treatment effects. To see the distinction, suppose for simplicity that the estimand is the ATE, $ \Psi(P_0,P_1) = \mathbb E_{P_1}[Y]-\mathbb E_{P_0}[Y], $ and let
\[
\Psi_v(P_0,P_1)
=
\mathbb E_{P_1}[Y\mid V=v]-\mathbb E_{P_0}[Y\mid V=v].
\]
Let $Q$ denote a common subgroup distribution used to aggregate the subgroup-specific treatment effects. An aggregation identity of the form
\[
\Psi(P_0,P_1)
=
\sum_v \Psi_v(P_0,P_1)Q(V=v)
\]
requires the subgroup weights used to aggregate the subgroup-specific treatment effects to align with the relevant arm-specific subgroup distributions. In general,
\[
\mathbb E_{P_a}[Y]
=
\sum_v \mathbb E_{P_a}[Y\mid V=v]P_a(V=v),
\]
so replacing the arm-specific subgroup distribution $P_a(V=v)$ with a common subgroup distribution $Q(V=v)$ introduces a remainder term unless the subgroup distributions agree.

For the additive case, the corresponding remainder can be written explicitly:
\begin{align*}
R(P_0,P_1;Q)
:=
\sum_v
&
\mathbb E_{P_1}[Y\mid V=v]
\{P_1(V=v)-Q(V=v)\}
- \mathbb E_{P_0}[Y\mid V=v]
\{P_0(V=v)-Q(V=v)\}.
\end{align*}
Then
\[
\Psi(P_0,P_1)
-
\sum_v \Psi_v(P_0,P_1)Q(V=v)
=
R(P_0,P_1;Q).
\]
Thus, exact aggregation using the subgroup distribution $Q$ holds when this remainder is zero. One sufficient condition is $P_1(V=v)=P_0(V=v)=Q(V=v)$ for each subgroup level $v$. In finite samples, this condition may fail to hold exactly under the unadjusted empirical distributions, even in a randomized trial, because randomization balances covariates only in expectation.

For the SBW estimator, the same decomposition can be applied with $P_a$ replaced by the SBW-weighted empirical distribution $\widehat P_{a}^{\,\mathrm{sbw}}$. Let $P_n$ denote the full-sample empirical distribution of $(X,A,Y)$, so that $ P_n(V=v) = \frac{1}{n}\sum_{i=1}^n \mathbbm{1}\{V_i=v\} $. If the subgroup indicators $\mathbbm{1}\{V=v\}$ are included among the balancing functions, then the SBW constraints imply
\[
\widehat P_{n,1}^{\,\mathrm{sbw}}(V=v)
=
\widehat P_{n,0}^{\,\mathrm{sbw}}(V=v)
=
P_n(V=v)
\]
for each included subgroup level $v$, provided the balance constraints are feasible. In that case, $ R( \widehat P_{n,0}^{\,\mathrm{sbw}}, \widehat P_{n,1}^{\,\mathrm{sbw}}; P_n) = 0 $. We report and discuss this aggregation gap in the next subsection's simulations because it captures a discrepancy practitioners may notice when comparing an overall treatment effect with subgroup-specific effects.

\subsection{Additional simulation diagnostics}
\label{appendix:subgroup_coherence_simulation}

We now report additional diagnostics from the simulation in Section~\ref{subsection:plugin_coherence}.  Table~\ref{tab:appendix_subgroup_coherence_simulation} reports both the arm-specific coherence gap and the common-weight ATE aggregation gap. The arm-specific coherence gap is as in \eqref{eq:coherence-gap}, and the common-weight ATE aggregation gap is
\[
\left|
\widehat\psi_{\mathrm{ATE}}
-
\sum_v \widehat\psi_{\mathrm{ATE},v} P_n(V=v)
\right|.
\]
For SBW, we report two versions: one balancing $X$ only, denoted SBW($X$), and
one balancing both $X$ and the subgroup indicator $V$, denoted SBW($X,V$). For
AIPW, we similarly report separately fit overall and subgroup analyses using
either $X$ or $(X,V)$ in the working outcome regression.

\begin{table}[t]
\centering
\begin{tabular}{lcccc}
\toprule
Estimator
&
Median arm gap
&
95th pct. arm gap
&
Median ATE gap
&
95th pct. ATE gap
\\
\midrule
Unadjusted
&
$1\times 10^{-16}$
&
$3\times 10^{-16}$
&
$0.09$
&
$0.26$
\\
SBW($X$)
&
$8\times 10^{-16}$
&
$2\times 10^{-15}$
&
$0.06$
&
$0.17$
\\
SBW($X,V$)
&
$8\times 10^{-16}$
&
$2\times 10^{-15}$
&
$3\times 10^{-11}$
&
$1\times 10^{-10}$
\\
AIPW $(X)$
&
$0.04$
&
$0.11$
&
$0.08$
&
$0.23$
\\
AIPW $(X,V)$
&
$0.05$
&
$0.14$
&
$0.004$
&
$0.02$
\\
\bottomrule
\end{tabular}
\caption{Additional simulation diagnostics for subgroup coherence. Values are summaries across 1000 Monte Carlo replicates. The arm-specific gap measures failure of the arm-specific plug-in identity. The ATE gap measures failure of a stronger common-weight aggregation identity for the additive treatment effect.}
\label{tab:appendix_subgroup_coherence_simulation}
\end{table}

Several points are worth noting. First, the arm-specific plug-in identity holds to numerical precision for both the unadjusted estimator and the SBW estimators, regardless of whether $V$ is explicitly included in the balance functions. This is the identity established in Proposition~\ref{prop:plugin_coherence}. Second, the common-weight ATE aggregation gap is a distinct diagnostic. It is not zero for the unadjusted estimator, because the realized subgroup distributions differ across treatment arms. Balancing $X$ reduces the median common-weight ATE aggregation gap, and explicitly balancing $V$ eliminates this gap up to numerical precision. Relative to the true ATE of approximately $0.60$, these median gaps correspond to about $15\%$, $9.5\%$, and $0\%$ of the true effect, respectively.

For AIPW, we observe that when the overall and subgroup analyses are fit separately using only $X$, both the arm-specific coherence gap and the common-weight ATE aggregation gap are nonzero. The median arm-specific gap is $0.04$, or about $7\%$ of the true ATE, and the 95th percentile is $0.11$, or about $19\%$ of the true ATE. Adding $V$ to the working outcome regression greatly reduces the common-weight ATE aggregation gap, from a median of $0.08$ to $0.004$. However, it does not restore arm-specific plug-in coherence: the median arm-specific gap is $0.05$, or about $8\%$ of the true ATE, and the 95th percentile is $0.14$, or about $23\%$ of the true ATE. Thus, incorporating subgroup structure into a model-based estimator can help with treatment-effect aggregation, but it does not by itself make separately fit analyses arise from a single empirical, weighted empirical, or fitted distribution. 

To make the discrepancy concrete, we examined one simulated dataset near the 90th percentile of the arm-specific gap for the separately fit AIPW analysis using $X$. In this replicate, the AIPW analysis reported
\[
\widehat{\mathbb E}_0[Y] = 0.17,
\qquad
\widehat{\mathbb E}_1[Y] = 0.70,
\]
whereas the subgroup-weighted reconstructions of these same arm-specific means were
\[
\sum_v \widehat{\mathbb E}_0[Y\mid V=v]\widehat P_0(V=v) = 0.07,
\qquad
\sum_v \widehat{\mathbb E}_1[Y\mid V=v]\widehat P_1(V=v) = 0.79.
\]
Thus, the reported arm-specific marginal means differed from their subgroup reconstructions by approximately $0.10$ and $0.09$, respectively. The same replicate had an overall AIPW ATE of $0.53$, while the common-weighted subgroup ATE was $0.80$, a discrepancy of $0.27$. For SBW($X$) in the same replicate, the arm-specific plug-in gaps were zero up to numerical precision, and its common-weight ATE aggregation gap was $0.18$. When $V$ was also included in the SBW balance functions, the common-weight ATE aggregation gap was zero up to numerical precision.

Model-based estimators can also be constructed to satisfy coherence identities, but doing so generally requires care. For example, one may fit a single joint model that includes subgroup indicators and relevant interactions, and then compute all overall and subgroup-specific summaries from that same fitted object. In such cases, model-based summaries can be made to reconcile. The distinction we emphasize is that this reconciliation is automatic for plug-in estimators. Once an empirical or weighted empirical distribution has been constructed, marginal and subgroup-specific summaries are obtained by applying the same plug-in rules to the same underlying distribution. Therefore, the arm-specific coherence identity holds by construction.

Coherence can also be assessed for treatment-effect summaries rather than arm-specific means. We investigate whether the marginal RR estimate may fall outside the convex hull of its own subgroup-specific RR estimates, $\{\widehat{RR}_{V=0},\widehat{RR}_{V=1}\}$, which we refer to below as a violation. In a simulation using the same $X$, $V$, $A$ as in the main text but with a binary outcome, with $\mathrm{logit}\,p(X,A) = -1.0+\sin(2X)+0.5V-0.25X+A(0.15+1.0V)$ and $Y\mid A,X\sim\mathrm{Bernoulli}\{p(X,A)\}$, SBW($X$)'s violation rate was consistently below both the unadjusted estimator's and AIPW($X$)'s -- the latter two were comparable to each other -- across $n=200$ to $2000$ (at $n=500$: 7\% for SBW($X$) vs.\ 13\% unadjusted and 11\% for AIPW($X$); at $n=200$: 9\% vs.\ 16\% and 17\%, respectively). Balancing $V$ restores exact coherence for SBW (0\% violations), and AIPW($X,V$) is close but not exact (1.1\% violations). The same ordering holds under variants of this data-generating process with baseline intercepts from $-1.5$ to $-1.0$.

\section{Additional results from simulations}\label{sec:appendix-simulation-results}

We describe two simulation studies: one with a binary outcome and one with a survival outcome. Both use the same NHANES-based covariates, a shared nonlinear transformation, and the same covariate views and adjustment-set constructions. They differ in their outcome-generating mechanisms, calibration procedures, estimands, and compared estimators.

\subsection{Shared simulation design}\label{sec:simulation_design}

\subsubsection{NHANES preprocessing and nonlinear covariate transformations}

Covariates are drawn from a cleaned version of the NHANES 2011--2012 dataset and consist of the ten variables listed in Table~\ref{tab:dgp_summary}. Continuous variables are standardized to mean zero and unit variance, and binary variables are left unchanged. We denote the resulting original covariates by $X_{*,1}, \dots, X_{*,10}$.

We transform $X_{*}$ into nonlinear features $X = (X_1, \dots, X_{10})$, which contribute to the outcome-generating mechanisms in both simulations. Table~\ref{tab:dgp_summary} gives the transformations and coefficients.

\begin{table}[ht!]
\centering
\begin{tabular}{lllr}
\hline
NHANES Variable & Original ($X_{*,j}$) & Transformed ($X_j$) & $\beta_j$ \\
\hline
Age & $X_{*,1}$ & $\sin(10 X_{*,1})$ & 0.70 \\
Sodium & $X_{*,2}$ & $\cos(X_{*,2}^2)$ & -0.70 \\
Glucose & $X_{*,3}$ & $\log(|X_{*,3}| + 1)$ & 0.60 \\
Creatinine & $X_{*,4}$ & $\exp(-X_{*,4})$ & 0.25 \\
Blood Urea Nitrogen & $X_{*,5}$ & $\sqrt{|X_{*,5}|}$ & -0.20 \\
Total Protein & $X_{*,6}$ & $\sin(5 X_{*,6})$ & 0.25 \\
Globulin & $X_{*,7}$ & $X_{*,7}^2$ & 0.02 \\
Female & $X_{*,8}$ & $X_{*,8}$ & 0.01 \\
Osmolality & $X_{*,9}$ & $X_{*,9}$ & -0.01 \\
Cholesterol & $X_{*,10}$ & $X_{*,10}$ & 0.00 \\
\hline
\end{tabular}
\caption{Summary of covariates, transformations, and coefficients used in both simulation studies.}
\label{tab:dgp_summary}
\end{table}

Treatment is randomized independently as $A \sim \mathrm{Bernoulli}(1/2)$. Larger values of $\abs{\beta_j}$ define more prognostic covariates. Neither simulation includes treatment-covariate interactions. 

\subsubsection{Covariate views}

To study how estimator performance depends on the degree of linearity between the outcome surface and the available covariates, we vary the covariate representation given to estimators while holding the data-generating process fixed. Each estimator uses one of three covariate views:

\begin{enumerate}
    \item \textbf{Untransformed (``Approximately nonlinear'') view}: The estimator observes the original NHANES covariates $X_*$. Since the outcome models depend on nonlinear transformations of these variables, this view induces substantial linear model misspecification.

    \item \textbf{Moderately nonlinear view}: The estimator observes engineered features that approximate, but do not exactly match, the true transformations:
    \begin{alignat*}{2}
    \tilde X_1 &= \sin(9.25 X_{*,1}), \qquad&
    \tilde X_2 &= \cos(X_{*,2}^3), \\
    \tilde X_3 &= \log(X_{*,3}^2 + 1),&
    \tilde X_4 &= \exp(-2 X_{*,4}), \\
    \tilde X_5 &= \sqrt{|X_{*,5}|},&
    \tilde X_j &= X_{*,j}, \quad j=6,\ldots,10,
    \end{alignat*}
    after which non-binary variables are standardized. This view represents intermediate linear model misspecification.

    \item \textbf{Oracle (``Approximately linear'') view}: The estimator observes the true transformed covariates $X$, corresponding to correct specification of the outcome model.
\end{enumerate}

The same covariate views are used in the survival simulation.

\subsubsection{Covariate subsets: dimension and prognostic strength}

Within each covariate view, we vary both the dimension and prognostic strength of the adjustment set. Dimension determines the number and form of covariates used, while prognostic strength determines which covariates are selected.

We consider three dimension settings:
\begin{enumerate}
\item \textbf{High}: all 10 covariates are included. Since all covariates are used, prognostic strength does not vary in this setting.

\item \textbf{Low}: three covariates are included, selected according to the prognostic-strength level described below.

\item \textbf{Trig basis}: the same three covariates as in the low-dimensional setting are used to construct trigonometric basis features. Each selected continuous covariate $x$ is first rescaled to $[-\pi,\pi]$, after which $\sin(x)$ and $\cos(x)$ are included. For the first selected covariate, we also include $\sin(2x)$ if that covariate is continuous. Binary covariates are not transformed.
\end{enumerate}

For the low and trig-basis settings, we vary prognostic strength by selecting covariates according to the magnitudes of the coefficients $\beta_j$ in Table~\ref{tab:dgp_summary}:
\begin{enumerate}
\item \textbf{Low}: the three least prognostic covariates $(X_8, X_9, X_{10})$;
\item \textbf{Medium}: one strongly, one moderately, and one weakly prognostic covariate, specifically $(X_1, X_4, X_7)$;
\item \textbf{High}: the three most prognostic covariates $(X_1, X_2, X_3)$.
\end{enumerate}

This design separates the effects of covariate representation, dimension, and prognostic strength on estimator performance.

\subsection{Binary simulation details}

Conditional on $(X, A)$, the binary outcome is generated as
\[
P(Y=1 \mid X, A)
=
\mathrm{expit}\!\big(X^\top \beta + \tau A\big),
\]
where $\mathrm{expit}(u) = (1 + e^{-u})^{-1}$. The coefficient vector $\beta$ is fixed across simulation settings, while $\tau$ controls the marginal treatment effect.

\subsubsection{Calibration of covariate views}

For the binary simulation, we quantify the linearity of each covariate view by regressing the true conditional mean $\mu(X,A)$ on the observed features and treatment. Larger $R^2$ values indicate a better linear approximation of the true conditional mean.

We compute this diagnostic using simulated datasets of size $10^7$ for each value of $\tau$. Table~\ref{tab:nonlinearity_calibration} shows that the oracle view is nearly linear, the moderately nonlinear view is intermediate, and the untransformed view induces substantial misspecification.

\begin{table}[ht!]
\centering
\begin{tabular}{lccc}
$\tau$ (total sample size) & Approx.\ nonlinear & Moderately nonlinear & Approx.\ linear \\
\hline
1.0 ($n=288$)  & 0.318 & 0.570 & 0.955 \\
0.5 ($n=1134$) & 0.239 & 0.534 & 0.967 \\
0.2 ($n=7086$) & 0.217 & 0.529 & 0.972 \\
\hline
\end{tabular} 
\caption{$R^2$ values from regressing the true conditional mean on the observed features and treatment.}
\label{tab:nonlinearity_calibration}
\end{table} 

\subsubsection{Power and sample size}

We calibrate sample sizes to achieve approximately 90\% power at $\alpha = 0.05$ for detecting the ATE. For each value of $\tau$, we simulate a dataset of size $10^7$ from the data-generating process, estimate the marginal risks $\mathbb{E}(Y \mid A=1)$ and $\mathbb{E}(Y \mid A=0)$, and use these estimates in a two-sample proportion power calculation. The resulting total sample sizes are shown in Table~\ref{tab:nonlinearity_calibration}.

\subsubsection{Estimators and inference}

We compare four estimators of the ATE, defined as the difference in marginal risks. All confidence intervals are constructed at the 95\% level.

\begin{enumerate}
    \item \textbf{Unadjusted:} The difference in sample outcome means, with standard error
    \[
    \left\{
    \frac{\widehat{\Var}_1(Y)}{N_1}
    +
    \frac{\widehat{\Var}_0(Y)}{N_0}
    \right\}^{1/2},
    \]
    where $N_a = \sum_{i=1}^n \mathbbm{1}(A_i=a)$, and $\widehat{\Var}_a(Y)$ denotes the sample variance of the outcomes among participants assigned to arm $a$.

    \item \textbf{SBW:} Stable balancing weights are computed separately within each treatment arm to match the overall covariate means. We first compute the closed-form weights in \eqref{eq:sbw_closed_form}. If any weights are negative, we instead solve the corresponding quadratic program with non-negativity constraints. The ATE is estimated as a weighted difference in means. Standard errors are obtained using the nonparametric bootstrap with $1500$ resamples. If the SBW point estimate fails in a Monte Carlo replicate, the unadjusted estimator is substituted and a failure indicator is recorded. If weight computation fails in a bootstrap replicate, the unadjusted estimator is substituted for that resample and the bootstrap failure rate is recorded.

    \item \textbf{AIPW (linear):} Outcome regressions  are fit using a linear model for $Y$ on treatment and the selected covariates. The propensity score is estimated by the empirical treatment probability,
    $
    \hat{\pi}_1 = \frac{1}{n} \sum_{i=1}^n A_i.
    $
    The ATE is estimated using the resulting AIPW score, and standard errors are computed from its empirical influence function.

    \item \textbf{AIPW (RF):} Outcome regressions are estimated using random forests fit separately within each treatment arm. We use $2$-fold cross-fitting: models are trained on one fold and evaluated on the other, and predictions are combined across folds. Random forests are fit using the \texttt{ranger} package with 1000 trees and minimum node size 5. The ATE is estimated by substituting the cross-fitted outcome predictions and the same constant propensity estimate $\hat{\pi}_1$ as above into the AIPW score. Standard errors are computed from the empirical influence function.
\end{enumerate}

\subsubsection{Implementation details}

Each simulation scenario is defined by a combination of sample size, covariate view, dimension, and prognostic level. For each scenario, we generate 2000 Monte Carlo replicates. Within each replicate, covariates are resampled from NHANES, treatment is assigned at random, outcomes are generated from the logistic model, and all estimators are computed.

\subsection{Survival simulation details}

\subsubsection{Data-generating mechanism}

Event times are generated from a proportional hazards model with exponential baseline hazard. Conditional on $(X,A)$,
\[
\lambda(t \mid X,A)
=
\lambda_0 \exp\!\big(X^\top \beta + \tau A\big),
\]
where $\lambda_0=0.1$ and $\beta$ is the shared coefficient vector from Table~\ref{tab:dgp_summary}. Negative values of $\tau$ correspond to beneficial treatment effects.

Independent censoring times are generated as $ C \sim \mathrm{Exponential}(\lambda_c). $ The observed time is $\tilde T = \min(T,C)$. We estimate the SR $ S_1(t_0)/S_0(t_0) $ at $t_0=5$.

\begin{table}[ht!]
\centering
\begin{tabular}{lcccc}
$\tau$ (total sample size) & $\lambda_c$ & Cens.\ rate & Power & $S_1(t_0)/S_0(t_0)$ \\
\hline
$-0.831$ ($n=340$)  & 0.00718 & 0.150 & 0.899 & 1.302 \\
$-0.414$ ($n=1294$) & 0.00924 & 0.151 & 0.900 & 1.160 \\
$-0.173$ ($n=7100$) & 0.01050 & 0.150 & 0.900 & 1.068 \\
\hline
\end{tabular}
\caption{Calibrated survival simulation parameters. Censoring rates and power are estimated from the final calibration check; the survival ratio is the Monte Carlo truth at $t_0=5$.}
\label{tab:surv_calibration}
\end{table}

\subsubsection{Calibration of treatment effect and censoring}

For each sample size, we calibrated the censoring rate $\lambda_c$ and treatment-effect parameter $\tau$ to target approximately 15\% censoring and 90\% power at $\alpha=0.05$. Power was evaluated using the unadjusted Kaplan--Meier estimator for $\log{S_1(t_0)/S_0(t_0)}$, with rejection based on whether the Wald interval excluded zero. We selected $\lambda_c$ to yield approximately the target censoring rate, calibrated $\tau$ conditional on $\lambda_c$, and then checked both power and censoring by simulation. For each calibrated setting, the true values of $S_a(t_0)$ were computed by Monte Carlo integration over the covariate distribution, averaging $\exp{-\lambda_0\exp(X^\top\beta+\tau a)t_0}$ separately for $a=0$ and $a=1$. The resulting survival ratio was used as the truth for evaluating bias and coverage.

\subsubsection{Estimators and inference}

We compare four estimators of the survival ratio $S_1(t_0)/S_0(t_0)$.

\begin{enumerate}
    \item \textbf{Unadjusted:} Kaplan--Meier curves are estimated separately by treatment arm. Inference is performed on the log-ratio scale, $\log \hat S_1(t_0) - \log \hat S_0(t_0)$,
    using a delta-method standard error based on Greenwood's formula for the Kaplan--Meier estimates, implemented through the log--log transformation \citep[Section~1.4]{kalbfleisch2002statistical}.

    \item \textbf{SBW:} Stable balancing weights are computed as in the binary simulation and applied within a weighted Kaplan--Meier estimator. Standard errors are obtained using the nonparametric bootstrap with $1500$ resamples. If weight computation or survival estimation fails in a bootstrap replicate, the unadjusted estimator is substituted and the failure rate is recorded.

    \item \textbf{IPW:} Stabilized inverse probability weights are estimated from a logistic propensity score model and applied within a weighted Kaplan--Meier estimator. Inference is performed using the same bootstrap procedure as for SBW.

    \item \textbf{CFsurvival:} We use \texttt{CFsurvival} to estimate the survival ratio with a cross-fitted, influence-function-based estimator. Event and censoring nuisance functions are estimated using Kaplan--Meier, Cox, and Weibull working models, with $5$-fold cross-fitting. Confidence intervals are constructed on the log scale using the estimated influence function.

    In rare cases, \texttt{CFsurvival} produced numerical errors. These replicates were treated as failures, replaced by the unadjusted estimator, and recorded.
\end{enumerate}

\subsubsection{Implementation details}

Each simulation scenario is defined by sample size, covariate view, dimension, and prognostic level, as described in Section~\ref{sec:simulation_design}. For each scenario, we generate 2000 Monte Carlo replicates. Within each replicate, covariates are resampled from NHANES, event and censoring times are generated, and all estimators are computed. Bootstrap inference for SBW and IPW is performed separately within each replicate.

\subsection{Additional binary simulation results}

We present additional diagnostics for the binary outcome simulations, focusing on Type I error, confidence interval coverage, bias, and numerical stability.

Empirical Type I error rates are shown in Figure~\ref{fig:binary_t1e}. Overall, all estimators remain close to the nominal $5\%$ level. The largest upward deviations are observed for AIPW (linear), and to a lesser extent SBW, most often in higher-dimensional settings. These deviations are modest and generally around $5.5\%$--$6.0\%$.

\begin{figure}
    \centering
    \includegraphics[width=1\linewidth]{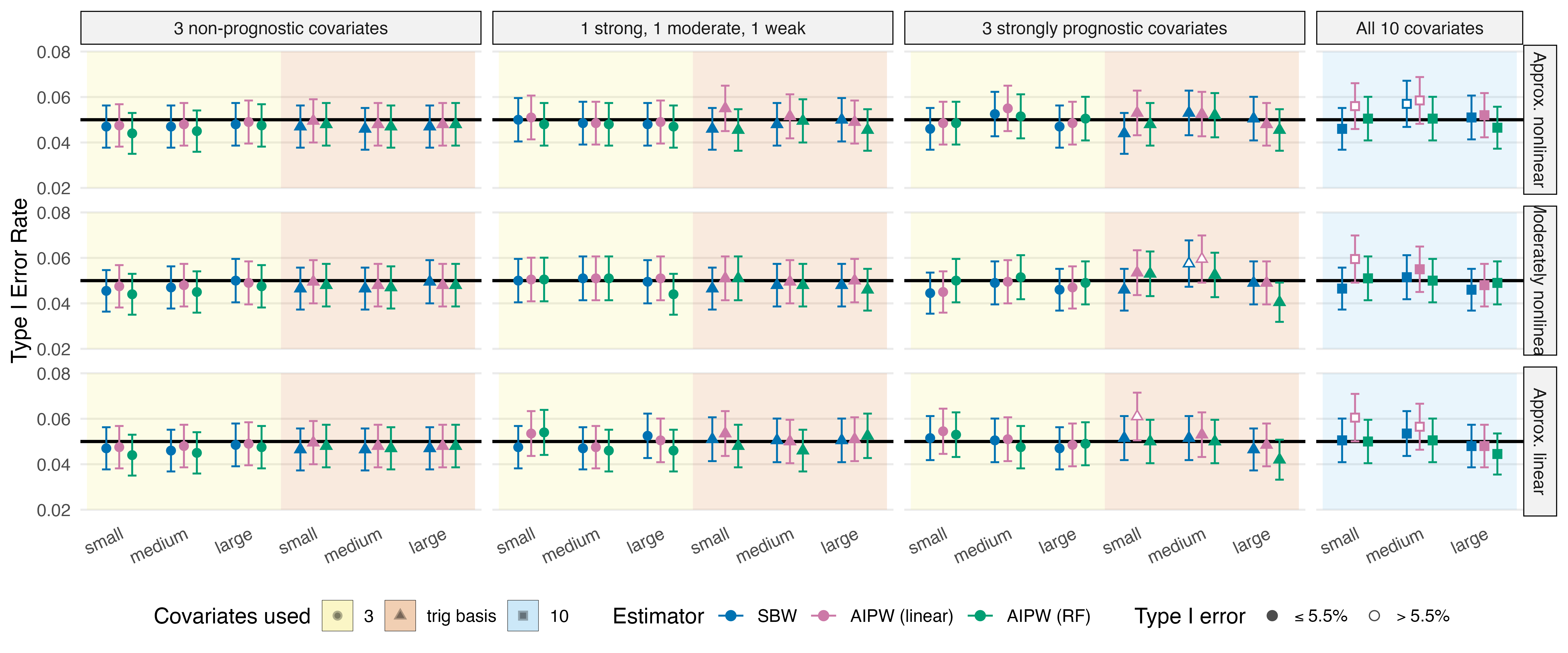}
    \caption{Empirical Type I error for the ATE across sample sizes, covariate views, and covariate subsets. The horizontal line marks the nominal 5\% level.}
    \label{fig:binary_t1e}
\end{figure}

Empirical $95\%$ confidence interval coverage in the power simulations is shown in Figure~\ref{fig:binary_coverage}. Coverage is close to nominal across most settings, with occasional mild undercoverage for AIPW (linear) and SBW in smaller-sample or more complex adjustment settings.

\begin{figure}
    \centering
    \includegraphics[width=1\linewidth]{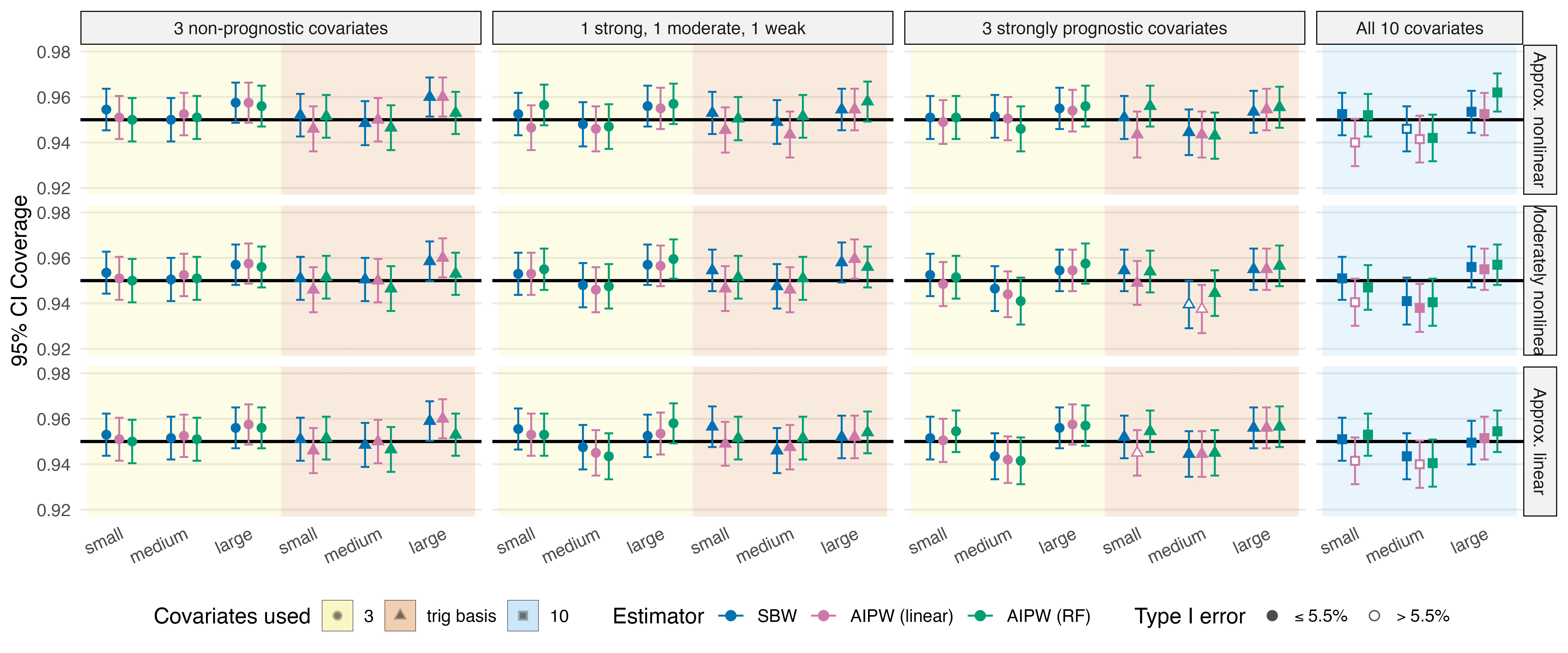}
    \caption{Empirical 95\% CI coverage for the ATE across sample sizes, covariate views, and covariate subsets. The horizontal line marks nominal coverage.}
    \label{fig:binary_coverage}
\end{figure}

Bias is negligible across the binary simulations: empirical mean bias is on the order of $10^{-4}$, with no meaningful pattern across covariate views, sample sizes, or adjustment sets.

The SBW estimator is computationally stable in these simulations. The point estimate failure rate is zero across all Monte Carlo replications, and at least one bootstrap failure occurs in only $0.094\%$ of Monte Carlo replicates (236 out of 252{,}000). When failures do occur, they arise from the quadratic program (\texttt{solve.QP} in the \texttt{quadprog} package) reporting that the constraints are inconsistent, implying that no feasible weight vector satisfies the exact balance conditions. This is likely due to near-collinearity or extreme covariate values shrinking the feasible region, and is observed only in the smallest sample size, primarily under the trig-basis view.

\subsection{Additional survival simulation results}

We present additional diagnostics for the survival outcome simulations. Empirical Type I error rates are shown in Figure~\ref{fig:surv_t1e}. All estimators maintain Type I error close to the nominal $5\%$ level. No cell exceeds $5.5\%$, and the largest value is approximately $5.5\%$.

\begin{figure}
\centering
\includegraphics[width=1\linewidth]{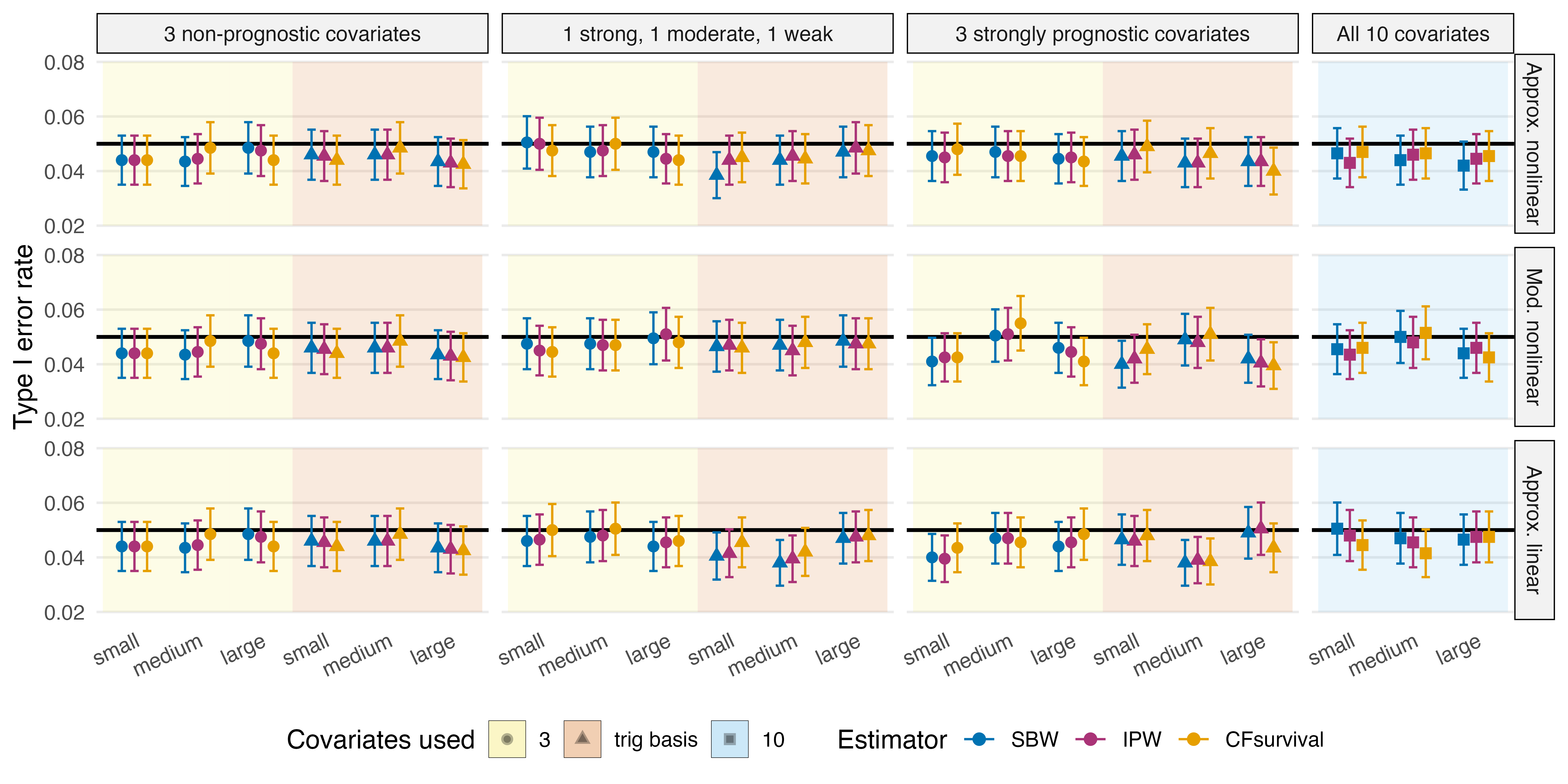}
\caption{Empirical Type I error for the SR across sample sizes, covariate views, and covariate subsets. The horizontal line marks the nominal 5\% level. No cell exceeded 5.5\%.}
\label{fig:surv_t1e}
\end{figure}

Empirical $95\%$ confidence interval coverage is shown in Figure~\ref{fig:surv_coverage}. Coverage remains close to nominal across estimators and settings, with only minor deviations.

\begin{figure}
\centering
\includegraphics[width=1\linewidth]{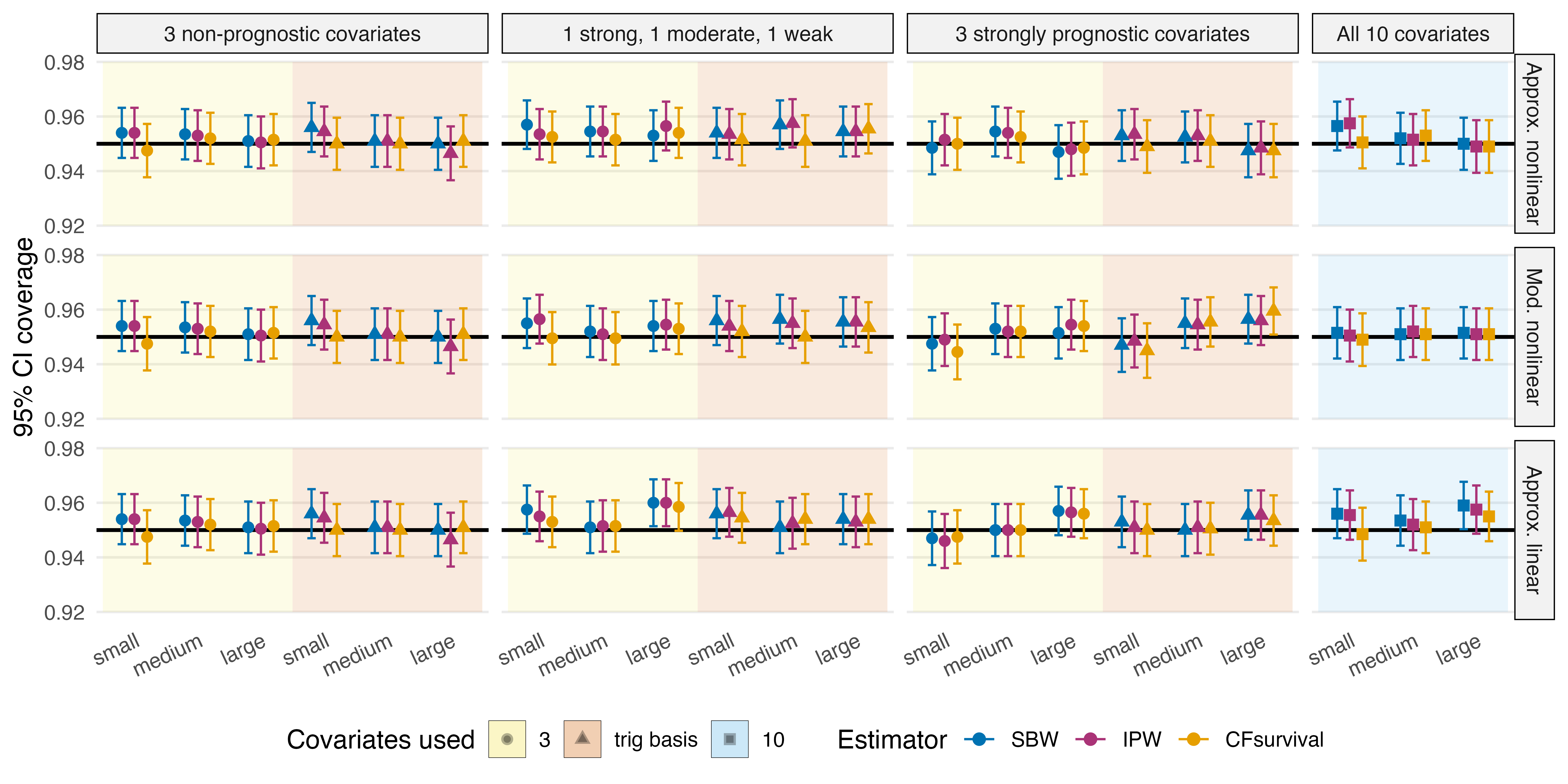}
\caption{Empirical 95\% CI coverage for the SR across sample sizes, covariate views, and covariate subsets. The horizontal line marks nominal coverage.}
\label{fig:surv_coverage}
\end{figure}

Mean signed bias on the original SR scale was small across the survival simulations and decreased with sample size. In the smallest sample-size setting, CFsurvival showed the largest finite-sample bias, with mean signed bias approaching 0.02 in some settings, corresponding to roughly 1.5\% of the true estimand value reported in Table~\ref{tab:surv_calibration}. This bias was systematically positive across covariate views and adjustment sets. In one cell, a single finite but extreme CFsurvival estimate inflates the mean signed bias to 0.34; that replicate is not flagged as a numerical failure, since the estimate is finite. Both SBW and IPW were closer to zero, with mean signed bias below 0.01 in the smallest sample size. By the medium and large sample sizes, bias was negligible for all estimators.

The SBW estimator has no observed Monte Carlo failures across scenarios, and bootstrap failures are rare: at least one bootstrap failure occurs in $0.034\%$ of Monte Carlo replicates (86 out of 252{,}000), all in the smallest sample size and almost all under the trig-basis view. CFsurvival occasionally produces non-finite estimates, in which case the estimate is replaced with the unadjusted estimator and the failure is recorded. These events are rare (29 of 252{,}000 replications, 0.01\%) and occur primarily in the smallest sample size (23 of 29 failures).

\section{Theoretical results}\label{section:theoretical_results}

\subsection{Hadamard differentiability of the SBW functional}

\subsubsection{Overview}

Below, we study the large-sample behavior of the SBW estimator by viewing it as a plug-in estimator. The central object is a functional $\Phi$, constructed so that its empirical plug-in value coincides with the SBW estimator. The main task is to show that $\Phi$ is Hadamard differentiable at $P$, in the sense described below. This allows us to apply the functional delta method to show asymptotic normality and bootstrap validity \citep[Theorems~20.8,~23.9]{van2000asymptotic}. Under an additional condition, the same derivative also yields an asymptotically linear representation and an influence function.

We use the subscript $n$ in two standard ways: to denote sample size in empirical quantities, such as $P_n$, and to index generic sequences in differentiability arguments, such as $t_n$ and $h_n$. The relevant meaning is determined by the object being indexed.

Following \citet[Section~20.2]{van2000asymptotic}, we use the following notion of Hadamard differentiability, with an explicit convention for the domain of the derivative when the tangential set is not itself a linear space. Let $\phi:\mathbb D_\phi\to\mathbb E$ be a map defined on a subset $\mathbb D_\phi$ of a normed space $\mathbb D$, taking values in another normed space $\mathbb E$, and let $\theta\in\mathbb D_\phi$. Let $\mathbb D_0\subseteq\mathbb D$ denote the set of allowable tangential directions. We say that $\phi$ is Hadamard differentiable at $\theta$ tangentially to $\mathbb D_0$ if there exists a continuous linear map $\dot\phi_\theta:\mathcal L_0\to\mathbb E$, with $ \mathcal L_0 := \overline{\operatorname{span}(\mathbb D_0)} \subseteq \mathbb D, $ where the closure is taken in the norm of $\mathbb D$, such that the following condition holds: for every sequence $t_n\downarrow0$ and every sequence $h_n\in\mathbb D$ with $h_n\to h\in\mathbb D_0$ and $\theta+t_nh_n\in\mathbb D_\phi$ for all $n$,
\[
\left\|
\frac{\phi(\theta+t_nh_n)-\phi(\theta)}{t_n}
-
\dot\phi_\theta[h]
\right\|_{\mathbb E}
\to0.
\]

Thus $\mathbb D_0$ determines the directions along which differentiability is checked, while $\mathcal L_0$ is the closed linear space on which we define the derivative. If $\mathbb D_0$ is already a closed linear subspace of $\mathbb D$, then $\mathcal L_0=\mathbb D_0$. In the arguments below, we sometimes obtain the derivative formula first for allowable directions $h\in\mathbb D_0$, or for finite linear combinations of such directions. If the resulting map extends continuously and linearly to $\mathcal L_0$, that extension is the derivative used in the definition above.

Let $\mathcal Z:=\mathcal X\times\{0,1\}\times\mathcal Y$ denote the support of $Z=(X,A,Y)$, where $\mathcal X\subset\mathbb R^{q+1}$ is bounded and includes the intercept coordinate. Let $\mathcal P$ denote the class of probability laws on $\mathcal Z$ considered below. We assume throughout this section that $\mathcal P$ contains the population law, the empirical laws, and the point masses $\delta_z$ for $z\in\mathcal Z$. Here $\mathcal Y$ may be a subset of Euclidean space or a finite outcome space, including ordinal or unordered categorical outcomes. For real-valued or ordinal outcomes, inequalities such as $x\le s$ and $y\le t$ are interpreted coordinatewise, with the usual ordinal ordering used for ordinal $Y$. For unordered categorical outcomes, the $Y$-threshold indicators can be replaced by categorical indicators; references below to CDFs should then be read as referring to the corresponding indicator-indexed probability maps. The differentiability arguments are analogous.

Define the joint indicator class
\begin{equation}
\label{eq:V-class}
\mathcal V
:=
\left\{
v_{s,a,t}:(x,a',y)\mapsto
\mathbf 1\{x\le s,\ a'\le a,\ y\le t\}:
(s,a,t)\in \mathcal X \times\{0,1\}\times\mathcal Y
\right\}. \nonumber
\end{equation}
For any distribution $P$ of $Z$, the map
\[
    (s,a,t)\mapsto Pv_{s,a,t}
    =
    P(X\le s,A\le a,Y\le t)
\]
is the joint cumulative distribution function of $(X,A,Y)$, or the corresponding indicator-indexed probability map when $Y$ is unordered categorical.

We identify each $P\in\mathcal P$ with its evaluation map $v\mapsto Pv$ on $\mathcal V$. Define the embedded domain
\begin{equation}
\label{eq:D-Phi-domain}
\mathbb D_\Phi
:=
\{v\mapsto Pv:P\in\mathcal P\}
\subseteq \ell^\infty(\mathcal V). \nonumber
\end{equation}

We also define the arm-specific indicator class
\begin{equation}
\label{eq:Vprime-class}
\mathcal V_{XY}
:=
\left\{
v_{s,t}:(x,y)\mapsto \mathbf 1\{x\le s,\ y\le t\}:
(s,t)\in \mathcal X \times\mathcal Y
\right\}, \nonumber
\end{equation}
and the corresponding arm-specific embedded domain
\begin{equation}
\label{eq:D-Vprime-domain}
\mathbb D_{\mathcal V_{XY}}
:=
\{v\mapsto P' v:P' \text{ is a probability law on }\mathcal X\times\mathcal Y\}
\subseteq \ell^\infty(\mathcal V_{XY}). \nonumber
\end{equation}

The construction of the SBW functional proceeds by composing simpler maps. First, define
\begin{equation}
\label{eq:lambda}
\lambda:
\mathbb D_\Phi
\to
\mathbb D_{\mathcal V_{XY}}\times\mathbb D_{\mathcal V_{XY}}\times(0,1),
\qquad
\lambda(P):=(P_0,P_1,\pi_1(P)),
\end{equation}
where $P_a$ denotes the conditional law of $(X,Y)\mid A=a$, and $\pi_1(P)=P(A=1)$. For $a\in\{0,1\}$, define the coordinate projection
\begin{equation}
\label{eq:projection}
\mathrm{pr}_a:
\ell^\infty(\mathcal V_{XY})\times\ell^\infty(\mathcal V_{XY})\times\mathbb R
\to
\ell^\infty(\mathcal V_{XY}),
\qquad
\mathrm{pr}_a(Q_0,Q_1,r):=Q_a.
\end{equation}
When restricted to $\mathbb D_{\mathcal V_{XY}}\times\mathbb D_{\mathcal V_{XY}}\times(0,1),$ this map sends $(P_0,P_1,\pi)$ to the arm-specific law $P_a\in\mathbb D_{\mathcal V_{XY}}$.

Next, when $\mathbb E_{P_a}[XX^\top]$ is invertible, define the population SBW coefficient in arm $a$ by
\begin{equation}
\label{eq:gamma}
\gamma_a(P_0,P_1,\pi)
=
\mathbb E_{P_a}[XX^\top]^{-1}
\big[
\mathbb E_{P_a}[X]
-
\{(1-\pi)\mathbb E_{P_0}[X]+\pi\mathbb E_{P_1}[X]\}
\big].
\end{equation}
When $\mathbb E_{P_a}[XX^\top]$ is not invertible, set $\gamma_a(P_0,P_1,\pi)=0$, so that the weighting map leaves $P_a$ unchanged. Equivalently, when $(P_0,P_1,\pi)=\lambda(P)$,
\[
\gamma_a\{\lambda(P)\}
=
\mathbb E_{P_a}[XX^\top]^{-1}
\left\{
\mathbb E_{P_a}[X]-\mathbb E_P[X]
\right\}.
\]
Thus $\gamma_a\{\lambda(P)\}$ is the population analogue of the closed-form balancing coefficient.

Given an arm-specific law $P_a$ and a coefficient vector $g\in\mathbb R^{q+1}$, define the weighted distribution map
\begin{equation}
\label{eq:f}
f:
\mathbb D_{\mathcal V_{XY}}\times\mathbb R^{q+1}
\to
\ell^\infty(\mathcal V_{XY}),
\qquad
f(P_a,g)(v)
=
\int (1-g^\top x)v(x,y)\,P_a(dx,dy),
\quad v\in\mathcal V_{XY}.
\end{equation}
The codomain is $\ell^\infty(\mathcal V_{XY})$, rather than $\mathbb D_{\mathcal V_{XY}}$, because the weighted law may not be a probability law for arbitrary $g$.

We then define the arm-specific weighted distribution functional
\begin{equation}
\label{eq:eta-map}
\eta_a:
\mathbb D_\Phi\to \ell^\infty(\mathcal V_{XY}),
\qquad
\eta_a(P)
:=
f\bigl(\mathrm{pr}_a\{\lambda(P)\},
        \gamma_a\{\lambda(P)\}\bigr).
\end{equation}
With $d$ denoting the dimension of the estimand, the preceding maps determine a unique SBW functional on $\mathbb D_\Phi$. Define
\begin{equation}
\label{eq:Phi-composition}
\Phi:
\mathbb D_\Phi\to\mathbb R^d,
\qquad
\Phi(P)
=
\Psi\bigl(\eta_0(P),\eta_1(P)\bigr).
\end{equation}
This construction has two important consequences. First, on the randomized trial model
\[
\mathcal M:=\{P\in\mathcal P:A\independent X\text{ under sampling from }P\},
\]
the functional agrees with the target parameter:
\begin{equation}
\label{eq:Phi-target-on-model}
    \Phi(P)=\Psi(P_0,P_1),
    \qquad P\in\mathcal M.
\end{equation}
Indeed, if $P\in\mathcal M$, then $\mathbb E_{P_a}[X]=\mathbb E_P[X]$ for $a\in\{0,1\}$, so $\gamma_a\{\lambda(P)\}=0$. Hence $f(P_a,\gamma_a\{\lambda(P)\})=P_a$, and therefore $\eta_a(P)=P_a$. Second, when evaluated at the empirical distribution $P_n$, the same
construction gives the closed-form SBW estimator $\Phi(P_n)=\tilde\psi_n^{\mathrm{\,cf}}.$ Here $\tilde\psi_n^{\mathrm{\,cf}}$ denotes the estimator obtained from the closed-form SBWs, i.e., the solution to the equality-constrained problem without imposing nonnegativity. The practical SBW estimator, denoted $\widehat{\psi}^{\,\mathrm{sbw}}$ in Algorithm~\ref{alg:gsbw} of the main text, is computed from the constrained quadratic program \eqref{eq:sbw_primal}. In Appendix~\ref{sec:equivalence-constrained-sbw}, we show that these two estimators coincide with probability tending to one. Thus $\Phi$ is statistically aligned with the target estimand and asymptotically aligned with the estimator computed in practice.

The proof of Hadamard differentiability proceeds by establishing differentiability of the component maps $\lambda$, $\mathrm{pr}_a$, $\gamma_a$, $f$, and $\Psi$, and then applying the chain rule for Hadamard differentiability. Combining these component derivatives yields the derivative of $\Phi$.

\subsubsection{Expectation functionals}\label{sec:expectation_functionals}

We begin by establishing Hadamard differentiability of expectation-type functionals, which underpin the maps $\lambda$, $\gamma_a$, and $f$ introduced in the previous section. These results allow us to handle integrals of the form $\int g_v\,dP$ uniformly over $v\in\mathcal V_{XY}$, which will be used repeatedly in the differentiability analysis of $\Phi$. Lemma~\ref{lemma:expectation} treats the scalar-valued case, while Lemma~\ref{lemma:vector_expectation} extends the result to vector-valued integrands by a coordinatewise argument.

Let the observed data vector $Z$ have dimension $k$. Since the support of $Z$ is bounded, we may regard the corresponding distribution functions as defined on a compact rectangle $\mathcal Z$, after rescaling coordinates if needed. Let $\mathbb D_S$ denote the multivariate Skorokhod space of functions $F:\mathcal Z\to\mathbb R$ that are right-continuous with left limits (càdlàg) in each coordinate, equipped with the sup norm
\[
\|F\|_\infty := \sup_{z\in\mathcal Z} |F(z)|.
\]
Let $\mathbb D_\varphi:=\{(s,a,t)\mapsto Pv_{s,a,t} : P\in\mathcal{P}\}\subset\mathbb D_S$ be the set of multivariate distribution functions on $\mathcal Z$. For $F\in\mathbb D_\varphi$, define the tangent cone
\begin{equation}
\label{eq:Dphi-DF}
\mathbb D_F
:=
\{\alpha(F_1-F):F_1\in\mathbb D_\varphi,\ \alpha>0\}
\subset\mathbb D_S. 
\end{equation}

The set $\mathbb D_F$ captures the allowable tangential directions, but it is not generally closed under arbitrary linear combinations. To make the derivative a linear map, we enlarge $\mathbb D_F$ to the closed linear subspace generated by $\mathbb D_F$: that is, we take the span of $\mathbb D_F$ and then close it in the sup norm. Define
\begin{equation}
\label{eq:LF}
\mathcal L_F
:=
\overline{\operatorname{span}(\mathbb D_F)}
\subseteq \mathbb D_S.
\end{equation}
Thus $\mathbb D_F$ is the set of allowable tangential directions, while
$\mathcal L_F$ is the closed linear space on which the derivative will
be defined.

For each $v=v_{s,t}\in\mathcal V_{XY}$, let $g_v:\mathcal Z\to\mathbb R$ be a càdlàg function of bounded Hardy--Krause variation on $\mathcal Z$, denoted $\|\cdot\|_{HK}$ (see \cite{owen2005multidimensional}; equivalent to the sectional variation of \cite{gill1993multivariate}). In the lemma below, we let $\varphi$ denote the functional
\begin{equation}
\label{eq:expectation-functional}
\varphi:\mathbb D_\varphi \to \ell^\infty(\mathcal V_{XY}),
\qquad
\varphi(F)(v) := \int_{\mathcal Z} g_v(z)\, dF(z),
\qquad v\in\mathcal V_{XY}.
\end{equation}

\begin{lemma}[Hadamard differentiability of expectation functionals]
\label{lemma:expectation}
Fix $F\in\mathbb D_\varphi$. If \\ $\sup_{v\in\mathcal V_{XY}} \|g_v\|_{HK} < \infty$, then the map $\varphi$ in \eqref{eq:expectation-functional} is Hadamard differentiable at $F$ tangentially to $\mathbb D_F$. Its derivative is
the continuous linear map $\dot\varphi_F:\mathcal L_F\to \ell^\infty(\mathcal V_{XY}).$ For $h\in\mathcal L_F$, this derivative is given by 
\begin{equation}
    \dot\varphi_F[h](v)
    =
    \int_{\mathcal Z}g_v(z)\,dh(z),
    \qquad v\in\mathcal V_{XY}. \label{eq:varphi-dot}
\end{equation}
\end{lemma}
\begin{proof}

Let $\epsilon_n\downarrow0$ and $h_n\to h$ in $\|\cdot\|_\infty$, with $h_n,h\in\mathbb D_F$ and $F+\epsilon_n h_n\in\mathbb D_\varphi$. We verify the Hadamard differentiability conditions for the candidate derivative $\dot\varphi_F$ in \eqref{eq:varphi-dot}. By linearity of the Lebesgue--Stieltjes integral in the integrator,
\[
\frac{\varphi(F+\epsilon_n h_n)-\varphi(F)}{\epsilon_n}(v)
=
\int_{\mathcal Z} g_v\,dh_n.
\]
Therefore,
\[
\left\|
\frac{\varphi(F+\epsilon_n h_n)-\varphi(F)}{\epsilon_n}
-
\dot\varphi_F[h]
\right\|_{\ell^\infty(\mathcal V_{XY})}
=
\sup_{v\in\mathcal V_{XY}}
\left|
\int_{\mathcal Z} g_v\,d(h_n-h)
\right|.
\]
By the multivariate Lebesgue--Stieltjes integration-by-parts bound in Eq.~19 of \citet{gill1993multivariate}, there exists a constant $c_k$, depending only on the dimension of $\mathcal Z$, such that
\[
\left|
\int_{\mathcal Z} g_v\,d(h_n-h)
\right|
\le
c_k \|g_v\|_{HK}\|h_n-h\|_\infty .
\]
Taking suprema over $v\in\mathcal V_{XY}$ gives
\[
\sup_{v\in\mathcal V_{XY}}
\left|
\int_{\mathcal Z} g_v\,d(h_n-h)
\right|
\le
c_k
\left(\sup_{v\in\mathcal V_{XY}}\|g_v\|_{HK}\right)
\|h_n-h\|_\infty
\to0.
\]
Thus $\varphi$ is Hadamard differentiable at $F$ tangentially to
$\mathbb D_F$.

It remains to specify the linear domain of the derivative. We extend the displayed formula from $\mathbb D_F$ to $\operatorname{span}(\mathbb D_F)$ by linearity of the Lebesgue--Stieltjes integral in the integrator. For continuity, the same integration-by-parts bound gives, for any $h_1,h_2\in\operatorname{span}(\mathbb D_F)$,
\[
\|\dot\varphi_F[h_1]-\dot\varphi_F[h_2]\|_{\ell^\infty(\mathcal V_{XY})}
\le
c_k
\left(\sup_{v\in\mathcal V_{XY}}\|g_v\|_{HK}\right)
\|h_1-h_2\|_\infty .
\]
Hence $\dot\varphi_F$ is bounded and linear on $\operatorname{span}(\mathbb D_F)$. Since $\operatorname{span}(\mathbb D_F)$ is dense in $\mathcal L_F$, the map of \eqref{eq:varphi-dot} extends uniquely to a continuous linear map on $\mathcal L_F$.
\end{proof}

\bigskip We now extend Lemma~\ref{lemma:expectation} to vector-valued integrands, which will be needed for maps such as $\gamma_a$ that involve vector and matrix expectations. Adopt the setup and notation of Lemma~\ref{lemma:expectation}. Fix $m\in\mathbb N$, and for each $v\in\mathcal V_{XY}$ let
\[
g_v:\mathcal Z\to\mathbb R^m,
\qquad
g_v=(g_{v,1},\dots,g_{v,m})^\top,
\]
where the $j$th coordinate $g_{v,j}:\mathcal Z\to\mathbb R$ is a
càdlàg function. Define
\[
\varphi:\mathbb D_\varphi \to
\ell^\infty(\mathcal V_{XY};\mathbb R^m),
\qquad
\varphi(F)(v)
:=
\int_{\mathcal Z} g_v(z)\, dF(z)
=
\begin{pmatrix}
\int_{\mathcal Z} g_{v,1}(z)\, dF(z)\\
\vdots\\
\int_{\mathcal Z} g_{v,m}(z)\, dF(z)
\end{pmatrix}.
\]
Equip $\ell^\infty(\mathcal V_{XY};\mathbb R^m)$ with
the norm
\[
\|f\|_{\ell^\infty(\mathcal V_{XY};\mathbb R^m)}
:=
\sup_{v\in\mathcal V_{XY}}\|f(v)\|_2 .
\]

\begin{lemma}[Vector-valued extension]
\label{lemma:vector_expectation}
Fix $F\in\mathbb D_\varphi$. Suppose that, for each
$j=1,\ldots,m$,
\[
\sup_{v\in\mathcal V_{XY}}\|g_{v,j}\|_{HK}<\infty.
\]
Then $\varphi$ is Hadamard differentiable at $F$ tangentially to $\mathbb D_F$. Its derivative is the continuous linear map $ \dot\varphi_F: \mathcal L_F\to \ell^\infty(\mathcal V_{XY};\mathbb R^m). $ For $h\in\mathcal L_F$, this derivative is given by
\[
\dot{\varphi}_F[h](v)
=
\int_{\mathcal Z} g_v(z)\,dh(z)=
\begin{pmatrix}
\int_{\mathcal Z} g_{v,1}(z)\, dh(z)\\
\vdots\\
\int_{\mathcal Z} g_{v,m}(z)\, dh(z)
\end{pmatrix},
\qquad v\in\mathcal V_{XY}.
\]
\end{lemma}
\begin{proof}
Write $\varphi=(\varphi_1,\dots,\varphi_{m})^\top$
with
\[
\varphi_j(F)(v)
=
\int_{\mathcal Z} g_{v,j}(z)\, dF(z).
\]
Each coordinate satisfies Lemma~\ref{lemma:expectation}, hence, for
$h\in\mathbb D_F$,
\[
\dot{\varphi}_{j,F}[h](v)
=
\int_{\mathcal Z} g_{v,j}(z)\,dh(z).
\]
Let $\epsilon_n\downarrow0$ and $h_n\to h$ in $\|\cdot\|_\infty$,
with $h_n,h\in\mathbb D_F$ and
$F+\epsilon_n h_n\in\mathbb D_\varphi$. Then
\begin{align*}
\left\|
\frac{\varphi(F+\epsilon_n h_n)-\varphi(F)}{\epsilon_n}
-
\dot{\varphi}_F[h]
\right\|_{\ell^\infty(\mathcal V_{XY};\mathbb R^m)}
&=
\sup_{v \in\mathcal V_{XY}}
\left\|
\frac{\varphi(F+\epsilon_n h_n)(v)-\varphi(F)(v)}{\epsilon_n}
-
\dot{\varphi}_F[h](v)
\right\|_2 \\
&\le 
\left(
\sum_{j=1}^m
\left\|
\frac{\varphi_j(F+\epsilon_n h_n)-\varphi_j(F)}{\epsilon_n}
-
\dot{\varphi}_{j,F}[h]
\right\|_{\ell^\infty(\mathcal V_{XY})}^2
\right)^{1/2} \\
&\to 0,
\end{align*}
with the last line following from Lemma~\ref{lemma:expectation}. Thus
$\varphi$ is Hadamard differentiable at $F$ tangentially to
$\mathbb D_F$, with the displayed derivative on $\mathbb D_F$.

It remains to specify the linear domain of the derivative. For each coordinate $j$, Lemma~\ref{lemma:expectation} gives a continuous linear map $\dot\varphi_{j,F}:\mathcal L_F\to\ell^\infty(\mathcal V_{XY}).$ For $h\in\mathcal L_F,$ define $\dot\varphi_F[h] := \bigl(\dot\varphi_{1,F}[h],\dots, \dot\varphi_{m,F}[h]\bigr)^\top. $ This map is linear coordinatewise. To show continuity, let $h_1,h_2\in\mathcal L_F$, and observe that
\begin{align*}
    \|\dot\varphi_F[h_1]-\dot\varphi_F[h_2]\|_{\ell^\infty(\mathcal V_{XY};\mathbb R^m)}
&=
\sup_{v\in\mathcal V_{XY}}
\left(
\sum_{j=1}^m
\left|
\dot\varphi_{j,F}[h_1](v)
-
\dot\varphi_{j,F}[h_2](v)
\right|^2
\right)^{1/2} \\
&\le
\left(
\sum_{j=1}^m
\|\dot\varphi_{j,F}[h_1]-\dot\varphi_{j,F}[h_2]\|_{\ell^\infty(\mathcal V_{XY})}^2
\right)^{1/2}.
\end{align*}
Each coordinate map is continuous on $\mathcal L_F$ by Lemma~\ref{lemma:expectation}, so the right-hand side tends to zero when $h_1-h_2\to0$ in $\mathcal L_F$. Hence
$\dot\varphi_F:\mathcal L_F\to
\ell^\infty(\mathcal V_{XY};\mathbb R^m)$ is continuous and linear. 
\end{proof}

\subsubsection{A bookkeeping map from empirical-process notation to CDF notation}
\label{sec:bookkeeping-map}

Before analyzing the map $\lambda$, we make explicit the relationship between two ways of representing the same distributional information. The empirical process naturally views a probability law through its action on indicator functions $v\in\mathcal V$, that is, through the map $v\mapsto Pv$. In contrast, the expectation-functional results in subsection~\ref{sec:expectation_functionals} are stated for CDF-indexed objects in the Skorokhod space $\mathbb D_S$. The following bookkeeping map connects these two representations.

Define
\begin{equation}
\label{eq:tau-map}
\tau:\ell^\infty(\mathcal V)\to\ell^\infty(\mathcal Z),
\qquad
\tau(Q)(s,a,t):=Qv_{s,a,t}. \nonumber
\end{equation}
In words, $\tau$ sends the indicator-indexed map $v\mapsto Qv$ to the threshold-indexed map $(s,a,t)\mapsto Qv_{s,a,t}$. Thus, $\tau$ simply changes the indexing of $Q$: instead of viewing $Q$ as a bounded map on the indicator class $\mathcal V$, we view $\tau(Q)$ as a bounded map on the threshold space $\mathcal Z$. When $Q$ corresponds to a distribution---meaning that $Q=P\in\mathbb{D}_\Phi$---this reindexing map makes $\tau(P)$ coincide with the CDF $F_P$; indeed,
\[
\tau(P)(s,a,t)
=
Pv_{s,a,t}
=
P(X\le s,\ A\le a,\ Y\le t)
=
F_P(s,a,t).
\]
For a generic $Q\in\ell^\infty(\mathcal V)$, however, $\tau(Q)$ may not be a distribution function or even a càdlàg function. This is why the codomain of $\tau$ is the larger space $\ell^\infty(\mathcal Z)$.

\begin{lemma}[Hadamard differentiability of the bookkeeping map]
\label{lemma:tau}
The map $\tau$ is linear and continuous. Consequently, $\tau$ is Hadamard differentiable at every $Q\in\ell^\infty(\mathcal V)$ with derivative $\dot\tau_Q[h]=\tau(h).$
\end{lemma}
\begin{proof}
Linearity follows immediately from the definition. Also, for any $Q\in\ell^\infty(\mathcal V)$,
\[
\|\tau(Q)\|_{\ell^\infty(\mathcal Z)}
=
\sup_{(s,a,t)\in\mathcal Z}|Q(v_{s,a,t})|
=
\sup_{v\in\mathcal V}|Q(v)|
=
\|Q\|_{\ell^\infty(\mathcal V)}.
\]
Thus $\tau$ is a bounded linear map, hence continuous.

Now let $\epsilon_n\downarrow0$ and $h_n\to h$ in $\ell^\infty(\mathcal V)$. By linearity,
\[
\frac{\tau(Q+\epsilon_n h_n)-\tau(Q)}{\epsilon_n}
=
\tau(h_n).
\]
Therefore,
\[
\left\|
\frac{\tau(Q+\epsilon_n h_n)-\tau(Q)}{\epsilon_n}
-
\tau(h)
\right\|_{\ell^\infty(\mathcal Z)}
=
\|\tau(h_n-h)\|_{\ell^\infty(\mathcal Z)}
=
\|h_n-h\|_{\ell^\infty(\mathcal V)}
\to0.
\]
Finally, the derivative map $h\mapsto \dot\tau_Q[h]=\tau(h)$ is linear and continuous because $\tau$ itself is linear and continuous. Hence $\tau$ is Hadamard differentiable at $Q$ tangentially to $\ell^\infty(\mathcal V)$.
\end{proof}

\subsubsection{The conditional law map}
\label{sec:conditional-law-map}

We now establish differentiability of the conditional-law map
$\lambda$ defined in \eqref{eq:lambda}. Recall that $\lambda(P)=(P_0,P_1,\pi_1(P)).$
Each conditional law $P_a$ can be written as a ratio of two expectation-type functionals: an arm-specific numerator $U_a(P)$ divided by the treatment probability $\pi_a(P)$. The bookkeeping map $\tau$ allows us to move between the empirical-process representation $v\mapsto Pv$ and the corresponding CDF representation $F_P=\tau(P)\in\mathbb D_\varphi$. Throughout this subsection, we use the convention from the overview and write $P$ both for a probability law and for its evaluation map $v\mapsto Pv$.

For $a\in\{0,1\}$ and $v_{s,t}\in\mathcal V_{XY}$, define
\begin{align}
\label{eq:pi-U-Pa-defs}
\pi_a(P)
:= P(A=a),  \qquad
U_a(P)(v_{s,t})
:= P(X\le s,A=a,Y\le t), \qquad
P_a v_{s,t}
:= \frac{U_a(P)(v_{s,t})}{\pi_a(P)}. \nonumber 
\end{align}
Equivalently, if $F_P=\tau(P)$, then
\[
U_a(P)(v_{s,t})
=
\int_{\mathcal Z}
\mathbf 1\{x\le s,\ a'=a,\ y\le t\}\,dF_P(x,a',y).
\]

For $P\in\mathbb D_\Phi$, let $F_P:=\tau(P)\in\mathbb D_\varphi$ denote the corresponding joint distribution function on $\mathcal Z$, and let $\mathbb D_{F_P}$ be as in \eqref{eq:Dphi-DF} with $F=F_P$. Define
\begin{equation}
\label{eq:tangent-set}
\mathcal T_P
:=
\left\{
h\in\ell^\infty(\mathcal V):
\tau(h)\in \mathbb D_{F_P}
\right\}. 
\end{equation}
Thus $\mathcal T_P$ consists of perturbation directions in $\ell^\infty(\mathcal V)$ whose CDF-indexed versions, obtained through $\tau$, lie in the tangent cone at $F_P$. In particular, this is the space in which we view empirical-process directions such as $\sqrt n(P_n-P)$. 
Because $\mathcal P$ contains the point masses $\delta_z$, the point-mass directions $\delta_z-P$ also belong to $\mathcal T_P$ for every $z\in\mathcal Z$. 

As in the CDF-indexed setting, the tangent cone $\mathcal T_P$ gives the allowable tangential directions, but it is not generally the linear space on which the derivative is defined. We therefore define its associated closed linear span
\begin{equation}
\label{eq:LP}
\mathcal L_P
:=
\overline{\operatorname{span}(\mathcal T_P)}
\subseteq \ell^\infty(\mathcal V),
\end{equation}
where the closure is taken with respect to the $\ell^\infty(\mathcal V)$ norm.

\begin{proposition}[Hadamard differentiability of $\lambda$]
\label{prop:lambda}
Fix $P\in\mathbb D_\Phi$ with $\pi_a(P)>0$ for $a\in\{0,1\}$. Then $\lambda$, defined in \eqref{eq:lambda}, is Hadamard differentiable at $P$ tangentially to $\mathcal T_P$. Its derivative is the continuous linear map $ \dot\lambda_P:\mathcal L_P\to \ell^\infty(\mathcal V_{XY})\times \ell^\infty(\mathcal V_{XY})\times \mathbb R, $ given, for $h\in\mathcal L_P$, by
\[
\dot\lambda_P[h]
=
\bigl(\dot P_{0,P}[h],\,\dot P_{1,P}[h],\,\dot\pi_{1,P}[h]\bigr),
\]
where, for $v=v_{s,t}\in\mathcal V_{XY}$,
\begin{align}
\label{eq:pi-U-Pa-dots}
\dot\pi_{a,P}[h]
&=
\int \mathbf 1\{A=a\}\,dh, \qquad
\dot U_{a,P}[h](v)
=
\int \mathbf 1\{x\le s,\ a'=a,\ y\le t\}\,dh,  \nonumber \\
\dot P_{a,P}[h](v)
&=
\frac{\dot U_{a,P}[h](v)}{\pi_a(P)}
-
\frac{U_a(P)(v)}{\pi_a(P)^2}\,\dot\pi_{a,P}[h].
\end{align}
\end{proposition}
\begin{proof}
We decompose $\lambda$ into the maps $U_a$, $\pi_a$, and the ratio $(u,p)\mapsto u/p$, and then apply the chain rule.

\noindent\textbf{Step 1: Differentiability of $U_a$ and $\pi_a$.} We first pass from the empirical-process representation of $P$ to its CDF representation. By Lemma~\ref{lemma:tau}, the map $\tau:\ell^\infty(\mathcal V)\to\ell^\infty(\mathcal Z)$ is Hadamard differentiable, with derivative $\dot\tau_P[h]=\tau(h)$. For such a $P$, we have $F_P=\tau(P)\in\mathbb D_\varphi$.

For each $(s,t)\in\mathbb R^{q+1}\times\mathcal Y$, define
\[
g_{s,a,t}:\mathcal Z\to\mathbb R,
\qquad
g_{s,a,t}(x,a',y)
:=
\mathbf 1\{x\le s,\ a'=a,\ y\le t\}.
\]
Then
\[
U_a(P)(v_{s,t})
=
\int_{\mathcal Z} g_{s,a,t}(z)\,dF_P(z)
=
\int_{\mathcal Z} g_{s,a,t}(z)\,d\tau(P)(z).
\]
Thus $U_a$ is the composition of $\tau$ with the expectation functional from Lemma~\ref{lemma:expectation}. The class $\{g_{s,a,t}:(s,t)\}$ consists of lower-rectangle indicators and hence satisfies the bounded Hardy--Krause variation condition in Lemma~\ref{lemma:expectation}. Therefore, by Lemma~\ref{lemma:tau}, Lemma~\ref{lemma:expectation}, and the chain rule, $U_a$ is Hadamard differentiable at $P$ tangentially to $\mathcal T_P$, with derivative
\[
\dot U_{a,P}[h](v_{s,t})
=
\int_{\mathcal Z} g_{s,a,t}(z)\,d\{\tau(h)\}(z).
\]
Equivalently, writing $h$ for the corresponding signed perturbation,
\[
\dot U_{a,P}[h](v_{s,t})
=
\int \mathbf 1\{x\le s,\ a'=a,\ y\le t\}\,dh.
\]

The treatment probability map $\pi_a$ is handled identically by taking $g_a(x,a',y):=\mathbf 1\{a'=a\}$. Hence $\pi_a$ is Hadamard differentiable at $P$ tangentially to $\mathcal T_P$, with
\[
\dot\pi_{a,P}[h]
=
\int_{\mathcal Z} g_a(z)\,d\{\tau(h)\}(z)
=
\int \mathbf 1\{a'=a\}\,dh.
\]

\medskip

\noindent\textbf{Step 2: Differentiability of $P_a$.}
For $v\in\mathcal V_{XY}$,
\[
P_a v
=
\frac{U_a(P)(v)}{\pi_a(P)}.
\]
We use the fact that the ratio map $ R:\ell^\infty(\mathcal V_{XY})\times(\mathbb R\setminus\{0\}) \to \ell^\infty(\mathcal V_{XY})$, $R(u,p):=u/p$, is continuously Fréchet differentiable, hence Hadamard differentiable, with derivative
\[
\dot R_{u,p}[h,r](v)
=
\frac{h(v)}{p}
-
\frac{u(v)}{p^2}r.
\]
Given that $P_a=R(U_a(P),\pi_a(P))$ and that $U_a$ and $\pi_a$ are Hadamard differentiable at $P$ tangentially to $\mathcal T_P$, the chain rule gives that $P_a$ is Hadamard differentiable at $P$ tangentially to $\mathcal T_P$. Its derivative is
\[
\dot P_{a,P}[h]
=
\dot R_{U_a(P),\pi_a(P)}
\bigl[\dot U_{a,P}[h],\,\dot\pi_{a,P}[h]\bigr],
\]
so, for $v\in\mathcal V_{XY}$,
\[
\dot P_{a,P}[h](v)
=
\frac{\dot U_{a,P}[h](v)}{\pi_a(P)}
-
\frac{U_a(P)(v)}{\pi_a(P)^2}\,\dot\pi_{a,P}[h].
\]

\medskip

\noindent\textbf{Step 3: Assembly of $\lambda$.}
Since $\lambda$ is formed by combining the maps $P\mapsto P_0$, $P\mapsto P_1$, and $P\mapsto \pi_1(P)$, each of which is Hadamard differentiable at $P$ tangentially to $\mathcal T_P$, it follows that $\lambda$ is Hadamard differentiable at $P$ tangentially to $\mathcal T_P$, with derivative
\[
\dot\lambda_P[h]
=
\bigl(\dot P_{0,P}[h],\,\dot P_{1,P}[h],\,\dot\pi_{1,P}[h]\bigr).
\]

As in the proof of Lemma~\ref{lemma:expectation}, the displayed derivative extends from $\mathcal T_P$ to $\operatorname{span}(\mathcal T_P)$ by linearity of its component maps. The componentwise bounds above imply that this linear extension is bounded on $\operatorname{span}(\mathcal T_P)$. Since $\operatorname{span}(\mathcal T_P)$ is dense in $\mathcal L_P$, it extends uniquely to a continuous linear map
$
\dot\lambda_P:\mathcal L_P\to
\ell^\infty(\mathcal V_{XY})\times \ell^\infty(\mathcal V_{XY})\times \mathbb R.
$
\end{proof}

\bigskip The next component map is the coordinate projection $\mathrm{pr}_a$ defined in \eqref{eq:projection}. Its restriction to $\mathbb D_{\mathcal V_{XY}}\times\mathbb D_{\mathcal V_{XY}}\times(0,1)$ extracts the arm-$a$ conditional law from the output of $\lambda$: for any $P\in\mathbb D_\Phi$,
\begin{equation}
\label{eq:Pa-as-projection}
    P_a=\mathrm{pr}_a\{\lambda(P)\}.
\end{equation}

The projection $\mathrm{pr}_a$ is linear and continuous. Under the
product norm
\[
\|(Q_0,Q_1,r)\|
=
\|Q_0\|_{\ell^\infty(\mathcal V_{XY})}
+
\|Q_1\|_{\ell^\infty(\mathcal V_{XY})}
+
|r|,
\]
we have
\[
\|\mathrm{pr}_a(Q_0,Q_1,r)\|_{\ell^\infty(\mathcal V_{XY})}
=
\|Q_a\|_{\ell^\infty(\mathcal V_{XY})}
\le
\|(Q_0,Q_1,r)\|.
\]
Therefore $\mathrm{pr}_a$ is Fréchet differentiable, hence Hadamard differentiable, with derivative equal to itself: $\dot{\mathrm{pr}}_{a,(Q_0,Q_1,r)}[h_0,h_1,\rho] = \mathrm{pr}_a(h_0,h_1,\rho) = h_a$. Equivalently,
\begin{equation}
\label{eq:projection-derivative}
\dot{\mathrm{pr}}_{a,(Q_0,Q_1,r)}[\cdot]
=
\mathrm{pr}_a\{\cdot\}.
\end{equation}
In particular, combining \eqref{eq:projection-derivative} with
Proposition~\ref{prop:lambda} gives
\begin{equation}
\label{eq:projected-lambda-derivative}
\dot{\mathrm{pr}}_{a,\lambda(P)}[\dot\lambda_P[h]]
=
\mathrm{pr}_a\{\dot\lambda_P[h]\}
=
\dot P_{a,P}[h],
\end{equation}
where $\dot P_{a,P}[h]$ is given in \eqref{eq:pi-U-Pa-dots}.

\subsubsection{Balancing coefficients $\gamma_a$}
\label{sec:balancing-coefficients}

We next study the differentiability of the balancing coefficient map $\gamma_a$, defined in \eqref{eq:gamma}. To do so, we express $\gamma_a$ as a composition of expectation functionals and a finite-dimensional matrix map.

Let $\mathcal B_{\mathrm{inv}}
:=
\{B\in\mathbb R^{(q+1)\times(q+1)}:B\text{ is invertible}\},$ and define
\begin{equation}
\label{eq:G-map}
G:\mathcal B_{\mathrm{inv}}\times\mathbb R^{q+1}\times\mathbb R^{q+1}
\to\mathbb R^{q+1},
\qquad
G(B,b,c):=B^{-1}(b-c).
\end{equation}

\begin{lemma}[Differentiability of $(B,b,c)\mapsto B^{-1}(b-c)$]
\label{lemma:G}
The map $G$ is Fréchet differentiable on
$\mathcal B_{\mathrm{inv}}\times\mathbb R^{q+1}\times\mathbb R^{q+1}$.
At $(B_0,b_0,c_0)$, its derivative in direction $(H,h,k)$ is
\[
\dot G_{B_0,b_0,c_0}[H,h,k]
=
- B_0^{-1} H B_0^{-1}(b_0-c_0)
+
B_0^{-1}(h-k).
\]
\end{lemma}
The proof follows by the chain rule for matrix derivatives, and so is omitted.

For the differentiability analysis of $\gamma_a$, define
\begin{align}
\label{eq:Ba-ba-c-defs}
B_a(P_0,P_1,\pi) &:= \mathbb E_{P_a}[XX^\top], \quad
b_a(P_0,P_1,\pi) := \mathbb E_{P_a}[X], \quad
c(P_0,P_1,\pi) := (1-\pi)\mathbb E_{P_0}[X]+\pi \mathbb E_{P_1}[X]. \nonumber
\end{align}
Then, by \eqref{eq:gamma},
\[
\gamma_a(P_0,P_1,\pi)
=
G\bigl(B_a(P_0,P_1,\pi),\,b_a(P_0,P_1,\pi),\,c(P_0,P_1,\pi)\bigr).
\]

We use the same bookkeeping convention for arm-specific laws. Let
$\mathcal Z_{XY}$ denote the support of $(X,Y)$, and define
\begin{equation}
\label{eq:tau-prime-map}
\tau_{XY}:\ell^\infty(\mathcal V_{XY})\to\ell^\infty(\mathcal Z_{XY}),
\qquad
\tau_{XY}(Q)(s,t):=Q(v_{s,t}).  \nonumber
\end{equation}
This is the arm-specific analogue of $\tau$. Hence $\tau_{XY}$ is linear
and continuous, and is Hadamard differentiable with derivative $\dot\tau_{XY, Q}[h]=\tau_{XY}(h).$
As with $\tau$, when $Q$ corresponds to a conditional distribution, meaning that $Q=P_a\in\mathbb{D}_{\mathcal V_{XY}}$,
this reindexing recovers the arm-specific CDF:
\[
\tau_{XY}(P_a)(s,t)
=
P_a v_{s,t}
=
P_a(X\le s,Y\le t)
=
F_{P_a}(s,t).
\]
We use $\mathbb D_{F_{P,a}}$ to denote the corresponding tangent cone
in the arm-specific Skorokhod space.

For $a\in\{0,1\}$, define the arm-specific tangent set
\begin{equation}
\label{eq:tangent-set-P_a}
    \mathcal T_{P_a}
    :=
    \left\{
    h_a\in\ell^\infty(\mathcal V_{XY}):
    \tau_{XY}(h_a)\in\mathbb D_{F_{P,a}}
    \right\}.
\end{equation}
Thus $\mathcal T_{P_a}$ consists of perturbation directions for the arm-specific law $P_a$, represented in $\ell^\infty(\mathcal V_{XY})$, whose CDF-indexed versions lie in the tangent cone at $F_{P_a}$.

As above, $\mathcal T_{P_a}$ gives the allowable tangential directions for perturbing the arm-specific law, while the derivative will be defined on the associated closed linear span. Define
\begin{equation}
\label{eq:LP-a}
\mathcal L_{P_a}
:=
\overline{\operatorname{span}(\mathcal T_{P_a})}
\subseteq \ell^\infty(\mathcal V_{XY}),
\end{equation}
where the closure is taken with respect to the
$\ell^\infty(\mathcal V_{XY})$ norm.

In the following proposition and throughout the remainder of this section,
integrals with respect to $h_a$ are shorthand for integrals with respect
to the CDF-indexed perturbation $\tau_{XY}(h_a)$. For example,
\[
    \int x\,dh_a
    :=
    \int_{\mathcal Z_{XY}} x\,d\{\tau_{XY}(h_a)\}(x,y).
\]

\begin{proposition}[Hadamard differentiability of $\gamma_a$]
\label{prop:gamma}
Fix $(P_0,P_1,\pi)\in \mathbb D_{\mathcal V_{XY}}\times\mathbb D_{\mathcal V_{XY}}\times(0,1),$ where $P_0$ and $P_1$ are probability-induced arm-specific laws, and assume $B_a(P_0,P_1,\pi)=\mathbb E_{P_a}[XX^\top]$ is invertible. Then $\gamma_a$ is Hadamard differentiable at $(P_0,P_1,\pi)$, tangentially to $\mathcal T_{P_0}\times \mathcal T_{P_1}\times \mathbb R.$ Its derivative is
the continuous linear map $\dot\gamma_{a,P_0,P_1,\pi}: \mathcal L_{P_0}\times \mathcal L_{P_1}\times \mathbb R \to \mathbb R^{q+1}.$ This derivative is given by:
\[
\dot\gamma_{a,P_0,P_1,\pi}[h_0,h_1,r]
=
- B_a^{-1}
\dot B_{a,P_0,P_1,\pi}[h_0,h_1,r]
\gamma_a
+
B_a^{-1}
\left\{
\dot b_{a,P_0,P_1,\pi}[h_0,h_1,r]
-
\dot c_{P_0,P_1,\pi}[h_0,h_1,r]
\right\},
\]
where $B_a$ and $\gamma_a$ on the right-hand side are evaluated at $(P_0,P_1,\pi)$, and
\begin{align*}
\dot b_{a,P_0,P_1,\pi}[h_0,h_1,r] &=
\int x\,dh_a, \qquad
\dot B_{a,P_0,P_1,\pi}[h_0,h_1,r]
=
\int xx^\top\,dh_a,\\
\dot c_{P_0,P_1,\pi}[h_0,h_1,r]
&=
(1-\pi)\int x\,dh_0
+
\pi\int x\,dh_1
+
r\{\mathbb E_{P_1}[X]-\mathbb E_{P_0}[X]\}.
\end{align*}
\end{proposition}
\begin{proof}
We decompose $\gamma_a$ into expectation functionals and the finite-dimensional map $G(B,b,c)=B^{-1}(b-c)$, and then apply the chain rule. Throughout the proof, derivatives are first computed along directions in $\mathcal T_{P_0}\times\mathcal T_{P_1}\times\mathbb R,$ and the displayed derivative maps are continuous linear maps on $\mathcal L_{P_0}\times\mathcal L_{P_1}\times\mathbb R.$
\medskip

\noindent \textbf{Step 1: Differentiability of the expectation components.}
For $a\in\{0,1\}$, consider the maps
\[
P_a\mapsto \mathbb E_{P_a}[X],
\qquad
P_a\mapsto \mathbb E_{P_a}[XX^\top].
\]
Let $F_{P_a}:=\tau_{XY}(P_a)$ denote the distribution function corresponding to the arm-specific law $P_a$ on $\mathcal Z_{XY}$. By the arm-specific bookkeeping convention preceding \eqref{eq:tangent-set-P_a}, perturbations $h_a\in\mathcal T_{P_a}$ correspond to CDF-indexed perturbations $\tau_{XY}(h_a)\in \mathbb D_{F_{P,a}}$.

After the rescaling of the support to a compact rectangle, the coordinate map $x\mapsto x$ and the coordinatewise product map $x\mapsto \operatorname{vec}(xx^\top)$ are càdlàg and have bounded Hardy--Krause variation coordinatewise. Therefore, Lemma~\ref{lemma:vector_expectation}, together with the arm-specific bookkeeping map $\tau_{XY}$, implies that these expectation maps are Hadamard differentiable at $P_a$ tangentially to $\mathcal T_{P_a}$. For $h_a\in\mathcal T_{P_a}$,
\begin{equation}
\label{eq:multivariate_expectation_derivatives}
\dot{\mathbb E}_{P_a}[X][h_a]
=
\int x\,dh_a,
\qquad
\dot{\mathbb E}_{P_a}[XX^\top][h_a]
=
\int xx^\top\,dh_a.
\end{equation}

\medskip
\noindent \textbf{Step 2: Differentiability of $b_a$ and $B_a$.}
Recall that
\[
b_a(P_0,P_1,\pi):=\mathbb E_{P_a}[X],
\qquad
B_a(P_0,P_1,\pi):=\mathbb E_{P_a}[XX^\top].
\]
These maps depend on $(P_0,P_1,\pi)$ only through $P_a$, so their derivatives are obtained by first projecting onto the $a$th arm and then applying the expectation derivatives from Step 1. By \eqref{eq:projection-derivative}, the coordinate projection $\mathrm{pr}_a$ is Hadamard differentiable with derivative $(h_0,h_1,r)\mapsto h_a$. Combining this projection with the expectation derivatives in \eqref{eq:multivariate_expectation_derivatives}, the chain rule gives that $b_a$ and $B_a$ are Hadamard differentiable at $(P_0,P_1,\pi)$, tangentially to $\mathcal T_{P_0}\times \mathcal T_{P_1}\times \mathbb R$. For $(h_0,h_1,r)$ in this tangent set,
\[
\dot b_{a,P_0,P_1,\pi}[h_0,h_1,r]
=
\int x\,dh_a,
\qquad
\dot B_{a,P_0,P_1,\pi}[h_0,h_1,r]
=
\int xx^\top\,dh_a.
\]

\medskip

\noindent \textbf{Step 3: Differentiability of the mixture term $c$.}
Next consider
\[
c(P_0,P_1,\pi)
=
(1-\pi)\mathbb E_{P_0}[X]+\pi \mathbb E_{P_1}[X].
\]
Using the differentiability of the expectation maps from Step 1 and the product rule for scalar multiplication by $\pi$, we obtain
\[
\dot c_{P_0,P_1,\pi}[h_0,h_1,r]
=
(1-\pi)\int x\,dh_0
+
\pi\int x\,dh_1
+
r\left(\mathbb E_{P_1}[X]-\mathbb E_{P_0}[X]\right).
\]
Thus $c$ is Hadamard differentiable at $(P_0,P_1,\pi)$, tangentially to $\mathcal T_{P_0}\times \mathcal T_{P_1}\times \mathbb R$.

\medskip

\noindent \textbf{Step 4: Chain rule.} The maps $B_a$, $b_a$, and $c$ are Hadamard differentiable by Steps 2 and 3, and $G$ is Hadamard differentiable by Lemma~\ref{lemma:G}. Therefore, by the chain rule, $\gamma_a$ is Hadamard differentiable at $(P_0,P_1,\pi)$, tangentially to $\mathcal T_{P_0}\times \mathcal T_{P_1}\times \mathbb R$. Its derivative, defined on $\mathcal L_{P_0}\times\mathcal L_{P_1}\times\mathbb R$, is given by
\[
\dot\gamma_{a,P_0,P_1,\pi}[h_0,h_1,r]
=
\dot G_{B_a,b_a,c}
\left[
\dot B_{a,P_0,P_1,\pi}[h_0,h_1,r],
\dot b_{a,P_0,P_1,\pi}[h_0,h_1,r],
\dot c_{P_0,P_1,\pi}[h_0,h_1,r]
\right].
\]
Substituting the derivative of $G$ and using the fact that $\gamma_a=B_a^{-1}(b_a-c)$ gives the claimed expression.
\end{proof}

\subsubsection{Weighted distribution map $f$}
\label{sec:weighted-distribution-map}

We next analyze the weighted distribution map $f$, defined in \eqref{eq:f}. This map takes an arm-specific law $P_a$ and a coefficient vector $g$, and returns the signed weighted law obtained by weighting $P_a$ with $1-g^\top X$. We study $f$ as a function of both arguments $(P_a,g)$. This allows us to combine the derivative of $f$ with the derivatives of $\mathrm{pr}_a$ and $\gamma_a$ when analyzing the composite map
\[
(P_0,P_1,\pi)
\mapsto
f\bigl(\mathrm{pr}_a(P_0,P_1,\pi),\,\gamma_a(P_0,P_1,\pi)\bigr),
\]
and ultimately $\Phi$.

As in the preceding section, integrals with respect to $h_a$ are shorthand for integrals with respect to the CDF-indexed perturbation $\tau_{XY}(h_a)$.

\begin{lemma}[Hadamard differentiability of $f$]
\label{lemma:f}
Fix $P_a\in\mathbb D_{\mathcal V_{XY}}$, and fix $g\in\mathbb R^{q+1}$. Then $f$, defined in \eqref{eq:f}, is Hadamard differentiable at $(P_a,g)$ tangentially to $\mathcal T_{P_a}\times\mathbb R^{q+1}$, where $\mathcal T_{P_a}$ is defined in \eqref{eq:tangent-set-P_a}. Its derivative is the continuous linear map $ \dot f_{P_a,g}: \mathcal L_{P_a}\times\mathbb R^{q+1} \to \ell^\infty(\mathcal V_{XY}). $ For $v=v_{s,t}\in\mathcal V_{XY}$, this derivative is given by
\[
\dot f_{P_a,g}[h_a,u](v)
=
\int (1-g^\top x)v(x,y)\,dh_a(x,y)
-
u^\top \int x\,v(x,y)\,P_a(dx,dy).
\]
\end{lemma}
\begin{proof}
Let $\epsilon_n\downarrow0$, let $h_{a,n}\to h_a$ in $\ell^\infty(\mathcal V_{XY})$, and let $u_n\to u$ in $\mathbb R^{q+1}$, with $h_{a,n},h_a\in\mathcal T_{P_a}$. We verify the Hadamard differentiability convergence in $\ell^\infty(\mathcal V_{XY})$.

For fixed $g$ and $v=v_{s,t}\in\mathcal V_{XY}$, define $m_{v,g}(x,y):=(1-g^\top x)v(x,y).$ The map $(x,y)\mapsto 1-g^\top x$ is bounded, càdlàg, and of bounded Hardy--Krause variation on the compact rectangle $\mathcal Z_{XY}$. The lower-rectangle indicators $v_{s,t}$ are uniformly bounded and have uniformly bounded Hardy--Krause variation over $(s,t)$. By the product inequality for functions of bounded Hardy--Krause variation \citep[bottom of p.~251]{blumlinger1989topological}, the products $m_{v,g}=(1-g^\top x)v_{s,t}$ have bounded Hardy--Krause variation uniformly over $v\in\mathcal V_{XY}$. Thus $\sup_{v\in\mathcal V_{XY}}\|m_{v,g}\|_{HK}<\infty,$ so the class $\{m_{v,g}:v\in\mathcal V_{XY}\}$ satisfies the conditions of Lemma~\ref{lemma:expectation}. Therefore, for some $C_g<\infty$,
\begin{equation}
\label{eq:f-proof-h-bound}
\sup_{v\in\mathcal V_{XY}}
\left|
\int m_{v,g}\,d(h_{a,1}-h_{a,2})
\right|
\le
C_g\|h_{a,1}-h_{a,2}\|_{\ell^\infty(\mathcal V_{XY})}
\end{equation}
for all $h_{a,1},h_{a,2}\in\mathcal T_{P_a}$. Also, boundedness of $\mathcal{X}$ gives
\begin{equation}
\label{eq:f-proof-X-bound}
C_X
:=
\sup_{v\in\mathcal V_{XY}}
\left\|
\int x\,v(x,y)\,P_a(dx,dy)
\right\|_2
<\infty.
\end{equation}

After dividing by $\epsilon_n$ and subtracting the candidate derivative, a direct expansion gives
\[
\frac{
f(P_a+\epsilon_n h_{a,n},\,g+\epsilon_n u_n)-f(P_a,g)
}{\epsilon_n}
-
\dot f_{P_a,g}[h_a,u]
=
C_{1,n}+C_{2,n}+C_{3,n},
\]
where, as elements of $\ell^\infty(\mathcal V_{XY})$,
\begin{align*}
C_{1,n}(v)
&:=
\int m_{v,g}\,d(h_{a,n}-h_a), \qquad
C_{2,n}(v)
:=
-(u_n-u)^\top\int x\,v(x,y)\,P_a(dx,dy),\\
C_{3,n}(v)
&:=
-\epsilon_n u_n^\top\int x\,v(x,y)\,dh_{a,n}.
\end{align*}
We show that each term converges to zero in $\ell^\infty(\mathcal V_{XY})$. By \eqref{eq:f-proof-h-bound},
\[
\|C_{1,n}\|_{\ell^\infty(\mathcal V_{XY})}
\le
C_g\|h_{a,n}-h_a\|_{\ell^\infty(\mathcal V_{XY})}
\to0.
\]
By Cauchy--Schwarz and \eqref{eq:f-proof-X-bound}, $\|C_{2,n}\|_{\ell^\infty(\mathcal V_{XY})}
\le
C_X\|u_n-u\|_2
\to0$. For the third term, Cauchy--Schwarz gives
\[
\|C_{3,n}\|_{\ell^\infty(\mathcal V_{XY})}
\le
\epsilon_n\|u_n\|_2
\sup_{v\in\mathcal V_{XY}}
\left\|
\int x\,v(x,y)\,dh_{a,n}(x,y)
\right\|_2.
\]

Since $u_n\to u$, the sequence $\{\|u_n\|_2\}$ is bounded. Also, $h_{a,n}\to h_a$ in $\ell^\infty(\mathcal V_{XY})$, so $\sup_n \|h_{a,n}\|_{\ell^\infty(\mathcal V_{XY})}<\infty.$ For each coordinate $j$, the functions $(x,y)\mapsto x_j v(x,y),$ with $v\in\mathcal V_{XY},$ have uniformly bounded Hardy--Krause variation. Indeed, $x_j$ is bounded, càdlàg, and of bounded Hardy--Krause variation on the compact support, and multiplication by the lower-rectangle indicator $v$ preserves this property uniformly by the product inequality of \citet[bottom of p.~251]{blumlinger1989topological}. Hence the same Lebesgue--Stieltjes bound used for \eqref{eq:f-proof-h-bound}, combined with the fact that there are only finitely many coordinates $j,$ gives
\[
\sup_n
\sup_{v\in\mathcal V_{XY}}
\left\|
\int x\,v(x,y)\,dh_{a,n}(x,y)
\right\|_2
<\infty.
\]
Since $\epsilon_n\to0$, it follows that $\|C_{3,n}\|_{\ell^\infty(\mathcal V_{XY})}\to0$. 
Therefore,
\[
\left\|
\frac{
f(P_a+\epsilon_n h_{a,n},\,g+\epsilon_n u_n)-f(P_a,g)
}{\epsilon_n}
-
\dot f_{P_a,g}[h_a,u]
\right\|_{\ell^\infty(\mathcal V_{XY})}
\to0.
\]

It remains to specify the linear domain of the derivative. The displayed formula extends from $\mathcal T_{P_a}\times\mathbb R^{q+1}$ to $\operatorname{span}(\mathcal T_{P_a})\times\mathbb R^{q+1}$ by linearity of the Lebesgue--Stieltjes integral in $h_a$ and linearity of the inner product in $u$. For $(h_{a,1},u_1)$ and $(h_{a,2},u_2)$ in this linear domain, \eqref{eq:f-proof-h-bound}, \eqref{eq:f-proof-X-bound}, and Cauchy--Schwarz give
\[
\begin{aligned}
\left\|
\dot f_{P_a,g}[h_{a,1},u_1]
-
\dot f_{P_a,g}[h_{a,2},u_2]
\right\|_{\ell^\infty(\mathcal V_{XY})}
&\le
C_g\|h_{a,1}-h_{a,2}\|_{\ell^\infty(\mathcal V_{XY})}
+
C_X\|u_1-u_2\|_2 .
\end{aligned}
\]
Hence $\dot f_{P_a,g}$ is bounded and linear on $\operatorname{span}(\mathcal T_{P_a})\times\mathbb R^{q+1}$. Since $\operatorname{span}(\mathcal T_{P_a})$ is dense in $\mathcal L_{P_a}$, this bounded linear map extends uniquely to a continuous linear map $ \dot f_{P_a,g}: \mathcal L_{P_a}\times\mathbb R^{q+1} \to \ell^\infty(\mathcal V_{XY}). $
\end{proof}

\subsubsection{Differentiability of $\Phi$}
\label{sec:differentiability-Phi}

We now combine the preceding componentwise differentiability results to establish Hadamard differentiability of the full SBW functional $\Phi$. Let $\eta_a$ be as in \eqref{eq:eta-map} and the full SBW functional $\Phi$ be as defined in \eqref{eq:Phi-composition}.

Fix $P\in\mathbb D_\Phi$ and assume:
\begin{enumerate}
    \item $\pi_a(P)>0$ for $a\in\{0,1\}$;
    \item $\mathbb E_{P_a}[XX^\top]$ is invertible for $a\in\{0,1\}$;
    \item Let $\mathbb D_\Psi\subseteq
    \ell^\infty(\mathcal V_{XY})\times \ell^\infty(\mathcal V_{XY})$ be a domain on which the target functional $\Psi$ is well defined, meaning that it contains the arm-specific inputs for which the estimand is meaningful, including $(\eta_0(P),\eta_1(P))$. Under randomization, $(\eta_0(P),\eta_1(P))=(P_0,P_1)$, so this condition is imposed at the arm-specific laws. Define
    \[
    \mathcal T_{\Psi,P}
    :=
    \left\{
    \bigl(\dot\eta_{0,P}[h],\dot\eta_{1,P}[h]\bigr):
    h\in\mathcal T_P
    \right\},
    \qquad
    \mathcal L_{\Psi,P}
    :=
    \overline{\operatorname{span}(\mathcal T_{\Psi,P})}.
    \]
    Assume that $\Psi:\mathbb D_\Psi\to\mathbb R^d$ is Hadamard differentiable at $(\eta_0(P),\eta_1(P))=(P_0,P_1)$ tangentially to $\mathcal T_{\Psi,P}$, with derivative given by a continuous linear map $\dot\Psi_{\eta_0(P),\eta_1(P)}: \mathcal L_{\Psi,P}\to\mathbb R^d.$ 
\end{enumerate}
We first analyze the intermediate maps $\eta_a$, which combine $\lambda$, $\mathrm{pr}_a$, $\gamma_a$, and $f$. We then apply the chain rule to the outer functional $\Psi$. 

\begin{theorem}[Hadamard differentiability of $\Phi$]
\label{thm:Phi}
Under the assumptions above, $\Phi$ is Hadamard differentiable at $P$
tangentially to $\mathcal T_P$. Its derivative is the continuous linear map $\dot\Phi_P:\mathcal L_P\to\mathbb R^d.$ For $h\in\mathcal L_P$, this derivative is given by
\begin{equation}
    \dot\Phi_P[h]
    =
    \dot\Psi_{\eta_0(P),\eta_1(P)}
    \bigl[
    \dot\eta_{0,P}[h],\,
    \dot\eta_{1,P}[h]
    \bigr]. \label{eq:dot-Phi_P}
\end{equation}
For $a\in\{0,1\}$,
$\dot\eta_{a,P}:\mathcal L_P\to\ell^\infty(\mathcal V_{XY})$ is a continuous linear map: for $h\in\mathcal L_P$ and $v=v_{s,t}\in\mathcal V_{XY}$,
\begin{align}
    \dot\eta_{a,P}[h](v) &=
    \int \bigl(1-\gamma_a\{\lambda(P)\}^\top x\bigr)\,v(x,y)\,
    d\!\left(\mathrm{pr}_a\{\dot\lambda_P[h]\}\right)(x,y) \nonumber \\
    &\quad
    -
    \bigl(\dot\gamma_{a,\lambda(P)}[\dot\lambda_P[h]]\bigr)^\top
    \int x\,v(x,y)\,P_a(dx,dy). \label{eq:eta-dot-definition}
\end{align}
\end{theorem}
\begin{proof}
For each $a\in\{0,1\}$, define $ T_a(P_0,P_1,\pi) := f\bigl(\mathrm{pr}_a(P_0,P_1,\pi),\,\gamma_a(P_0,P_1,\pi)\bigr). $ Then $\eta_a=T_a\circ\lambda$.

\medskip

\noindent\textbf{Step 1: Differentiability of $T_a$.}
The map $T_a$ is the composition of
\[
(P_0,P_1,\pi)
\mapsto
\bigl(\mathrm{pr}_a(P_0,P_1,\pi),\gamma_a(P_0,P_1,\pi)\bigr)
\]
with the map $f$. By \eqref{eq:projection-derivative} and Proposition~\ref{prop:gamma}, the
pair map
\[
(P_0,P_1,\pi)
\mapsto
\bigl(\mathrm{pr}_a(P_0,P_1,\pi),\gamma_a(P_0,P_1,\pi)\bigr)
\]
is Hadamard differentiable at $\lambda(P)$ tangentially to
$\mathcal T_{P_0}\times\mathcal T_{P_1}\times\mathbb R$. Its derivative is
\[
(h_0,h_1,r)
\mapsto
\bigl(h_a,\dot\gamma_{a,\lambda(P)}[h_0,h_1,r]\bigr).
\]
For tangential directions $(h_0,h_1,r)\in\mathcal T_{P_0}\times\mathcal T_{P_1}\times\mathbb R$, the first component $h_a$ belongs to $\mathcal T_{P_a}$, so the derivative of the pair map takes tangential directions into the tangent set required by Lemma~\ref{lemma:f}. Applying the chain rule with Lemma~\ref{lemma:f}, we obtain that $T_a$ is Hadamard differentiable at $\lambda(P)$ tangentially to $\mathcal T_{P_0}\times\mathcal T_{P_1}\times\mathbb R$, with derivative
\[
\dot T_{a,\lambda(P)}[h_0,h_1,r]
=
\dot f_{P_a,\gamma_a\{\lambda(P)\}}
\bigl[
h_a,\dot\gamma_{a,\lambda(P)}[h_0,h_1,r]
\bigr],
\]
where $P_a=\mathrm{pr}_a\{\lambda(P)\}$. The displayed formula defines a continuous linear map on the corresponding linear domain because $\mathrm{pr}_a$, $\dot\gamma_{a,\lambda(P)}$, and $\dot f_{P_a,\gamma_a\{\lambda(P)\}}$ are continuous linear maps on their respective linear domains.

\medskip

\noindent\textbf{Step 2: Differentiability of $\eta_a$.} By Proposition~\ref{prop:lambda}, $\lambda$ is Hadamard differentiable at $P$ tangentially to $\mathcal T_P$, with derivative $ \dot\lambda_P:\mathcal L_P \to \ell^\infty(\mathcal V_{XY})\times \ell^\infty(\mathcal V_{XY})\times\mathbb R $ a continuous linear map. For tangential directions $h\in\mathcal T_P$, $\dot\lambda_P[h]$ lies in the tangent set for $T_a$, so the chain rule applied to $\eta_a=T_a\circ\lambda$ gives that $\eta_a$ is Hadamard differentiable tangentially to $\mathcal T_P$.

Its derivative is the continuous linear map
\[
\dot\eta_{a,P}
:=
\dot T_{a,\lambda(P)}\circ \dot\lambda_P
:
\mathcal L_P\to \ell^\infty(\mathcal V_{XY}).
\]
Thus, for $h\in\mathcal L_P$,
\[
\dot\eta_{a,P}[h]
=
\dot f_{P_a,\gamma_a\{\lambda(P)\}}
\Bigl[
\mathrm{pr}_a\{\dot\lambda_P[h]\},\,
\dot\gamma_{a,\lambda(P)}[\dot\lambda_P[h]]
\Bigr].
\]
Using the derivative formula for $f$ from Lemma~\ref{lemma:f}, we obtain \eqref{eq:eta-dot-definition}.

\medskip

\noindent\textbf{Step 3: Differentiability of $\Phi$.} By assumption, $\Psi$ is Hadamard differentiable at $(\eta_0(P),\eta_1(P))$ tangentially to $\mathcal T_{\Psi,P}$ with derivative a continuous linear map $ \dot\Psi_{\eta_0(P),\eta_1(P)}: \mathcal L_{\Psi,P}\to\mathbb R^d. $ Since the map $ h \mapsto \bigl(\dot\eta_{0,P}[h],\dot\eta_{1,P}[h]\bigr) $ is continuous and linear from $\mathcal L_P$ into $\mathcal L_{\Psi,P}$, the chain rule applied to $ \Phi(P)=\Psi\bigl(\eta_0(P),\eta_1(P)\bigr) $ implies that $\Phi$ is Hadamard differentiable at $P$ tangentially to $\mathcal T_P$. Its derivative is the continuous linear map in \eqref{eq:dot-Phi_P}.
\end{proof}

\subsection{Continuous extension of the derivative}
\label{sec:continuous-extension}

The Hadamard differentiability result above is sufficient for the ordinary functional delta method, and hence for asymptotic normality and bootstrap validity of the plug-in estimator. In this subsection we establish a stronger property needed for Corollary~\ref{thm:Phi_asymptotic_linearity}: the derivative $\dot\Phi_P$, already defined as a continuous linear map on $\mathcal L_P$, admits a continuous linear extension to a larger space of càdlàg perturbations.

\subsubsection{Extension of the derivative components}

We do not extend the derivative to all of $\ell^\infty(\mathcal V)$, since an arbitrary bounded function on $\mathcal V$ may not correspond, through $\tau$, to a càdlàg CDF-indexed perturbation, and the Lebesgue--Stieltjes integrals appearing in the derivative may then be ill-defined. Instead, define
\begin{equation}
\label{eq:C}
\mathcal C
:=
\left\{
h\in\ell^\infty(\mathcal V):\tau(h)\in\mathbb D_S
\right\},
\qquad
\|h\|_{\mathcal C}
:=
\|\tau(h)\|_\infty. \nonumber
\end{equation}
Since $\mathbb D_{F_P}\subseteq\mathbb D_S$, we have $\mathcal T_P\subseteq\mathcal C$, and because $\mathcal C$ is a closed linear subspace of $\ell^\infty(\mathcal V)$ under the ambient sup norm, $\mathcal L_P$ is also contained in $\mathcal C$. The spaces $\mathcal L_P$ and $\mathcal C$ play different roles: $\mathcal L_P$ is the linear derivative domain used in the Hadamard differentiability argument, while $\mathcal C$ is a càdlàg extension space on which the same derivative formula can be interpreted continuously.

We will also use the arm-specific analogue of $\mathcal C$. Let $\mathbb D_{XY}$ denote the
multivariate Skorokhod space of càdlàg functions on $\mathcal Z_{XY}$, equipped
with the sup norm, and define
\begin{equation}
\label{eq:C-prime}
\mathcal C_{XY}
:=
\left\{
h\in\ell^\infty(\mathcal V_{XY}):\tau_{XY}(h)\in\mathbb D_{XY}
\right\},
\qquad
\|h\|_{\mathcal C_{XY}}
:=
\|\tau_{XY}(h)\|_\infty. \nonumber
\end{equation}
For $h\in\mathcal C_{XY}$, integrals with respect to $h$ are interpreted through
$\tau_{XY}(h)$.

We first extend the derivative of the conditional-law map $\lambda$. For
$h\in\mathcal C$, write $F_h:=\tau(h)$ and define
\[
\tilde F_{h,a}(s,t)
:=
\begin{cases}
F_h(s,0,t), & a=0,\\
F_h(s,1,t)-F_h(s,0,t), & a=1.
\end{cases}
\]
Thus $\tilde F_{h,a}$ is the CDF-indexed version of the arm-specific numerator
perturbation for the event $A=a$. If $(\bar s,\bar t)$ denotes the upper endpoint
of the support of $(X,Y)$, then the extended versions of the numerator and
treatment-probability derivatives are
\[
\dot U_{a,P}[h](v_{s,t})
=
\tilde F_{h,a}(s,t),
\qquad
\dot\pi_{a,P}[h]
=
\tilde F_{h,a}(\bar s,\bar t).
\]
We therefore define
\begin{equation}
\label{eq:dot-P_a-as-cdfs}
\dot P_{a,P}[h](v_{s,t})
=
\frac{1}{\pi_a(P)}
\left\{
\tilde F_{h,a}(s,t)
-
P_a(v_{s,t})\tilde F_{h,a}(\bar s,\bar t)
\right\},
\qquad h\in\mathcal C.
\end{equation}
For $h\in\mathcal L_P$, this agrees with the derivative formula in
\eqref{eq:pi-U-Pa-dots}.

\begin{lemma}[Continuous extension of $\dot\lambda_P$]
\label{lemma:lambda_continuous_extension}
Assume $\pi_a(P)>0$ for $a\in\{0,1\}$. The derivative $\dot{\lambda}_P$ from Proposition~\ref{prop:lambda}, defined on $\mathcal L_P$, extends continuously to $\mathcal C\supseteq \mathcal L_P$, yielding a continuous linear map $ \dot\lambda_P: \mathcal C \to \mathcal C_{XY} \times \mathcal C_{XY} \times \mathbb R $ via $ \dot\lambda_P[h] = \bigl( \dot P_{0,P}[h], \dot P_{1,P}[h], \dot\pi_{1,P}[h] \bigr). $
\end{lemma}
\begin{proof}
Linearity of $\dot \lambda_P$ follows from the definitions of $\tilde F_{h,a}$, $\dot\pi_{a,P}[h]$, and $\dot P_{a,P}[h]$. Indeed, $h\mapsto F_h=\tau(h)$ is linear by definition, and the remaining operations are fixed linear combinations and endpoint evaluations.

We now check that the image lies in the claimed space. If $h\in\mathcal C$, then $F_h\in\mathbb D_S$, so $(s,t)\mapsto\tilde F_{h,a}(s,t)$ is càdlàg on $\mathcal Z_{XY}$. Also, $(s,t)\mapsto P_a(v_{s,t})$ is the arm-specific distribution function, and hence is càdlàg. Equation~\eqref{eq:dot-P_a-as-cdfs} is therefore a linear combination of càdlàg functions, so $\dot P_{a,P}[h]\in\mathcal C_{XY}$.

It remains to show boundedness; continuity then follows from linearity. For any
$h\in\mathcal C$,
\[
\|\tilde F_{h,a}\|_\infty
\le
C_a\|F_h\|_\infty
=
C_a\|h\|_{\mathcal C},
\qquad
C_0=1,\quad C_1=2,
\]
and the same bound holds at the upper endpoint $(\bar s,\bar t)$. Hence,
using $\|P_a\|_{\ell^\infty(\mathcal V_{XY})}\le1$,
\begin{align*}
\left\|
\dot P_{a,P}[h]
\right\|_{\mathcal C_{XY}}
&\le
\frac{1}{\pi_a(P)}
\left\{
\|\tilde F_{h,a}\|_\infty
+
\|P_a\|_{\ell^\infty(\mathcal V_{XY})}
|\tilde F_{h,a}(\bar s,\bar t)|
\right\} \\
&\le
\frac{2C_a}{\pi_a(P)}
\|h\|_{\mathcal C}.
\end{align*}
Thus $h\mapsto\dot P_{a,P}[h]$ is a bounded linear map from
$\mathcal C$ to $\mathcal C_{XY}$, and hence is continuous. Similarly,
\[
\left|
\dot\pi_{1,P}[h]
\right|
=
|\tilde F_{h,1}(\bar s,\bar t)|
\le
2\|h\|_{\mathcal C},
\]
so $h\mapsto\dot\pi_{1,P}[h]$ is a bounded linear map from $\mathcal C$ to
$\mathbb R$, and hence is continuous.

Combining the two arm-specific components with the scalar component proves that $\dot\lambda_P$ is continuous and linear on $\mathcal C$.
\end{proof}
As an immediate consequence of Lemma~\ref{lemma:lambda_continuous_extension},
for each $a\in\{0,1\}$,
\begin{equation}
\label{eq:projected-lambda-continuous-extension}
h
\mapsto
\dot P_{a,P}[h]
=
\mathrm{pr}_a\{\dot\lambda_P[h]\}
\end{equation}
is a continuous linear map from $\mathcal C$ to $\mathcal C_{XY}$.

We next extend the derivative of the balancing coefficient map, evaluated along the derivative of $\lambda$. From this point forward in the extension argument, we use the randomized trial structure. Under randomization, $P_0$ and $P_1$ have the same covariate marginal distribution, so $\mathbb E_{P_0}[X]=\mathbb E_{P_1}[X]$. This removes the term involving $\dot\pi_1$ from the derivative of $\gamma_a$.

\begin{lemma}[Continuity of $\dot\gamma_a\circ\dot\lambda_P$]
\label{lemma:gamma_lambda_continuous_extension}
Fix $a\in\{0,1\}$. Assume $B_a=\mathbb E_{P_a}[XX^\top]$ is invertible, and that $\pi_b(P)>0$ for $b\in\{0,1\}$. Under randomization,
$\dot\gamma_{a,\lambda(P)}\circ\dot\lambda_P$, defined on $\mathcal L_P$,
extends continuously to $\mathcal C\supseteq\mathcal L_P$, yielding a
continuous linear map $\mathcal C\to\mathbb R^{q+1}$. For $h\in\mathcal C$,
this extension is given by
\begin{equation}
\label{eq:gamma-lambda-continuous-extension}
\dot\gamma_{a,\lambda(P)}[\dot\lambda_P[h]]
=
B_a^{-1}
\left[
\int x\,d\{\dot P_{a,P}[h]\}
-
(1-\pi)\int x\,d\{\dot P_{0,P}[h]\}
-
\pi\int x\,d\{\dot P_{1,P}[h]\}
\right].
\end{equation}
\end{lemma}
\begin{proof}
Under randomization, the general derivative formula for $\gamma_a$ reduces to \eqref{eq:gamma-lambda-continuous-extension}. For each $j$, $x\mapsto x_j$ is càdlàg and has bounded Hardy--Krause variation on the bounded support $\mathcal Z_{XY}$. Therefore the same Lebesgue--Stieltjes bound used in Lemma~\ref{lemma:expectation} implies that $ Q\mapsto \int x\,dQ $ is a continuous linear map from $\mathcal C_{XY}$ to $\mathbb R^{q+1}$. Combining this result with the continuity and linearity of $h\mapsto\dot P_{a,P}[h]$ in \eqref{eq:projected-lambda-continuous-extension}, we have that $ h \mapsto \int x\,d\{\dot P_{a,P}[h]\} $ is continuous and linear from $\mathcal C$ to $\mathbb R^{q+1}$ for each $a\in\{0,1\}$. The right-hand side of \eqref{eq:gamma-lambda-continuous-extension} is a fixed finite linear combination of these maps, followed by multiplication by the fixed matrix $B_a^{-1}$. Thus $h\mapsto\dot\gamma_{a,\lambda(P)}[\dot\lambda_P[h]]$ is continuous and linear on $\mathcal C$.
\end{proof}
We now extend the derivative of the arm-specific weighted distribution map $\eta_a$. Under randomization, $\gamma_a\{\lambda(P)\}=0$, so the derivative from Theorem~\ref{thm:Phi} reduces to
\begin{equation}
\label{eq:eta-derivative-extension}
\dot\eta_{a,P}[h](v)
=
\dot P_{a,P}[h](v)
-
\dot\gamma_{a,\lambda(P)}[\dot\lambda_P[h]]^\top M_a(v),
\qquad
M_a(v):=\int x\,v(x,y)\,P_a(dx,dy).
\end{equation}

\begin{lemma}[Continuity of $\dot\eta_{a,P}$]
\label{lemma:eta_continuous_extension}
Assume the conditions of Lemmas~\ref{lemma:lambda_continuous_extension} and~\ref{lemma:gamma_lambda_continuous_extension}. Under randomization, $\dot\eta_{a,P}$, defined on $\mathcal L_P$, extends continuously to $\mathcal C\supseteq\mathcal L_P$, yielding a continuous linear map $\dot\eta_{a,P}:\mathcal C\to\mathcal C_{XY}$.
\end{lemma}
\begin{proof}
The first term of \eqref{eq:eta-derivative-extension} belongs to $\mathcal C_{XY}$ by
\eqref{eq:projected-lambda-continuous-extension}. For the second term, write
$M_a(s,t):=M_a(v_{s,t})$. Each coordinate of $M_a$ has the form
\[
M_{a,j}(s,t)
=
\int x_j \mathbf 1\{x\le s,y\le t\}\,P_a(dx,dy),
\]
and is therefore the distribution function of the finite signed measure $B\mapsto\int_B x_j\,dP_a$. Thus each coordinate of $M_a$ is càdlàg on $\mathcal Z_{XY}$. Since $\dot\gamma_{a,\lambda(P)}[\dot\lambda_P[h]]\in\mathbb R^{q+1}$, the map $(s,t)\mapsto \dot\gamma_{a,\lambda(P)}[\dot\lambda_P[h]]^\top M_a(s,t)$ is a finite linear combination of càdlàg functions and is therefore càdlàg. Thus $\dot\eta_{a,P}[h]\in\mathcal C_{XY}$.

Linearity follows from the same display, since
$h\mapsto\dot P_{a,P}[h]$ and
$h\mapsto\dot\gamma_{a,\lambda(P)}[\dot\lambda_P[h]]$ are linear and $M_a$ is
fixed. To prove continuity, let $h_n\to h$ in $\mathcal C$. Then
\begin{align*}
\|\dot\eta_{a,P}[h_n]-\dot\eta_{a,P}[h]\|_{\mathcal C_{XY}}
&\le
\|\dot P_{a,P}[h_n]-\dot P_{a,P}[h]\|_{\mathcal C_{XY}} \\
&\quad+
\|\dot\gamma_{a,\lambda(P)}[\dot\lambda_P[h_n]]
-
\dot\gamma_{a,\lambda(P)}[\dot\lambda_P[h]]\|_2
\sup_{v\in\mathcal V_{XY}}\|M_a(v)\|_2 .
\end{align*}
The first term converges to zero by \eqref{eq:projected-lambda-continuous-extension}, and the second difference converges to zero by Lemma~\ref{lemma:gamma_lambda_continuous_extension}. In addition, $ \sup_{v\in\mathcal V_{XY}}\|M_a(v)\|_2 \le \sup_{x\in\mathcal X}\|x\|_2 <\infty, $ because $X$ has bounded support. Therefore $\dot\eta_{a,P}[h_n]\to\dot\eta_{a,P}[h]$ in $\mathcal C_{XY}$.
\end{proof}

\subsubsection{Continuity of the derivative of $\Phi$}
\label{sec:Phi-derivative-continuity}

The preceding lemmas extend the inner derivatives
$\dot\eta_{0,P}$ and $\dot\eta_{1,P}$ to $\mathcal C$. The remaining condition is
estimand-specific: the outer derivative of $\Psi$ must be continuous on the
corresponding arm-specific extension space.

\begin{assumption}[Continuous extension of the derivative of $\Psi$]
\label{assumption:Psi_continuous_extension}
The derivative $\dot\Psi_{\eta_0(P),\eta_1(P)}$ extends to a continuous linear map $ \dot\Psi_{\eta_0(P),\eta_1(P)} : \mathcal C_{XY}\times\mathcal C_{XY} \to \mathbb R^d . $
\end{assumption}
Assumption~\ref{assumption:Psi_continuous_extension} is the only additional estimand-specific condition used to obtain the asymptotically linear representation in Corollary~\ref{thm:Phi_asymptotic_linearity}. This type of continuity condition is standard in functional delta method arguments; it requires that the derivative of the estimand map acts continuously on the perturbations appearing in the limiting expansion. For the estimands considered in this paper, $\Psi$ depends on the arm-specific laws through finite-dimensional summaries, such as evaluations at fixed points or integrals against fixed functions, for which this continuity condition can be checked directly.

\begin{lemma}[Continuous extension of $\dot\Phi_P$]
\label{lemma:Phi_continuous_extension}
Assume the conditions of Lemma~\ref{lemma:eta_continuous_extension} hold for
$a\in\{0,1\}$, and suppose Assumption~\ref{assumption:Psi_continuous_extension}
holds. Then the derivative $\dot\Phi_P$ extends to a
continuous linear map $\dot\Phi_P:\mathcal C\to\mathbb R^d$, given by
\begin{equation}
\label{dot-Phi_P-continuous-extension}
\dot\Phi_P[h]
=
\dot\Psi_{\eta_0(P),\eta_1(P)}
\bigl[
\dot\eta_{0,P}[h],
\dot\eta_{1,P}[h]
\bigr],
\qquad h\in\mathcal C.
\end{equation}
\end{lemma}
\begin{proof}
By Lemma~\ref{lemma:eta_continuous_extension}, $ h \mapsto \bigl(\dot\eta_{0,P}[h],\dot\eta_{1,P}[h]\bigr) $ is continuous and linear from $\mathcal C$ to $\mathcal C_{XY}\times\mathcal C_{XY}$. Composing this map with the continuous linear extension of $\dot\Psi_{\eta_0(P),\eta_1(P)}$ from Assumption~\ref{assumption:Psi_continuous_extension} yields a continuous linear map from $\mathcal C$ to $\mathbb R^d$. On $\mathcal L_P$, this composition agrees with the derivative formula for $\dot\Phi_P$, so it is the desired continuous linear extension.
\end{proof}
The next subsection uses this continuous extension to obtain the asymptotically
linear representation and the corresponding influence-function interpretation.
The ordinary asymptotic normality and bootstrap conclusions follow from the
Hadamard differentiability result in Appendix~\ref{sec:differentiability-Phi}.

\subsection{Asymptotic normality and linearity, and bootstrap validity}
\label{sec:asymptotic-linearity-bootstrap}

The preceding sections establish Hadamard differentiability of the SBW plug-in functional $\Phi$. This is sufficient for the ordinary functional delta method and, separately, for the bootstrap delta method. We first use the functional delta method to obtain the weak limit of $\Phi(P_n)$. We then show that, under Assumption~\ref{assumption:Psi_continuous_extension}, this weak limit also admits an asymptotically linear representation with an influence function. Then we show bootstrap validity for the plug-in estimator $\Phi(P_n)$.

\subsubsection{Functional delta method and influence function}
\label{sec:fdm-and-if}

Let
\begin{equation}
\label{eq:empirical-process}
P_n:=\frac1n\sum_{i=1}^n\delta_{Z_i},
\qquad
\mathbb G_n:=\sqrt n(P_n-P), \nonumber
\end{equation}
where $\delta_z$ denotes the Dirac measure at $z$. Since $\mathcal V$ is a VC class of uniformly bounded indicator functions, it is $P$-Donsker. Hence $\mathbb G_n \rightsquigarrow \mathbb G_P$ in $\ell^\infty(\mathcal V),$ where $\mathbb G_P$ is the $P$-Brownian bridge indexed by $\mathcal V$ \citep[Theorem~2.6.7]{van1996weak}.

The functional-delta-method limit yields weak convergence of the plug-in estimator $\Phi(P_n)$ when $\Phi$ is Hadamard differentiable \citep[Theorem~20.8]{van2000asymptotic}, as established in Theorem~\ref{thm:Phi}.

\begin{theorem}[Asymptotic distribution of plug-in estimator]
\label{thm:Phi_asymptotic_distribution}
Under the conditions of Theorem~\ref{thm:Phi} and letting $\Sigma_{\mathrm{sbw}}
:=
\Var\{\dot\Phi_P[\mathbb G_P]\}$,
\begin{align}
\sqrt n\{\Phi(P_n)-\Phi(P)\}
\rightsquigarrow
\dot\Phi_P[\mathbb G_P]\overset{\mathrm{d}}{=}N_d(0,\Sigma_{\mathrm{sbw}}). \label{eq:asymptotic-normality}
\end{align}
\end{theorem}
\begin{proof}
By Theorem~\ref{thm:Phi}, $\Phi$ is Hadamard differentiable at $P$ tangentially to $\mathcal T_P$. Since $\sqrt n(P_n-P)\rightsquigarrow \mathbb G_P$ in $\ell^\infty(\mathcal V)$, the functional delta method \citep[Theorem~20.8]{van2000asymptotic} yields \eqref{eq:asymptotic-normality}. 
Since $\dot\Phi_P$ is linear and $\mathbb G_P$ is a mean-zero Gaussian process, the limit is $N(0,\sigma_{\mathrm{sbw}}^2)$, with $\sigma_{\mathrm{sbw}}^2=\Var\{\dot\Phi_P(\mathbb G_P)\}$.
\end{proof}
We next derive the stronger asymptotically linear representation. This uses the second conclusion of the functional delta method in \cite{van2000asymptotic}: if the derivative is defined and continuous on a larger space containing the empirical-process directions, then the first-order expansion can be written directly in terms of the derivative applied to $\sqrt n(P_n-P)$. 

Under the conditions of Lemma~\ref{lemma:Phi_continuous_extension}, $\dot\Phi_P$ extends continuously and linearly to $\mathcal C$. We next check that the relevant directions, including the empirical directions, lie in this extension space. First, $\mathbb D_{F_P}\subseteq\mathbb D_S$ implies $\mathcal T_P\subseteq\mathcal C$. Second, for every fixed
$z=(x,a',y)$,
\[
\tau(\delta_z-P)(s,a,t)
=
\mathbf 1\{x\le s,\ a'\le a,\ y\le t\}
-
F_P(s,a,t),
\]
which is càdlàg as a function of $(s,a,t)$; hence
$\delta_z-P\in\mathcal C$. Finally,
\[
\mathbb G_n
=
\frac1{\sqrt n}\sum_{i=1}^n(\delta_{Z_i}-P),
\]
so $\tau(\mathbb G_n)$ is a finite linear combination of càdlàg functions.
Thus $\mathbb G_n\in\mathcal C$, and
$\dot\Phi_P[\mathbb G_n]$ is well defined. This is the sense in which the
continuous extension permits the first-order term to be written as
$\dot\Phi_P\left[\sqrt n(P_n-P)\right]$.

\begin{corollary}[Asymptotic linearity and influence function]
\label{thm:Phi_asymptotic_linearity}
Assume the conditions of Lemma~\ref{lemma:Phi_continuous_extension}. Then
\[
\sqrt n\{\Phi(P_n)-\Phi(P)\}
=
\frac{1}{\sqrt n}\sum_{i=1}^n
\mathrm{IF}_{\Phi}(Z_i;P)
+
o_P(1),
\]
where $ \mathrm{IF}_{\Phi}(z;P) = \dot\Phi_P[\delta_z-P]. $
\end{corollary}
\begin{proof}
Because $\dot\Phi_P$ has a continuous linear extension to $\mathcal C$, the functional delta method \citep[Theorem~20.8]{van2000asymptotic} also gives
\[
\sqrt n\{\Phi(P_n)-\Phi(P)\}
=
\dot\Phi_P\left[\sqrt n(P_n-P)\right]
+
o_P(1).
\]
Using
\[
\sqrt n(P_n-P)
=
\frac{1}{\sqrt n}\sum_{i=1}^n(\delta_{Z_i}-P)
\]
and linearity of the extended derivative, we have
\[
\dot\Phi_P\left[\sqrt n(P_n-P)\right]
=
\frac{1}{\sqrt n}\sum_{i=1}^n
\dot\Phi_P[\delta_{Z_i}-P].
\]
\end{proof}

\subsubsection{Bootstrap validity}

Let $P_n^*$ denote the nonparametric bootstrap empirical distribution, formed by sampling $Z_1^*,\ldots,Z_n^*$ with replacement from
$Z_1,\ldots,Z_n$, and let $\mathbb G_n^*
:=
\sqrt n(P_n^*-P_n).$

\begin{theorem}[Bootstrap validity for $\Phi(P_n)$]
\label{thm:bootstrap_phi}
Assume the conditions of Theorem~\ref{thm:Phi}, so that $\Phi$ is Hadamard differentiable at $P$ tangentially to $\mathcal T_P$. Then
\begin{equation} \label{eq:bootstrap-validity}
    \sqrt n\bigl\{\Phi(P_n^*)-\Phi(P_n)\bigr\}
    \rightsquigarrow
    \dot\Phi_P(\mathbb G_P)
    \qquad
    \text{conditionally on $Z_1,Z_2,\dots$ in probability.} 
\end{equation}
\end{theorem}
In view of Theorem~\ref{thm:Phi_asymptotic_distribution}, the nonparametric bootstrap consistently estimates the limiting distribution of $\sqrt n\bigl\{\Phi(P_n)-\Phi(P)\bigr\}$.

\begin{proof}
As discussed at the beginning of Appendix~\ref{sec:fdm-and-if}, $\mathcal V$ is $P$-Donsker and admits the square-integrable envelope of $1$. Therefore the nonparametric bootstrap empirical process satisfies $\mathbb G_n^*
\rightsquigarrow
\mathbb G_P$ conditionally given $Z_1,\dots,Z_n,$
in probability, as a random element of $\ell^\infty(\mathcal V)$ \citep[Theorem~23.7]{van2000asymptotic}.

By Theorem~\ref{thm:Phi}, $\Phi$ is Hadamard differentiable at $P$ tangentially to $\mathcal T_P$. Hence the bootstrap delta method
\citep[Theorem~23.9]{van2000asymptotic} applies, and \eqref{eq:bootstrap-validity} follows. Thus the conditional bootstrap law converges to the same limiting law as the original centered statistic.
\end{proof}
Theorem~\ref{thm:bootstrap_phi} justifies using the bootstrap for distributional approximation and variance estimation. In practice, if
$\Phi(P_n^{*(1)}),\ldots,\Phi(P_n^{*(B)})$ are bootstrap replicates, we estimate the sampling variance of $\Phi(P_n)$ by the empirical variance of these bootstrap estimates. Equivalently, multiplying this quantity by $n$ estimates the asymptotic variance of
$\sqrt n\{\Phi(P_n)-\Phi(P)\}$. Under the conditions of the theorem, the bootstrap variance estimator consistently estimates the limiting variance
$\mathrm{Var}\{\dot\Phi_P(\mathbb G_P)\}$, which agrees with
$\mathrm{Var}\{\mathrm{IF}_\Phi(Z;P)\}$ when the influence-function representation is available.

The same conditional distributional consistency also justifies percentile bootstrap intervals, provided the limiting distribution has a continuous distribution function at the relevant quantiles. Thus, for scalar estimands, one may form either a Wald-type interval using the bootstrap standard error or a percentile interval using the empirical quantiles of the bootstrap replicates. When inference is performed on a transformed scale, such as the log scale for ratio estimands, the bootstrap interval is constructed on that scale and then mapped back to the original scale.

\subsection{Variance reduction}\label{sec:variance_reduction}

We now compare the (Hampel) influence functions of the functionals used to define  the SBW and unadjusted plug-in estimators \citep{hampel1974influence}. For the SBW functional, define
\[
\theta_{\mathrm{sbw}}(z)
:=
\dot\Phi_P[\delta_z-P].
\]
Let $\theta_{\mathrm{unadj}}$ denote the corresponding influence function for the unadjusted plug-in functional, defined formally in \eqref{eq:influence-functions}. We show that, for a scalar-valued estimand $\Psi(P_0,P_1)$, the variance of $\theta_{\mathrm{sbw}}(Z)$ is weakly smaller than the variance of $\theta_{\mathrm{unadj}}(Z)$. The vector-valued case can be handled coordinatewise or, equivalently, by replacing variances with covariance matrices. Assume throughout this subsection that treatment is randomized, and define $c:=\mathbb E_P[X],$ and $B:=\mathbb E_P[XX^\top].$

\begin{assumption}[Basic variance comparison conditions]
\label{assumption:variance_basic_conditions}
Suppose that $B=\mathbb E_P[XX^\top]$ is nonsingular and $X$ includes an intercept as its first coordinate.
\end{assumption}
The next lemma simplifies the derivative of the population SBW coefficient. This finite-dimensional calculation will be useful in the variance comparison of Lemma~\ref{lem:derivative-level-variance-reduction}.

\begin{lemma}[Derivative of the population SBW coefficient under randomization]
\label{lemma: derivative of SBW coefficient under randomization}
Under Assumption~\ref{assumption:variance_basic_conditions}, for
$a\in\{0,1\}$ and $z=(x,a',y)$,
\[
\dot\gamma_{a,\lambda(P)}[\dot\lambda_P[\delta_z-P]]
=
B^{-1}
\left[
\left(
\frac{\mathbf 1\{a'=a\}}{\pi_a}
-
1
\right)(x-c)
\right].
\]
\end{lemma}
\begin{proof}
Apply Proposition~\ref{prop:gamma} with
$(P_0,P_1,\pi) = \lambda(P)$. Under randomization, the quantities appearing in
that derivative simplify as follows: $B_a=B$, $\gamma_a\{\lambda(P)\}=0$, and
$\mathbb E_{P_1}[X]=\mathbb E_{P_0}[X]=c$. Hence, for any perturbation $h$,
\begin{equation}
\label{eq:gamma-rand-simplified}
\dot\gamma_{a,\lambda(P)}[\dot\lambda_P[h]]
=
B^{-1}
\left\{
\int x\,d\{\dot P_{a,P}[h]\}
-
\pi_0\int x\,d\{\dot P_{0,P}[h]\}
-
\pi_1\int x\,d\{\dot P_{1,P}[h]\}
\right\}.
\end{equation}
Now, take $h=\delta_z-P$, where $z=(x,a',y)$. For indicator functions
$v\in\mathcal V_{XY}$, the derivative formula for the conditional law gives
\[
\dot P_{a,P}[\delta_z-P](v)
=
\frac{\mathbf 1\{a'=a\}}{\pi_a}
\{v(x,y)-P_av\}.
\]
The same calculation applies to any bounded-HK function
$m:\mathcal Z_{XY}\to\mathbb R$, giving
\[
\int m\,d\{\dot P_{a,P}[\delta_z-P]\}
=
\frac{\mathbf 1\{a'=a\}}{\pi_a}
\{m(x,y)-\mathbb E_{P_a}[m(X,Y)]\}.
\]
Applying this identity to each coordinate function $m_j(x,y)=x_j$ gives
\[
\int x\,d\{\dot P_{a,P}[\delta_z-P]\}
=
\frac{\mathbf 1\{a'=a\}}{\pi_a}
\{x-\mathbb E_{P_a}[X]\}
=
\frac{\mathbf 1\{a'=a\}}{\pi_a}
(x-c),
\]
where the last equality uses randomization. Substituting the above display for each arm into \eqref{eq:gamma-rand-simplified} and simplifying yields the desired result.
\end{proof}

\medskip

\begin{assumption}[Arm-specific derivative decomposition]
\label{assumption:arm_specific_derivative_decomposition}
For the directions appearing in this subsection, suppose that the derivative of $\Psi$ admits arm-specific linear components $\dot\Psi_0$ and $\dot\Psi_1$ such that, for each paired direction $(h_0,h_1)$ used below,
\[
\dot\Psi_{\eta_0(P),\eta_1(P)}[h_0,h_1]
=
\dot\Psi_0[h_0]+\dot\Psi_1[h_1].
\]
In particular, this decomposition holds for the paired directions generated by $\dot P_{a,P}[\delta_z-P]$, $\dot\eta_{a,P}[\delta_z-P]$,
$\dot\eta_{a,P}^{(2)}[\delta_z-P]$, $z\in\mathcal Z$, and for the
directions used to define $\beta_{\Psi,a}$ in \eqref{eq:beta-Psi-action}.
\end{assumption}
When Assumption~\ref{assumption:Psi_continuous_extension} holds, the arm-specific maps in Assumption~\ref{assumption:arm_specific_derivative_decomposition} are the partial derivatives of the continuous extension of $\dot\Psi$ at $(P_0,P_1)$.

Under Assumption~\ref{assumption:arm_specific_derivative_decomposition}, define the unadjusted and SBW influence functions by
\begin{align}
\label{eq:influence-functions}
\theta_{\mathrm{unadj}}(z)
&:=
\dot\Psi_0\!\left[\dot P_{0,P}[\delta_z-P]\right]
+
\dot\Psi_1\!\left[\dot P_{1,P}[\delta_z-P]\right],
 \nonumber \\
\theta_{\mathrm{sbw}}(z)
&:=
\dot\Psi_0\!\left[\dot\eta_{0,P}[\delta_z-P]\right]
+
\dot\Psi_1\!\left[\dot\eta_{1,P}[\delta_z-P]\right].
\end{align}
The second display is the chain-rule expansion of
$\dot\Phi_P[\delta_z-P]$.

For $z=(x,a',y)$, define the covariate-imbalance space
\begin{equation}
\label{eq:H-covariate-imbalance-space}
\mathcal H
:=
\left\{
z\mapsto (a'-\pi_1)\alpha^\top x:
\alpha\in\mathbb R^{q+1}
\right\}.
\end{equation}
We will show that $\theta_{\mathrm{sbw}}$ differs from
$\theta_{\mathrm{unadj}}$ by an element of $\mathcal H$, and that
$\theta_{\mathrm{sbw}}$ is orthogonal to $\mathcal H$ in $L_2(P)$.

For later use, define $\beta_{\Psi,a}\in\mathbb R^{q+1}$ by
\[
\beta_{\Psi,a,j}
:=
\dot\Psi_a\left[
v\mapsto
e_j^\top B^{-1}\int x v(x,y)\,P_a(dx,dy)
\right],
\qquad j=1,\ldots,q+1,
\]
where $e_1,\ldots,e_{q+1}$ are the standard basis vectors in
$\mathbb R^{q+1}$. By linearity of $\dot\Psi_a$, for any
$u\in\mathbb R^{q+1}$,
\begin{equation}
\label{eq:beta-Psi-action}
\dot\Psi_a\left[
v\mapsto
u^\top B^{-1}\int x v(x,y)\,P_a(dx,dy)
\right]
=
u^\top\beta_{\Psi,a}.
\end{equation}

\begin{lemma}[SBW influence function correction]
\label{lemma:sbw-derivative-level-correction}
Under Assumptions~\ref{assumption:variance_basic_conditions}
and~\ref{assumption:arm_specific_derivative_decomposition}, there exists
$\tilde\beta_\Psi\in\mathbb R^{q+1}$ such that
\begin{equation}
\label{eq:theta-sbw-unadj-decomposition}
\theta_{\mathrm{sbw}}(z)
=
\theta_{\mathrm{unadj}}(z)
+
(a'-\pi_1)(x-c)^\top\tilde\beta_\Psi .
\end{equation}
Consequently,
$\theta_{\mathrm{sbw}}-\theta_{\mathrm{unadj}}\in\mathcal H$.
\end{lemma}
\begin{proof}
As shown in \eqref{eq:eta-derivative-extension}, under randomization, the derivative formula for $\eta_a$ gives
\[
\dot\eta_{a,P}[h]
=
\dot P_{a,P}[h]
+
\dot\eta_{a,P}^{(2)}[h],
\qquad
\dot\eta_{a,P}^{(2)}[h](v)
:=
-
\dot\gamma_{a,\lambda(P)}[\dot\lambda_P[h]]^\top
\int x v(x,y)\,P_a(dx,dy).
\]
Thus the first summand in $\dot\eta_{a,P}[h]$ is exactly the unadjusted
arm-specific derivative $\dot P_{a,P}[h]$. Using
\eqref{eq:influence-functions}, we therefore have
\begin{equation}
\label{eq:theta-second-piece-decomposition}
\theta_{\mathrm{sbw}}(z)
=
\theta_{\mathrm{unadj}}(z)
+
\dot\Psi_0\!\left[
\dot\eta_{0,P}^{(2)}[\delta_z-P]
\right]
+
\dot\Psi_1\!\left[
\dot\eta_{1,P}^{(2)}[\delta_z-P]
\right].
\end{equation}

Now, fix $z=(x,a',y)$. By
Lemma~\ref{lemma: derivative of SBW coefficient under randomization},
\[
\dot\eta_{a,P}^{(2)}[\delta_z-P](v)
=
-
\left(\frac{\mathbf 1\{a'=a\}}{\pi_a}-1\right)
(x-c)^\top B^{-1}
\int x v(x,y)\,P_a(dx,dy).
\]
Applying \eqref{eq:beta-Psi-action} gives, for each $a\in\{0,1\}$,
\begin{equation}
\label{eq:eta-second-piece-Psi-action}
\dot\Psi_a\!\left[
\dot\eta_{a,P}^{(2)}[\delta_z-P]
\right]
=
-
\left(\frac{\mathbf 1\{a'=a\}}{\pi_a}-1\right)
(x-c)^\top\beta_{\Psi,a}.
\end{equation}
Since $\mathbf 1\{a'=0\}=1-a'$ and $\pi_0=1-\pi_1$, we have 
\[
-
\left(\frac{\mathbf 1\{a'=0\}}{\pi_0}-1\right)
=
\frac{a'-\pi_1}{\pi_0}, \qquad -
\left(\frac{\mathbf 1\{a'=1\}}{\pi_1}-1\right)
=
-\frac{a'-\pi_1}{\pi_1}.
\]
Substituting the above display into
\eqref{eq:theta-second-piece-decomposition} yields
\[
\theta_{\mathrm{sbw}}(z)
=
\theta_{\mathrm{unadj}}(z)
+
\frac{a'-\pi_1}{\pi_0}(x-c)^\top\beta_{\Psi,0}
-
\frac{a'-\pi_1}{\pi_1}(x-c)^\top\beta_{\Psi,1}.
\]
Thus \eqref{eq:theta-sbw-unadj-decomposition} holds with
\[
\tilde\beta_\Psi
:=
\frac{\beta_{\Psi,0}}{\pi_0}
-
\frac{\beta_{\Psi,1}}{\pi_1}.
\]

Finally, because $X$ includes an intercept, $e_1^\top X=1$ and
$e_1^\top c=1$. Hence
\[
(X-c)^\top\tilde\beta_\Psi
=
X^\top\{\tilde\beta_\Psi-(c^\top\tilde\beta_\Psi)e_1\},
\]
meaning that the centered function is also an uncentered linear function of $X$, which yields the result $\theta_{\mathrm{sbw}}-\theta_{\mathrm{unadj}}\in\mathcal H$.
\end{proof}
The representation in Lemma~\ref{lemma:sbw-derivative-level-correction} shows that the SBW influence function differs from the unadjusted influence function by an element of the covariate-imbalance space. The next lemma shows that this correction is orthogonal to the SBW influence function, yielding the variance comparison.

\begin{lemma}[Variance reduction]
\label{lem:derivative-level-variance-reduction}
Under Assumptions~\ref{assumption:variance_basic_conditions} and 
\ref{assumption:arm_specific_derivative_decomposition},
$\mathbb E_P\{\theta_{\mathrm{sbw}}(Z)h(Z)\}=0$ for every $h\in\mathcal H$.
Consequently,
\[
\Var\{\theta_{\mathrm{sbw}}(Z)\}
\le
\Var\{\theta_{\mathrm{unadj}}(Z)\}.
\]
\end{lemma}
\begin{proof}
Fix $\alpha\in\mathbb R^{q+1}$ and define $\omega_\alpha(Z):=(A-\pi_1)\alpha^\top X.$ Since every element of $\mathcal H$ can be written as $\omega_\alpha$ for some $\alpha\in\mathbb R^{q+1}$, it suffices to show that $ \mathbb E_P\{\theta_{\mathrm{sbw}}(Z)\omega_\alpha(Z)\}=0. $ By Lemma~\ref{lemma:sbw-derivative-level-correction},
\begin{align}
\label{eq:orthogonality-decomposition}
\mathbb E_P\{\theta_{\mathrm{sbw}}(Z)\omega_\alpha(Z)\}
&=
\mathbb E_P\{\theta_{\mathrm{unadj}}(Z)\omega_\alpha(Z)\} 
+
\mathbb E_P\left[
(A-\pi_1)(X-c)^\top\tilde\beta_\Psi\,
\omega_\alpha(Z)
\right].
\end{align}
We evaluate the two terms separately. First, under randomization,
$\mathbb E_P\{\omega_\alpha(Z)\}=0$. Define
$\omega_\alpha P\in\ell^\infty(\mathcal V)$ by
\[
(\omega_\alpha P)(v)
:=
\int \omega_\alpha(z)v(z)\,P(dz),
\qquad v\in\mathcal V.
\]
Then, for each $v\in\mathcal V$,
\[
\mathbb E_P\{\omega_\alpha(Z)(\delta_Z-P)(v)\}
=
\mathbb E_P\{\omega_\alpha(Z)v(Z)\}
-
P(v)\mathbb E_P\{\omega_\alpha(Z)\}
=
(\omega_\alpha P)(v).
\]
Thus, as an element of $\ell^\infty(\mathcal V)$,
\[
\mathbb E_P\{\omega_\alpha(Z)(\delta_Z-P)\}
=
\omega_\alpha P.
\]

We next justify applying the derivative maps after taking this expectation. Since $\mathcal P$ contains the point masses, $\delta_z-P\in\mathcal T_P$ for each $z\in\mathcal Z$. Hence, for each fixed $z$, $\omega_\alpha(z)(\delta_z-P)$ belongs to $\mathcal L_P$. The map
\[
z\mapsto \omega_\alpha(z)(\delta_z-P)
\]
is bounded as an $\ell^\infty(\mathcal V)$-valued map, and $\mathcal L_P$ is
closed. Therefore its expectation, $\omega_\alpha P$, also belongs to
$\mathcal L_P$. The composite maps
\[
T_a:h\mapsto \dot\Psi_a[\dot P_{a,P}[h]],
\qquad a\in\{0,1\},
\]
are continuous and linear on the relevant closed linear span of tangent
directions. Consequently,
\[
\mathbb E_P\{T_a[\omega_\alpha(Z)(\delta_Z-P)]\}
=
T_a\left[
\mathbb E_P\{\omega_\alpha(Z)(\delta_Z-P)\}
\right],
\qquad a\in\{0,1\}.
\]
Using the definition of $\theta_{\mathrm{unadj}}$ and linearity of
$\dot P_{a,P}$ and $\dot\Psi_a$, we get
\begin{align}
\label{eq:unadj-orthogonality-term}
\mathbb E_P\{\theta_{\mathrm{unadj}}(Z)\omega_\alpha(Z)\}
&=
\mathbb E_P\left[
\omega_\alpha(Z)
\left\{
\dot\Psi_0[\dot P_{0,P}[\delta_Z-P]]
+
\dot\Psi_1[\dot P_{1,P}[\delta_Z-P]]
\right\}
\right] \nonumber \\
&=
\dot\Psi_0\!\left[
\dot P_{0,P}\left\{
\mathbb E_P[\omega_\alpha(Z)(\delta_Z-P)]
\right\}
\right]
+
\dot\Psi_1\!\left[
\dot P_{1,P}\left\{
\mathbb E_P[\omega_\alpha(Z)(\delta_Z-P)]
\right\}
\right] \nonumber \\
&=
\dot\Psi_0[\dot P_{0,P}\{\omega_\alpha P\}]
+
\dot\Psi_1[\dot P_{1,P}\{\omega_\alpha P\}].
\end{align}
The second equality follows from the continuous linearity of the composite maps $T_a:h\mapsto\dot\Psi_a[\dot P_{a,P}[h]]$ on the closed linear span of the relevant tangent directions, as justified above.

Next define $\bar\alpha:=\alpha-(\alpha^\top c)e_1.$ Because $X$ includes an intercept, $\bar\alpha^\top X=\alpha^\top(X-c)$. We now compute the two conditional-law derivatives appearing in \eqref{eq:unadj-orthogonality-term}. For arm 1, the derivative of the conditional law gives
\[
\dot P_{1,P}[h](v)
=
\frac{h(Av)}{\pi_1}
-
\frac{\mathbb E_P[Av(X,Y)]}{\pi_1^2}h(A).
\]
With $h=\omega_\alpha P$, and using
$A(A-\pi_1)=A\pi_0$ and randomization, this becomes
\begin{align*}
\dot P_{1,P}[\omega_\alpha P](v)
&=
\frac{1}{\pi_1}
\mathbb E_P[A\,v(X,Y)(A-\pi_1)\alpha^\top X]
-
\frac{\mathbb E_P[Av(X,Y)]}{\pi_1^2}
\mathbb E_P[A(A-\pi_1)\alpha^\top X] \\
&=
\pi_0\mathbb E_{P_1}\{v(X,Y)\alpha^\top X\}
-
\pi_0\mathbb E_{P_1}\{v(X,Y)\}\alpha^\top c \\
&=
\pi_0\mathbb E_{P_1}\{v(X,Y)\alpha^\top(X-c)\} \\
&=
\pi_0\,\bar\alpha^\top
\int x v(x,y)\,P_1(dx,dy).
\end{align*}
An analogous calculation for arm 0 gives $ \dot P_{0,P}[\omega_\alpha P](v) = -\pi_1\,\bar\alpha^\top \int x v(x,y)\,P_0(dx,dy) $. Therefore, by \eqref{eq:beta-Psi-action},
\[
\dot\Psi_1[\dot P_{1,P}\{\omega_\alpha P\}]
=
\pi_0\bar\alpha^\top B\beta_{\Psi,1},
\qquad
\dot\Psi_0[\dot P_{0,P}\{\omega_\alpha P\}]
=
-\pi_1\bar\alpha^\top B\beta_{\Psi,0}.
\]
Substituting these two displays into
\eqref{eq:unadj-orthogonality-term} gives
\begin{equation}
\label{eq:unadj-orthogonality-eval}
\mathbb E_P\{\theta_{\mathrm{unadj}}(Z)\omega_\alpha(Z)\}
=
\pi_0\bar\alpha^\top B\beta_{\Psi,1}
-
\pi_1\bar\alpha^\top B\beta_{\Psi,0}.
\end{equation}

We now evaluate the second term in
\eqref{eq:orthogonality-decomposition}. We use the fact that 
$\mathbb E_P[(A-\pi_1)^2]=\pi_1\pi_0$ and $\mathbb E_P\{(X-c)^\top\tilde\beta_\Psi\,\alpha^\top X\}
=
\bar\alpha^\top B\tilde\beta_\Psi$ to show that
\begin{align}
\label{eq:correction-orthogonality-eval}
\mathbb E_P\left[
(A-\pi_1)(X-c)^\top\tilde\beta_\Psi\,
\omega_\alpha(Z)
\right]
&=
\pi_1\pi_0\,\bar\alpha^\top B\tilde\beta_\Psi =
\pi_1\bar\alpha^\top B\beta_{\Psi,0}
-
\pi_0\bar\alpha^\top B\beta_{\Psi,1}.
\end{align}
Therefore, \eqref{eq:unadj-orthogonality-eval} and
\eqref{eq:correction-orthogonality-eval} cancel in
\eqref{eq:orthogonality-decomposition}, yielding
\[
\mathbb E_P\{\theta_{\mathrm{sbw}}(Z)\omega_\alpha(Z)\}=0.
\]
Since $\alpha$ was arbitrary, this proves $\mathbb E_P\{\theta_{\mathrm{sbw}}(Z)h(Z)\}=0$ for every $h\in\mathcal H$.

Finally, Lemma~\ref{lemma:sbw-derivative-level-correction} shows that
$\theta_{\mathrm{unadj}}-\theta_{\mathrm{sbw}}\in\mathcal H$, so this orthogonality result applied to
$h=\theta_{\mathrm{unadj}}-\theta_{\mathrm{sbw}}$ gives 
\[
\mathbb E_P\left[
\theta_{\mathrm{sbw}}(Z)
\{\theta_{\mathrm{unadj}}(Z)-\theta_{\mathrm{sbw}}(Z)\}
\right]=0.
\]
The functions $\theta_{\mathrm{sbw}}$ and $\theta_{\mathrm{unadj}}$ are
mean-zero because they are linear derivatives applied to the centered direction
$\delta_Z-P$. Hence the $L_2(P)$ orthogonal decomposition above gives
\[
\Var\{\theta_{\mathrm{unadj}}(Z)\}
=
\Var\{\theta_{\mathrm{sbw}}(Z)\}
+
\Var\{\theta_{\mathrm{unadj}}(Z)-\theta_{\mathrm{sbw}}(Z)\},
\]
and the variance inequality follows.
\end{proof}
This decomposition also shows that the variance inequality is strict whenever
\[
\Var\{\theta_{\mathrm{unadj}}(Z)-\theta_{\mathrm{sbw}}(Z)\}
=
\Var\{(A-\pi_1)(X-c)^\top\tilde\beta_\Psi\}
>0.
\]
Thus, strict variance reduction occurs when $X$ has a nonzero linear association with the arm-specific components of the estimand.

\subsection{Equivalence with the constrained SBW estimator}
\label{sec:equivalence-constrained-sbw}

The asymptotic theory developed so far is based on a closed-form representation of the SBWs, corresponding to the solution of \eqref{eq:sbw_primal} without the nonnegativity constraint. In particular, the plug-in estimator $\Phi(P_n)$ equals the closed-form SBW estimator $\tilde\psi_n^{\mathrm{\,cf}}$. In practice, however, the weights are computed by solving the quadratic program \eqref{eq:sbw_primal}, which additionally imposes a nonnegativity constraint.

We now show that these two formulations coincide with probability tending to one. We establish that the closed-form weights are strictly positive with probability tending to one, implying that the nonnegativity constraint is asymptotically inactive. This allows us to conclude that the estimator analyzed in our theory coincides with the practical SBW estimator. We then use this result to prove the main asymptotic results stated in Section~\ref{section:theory} of the main text.

\subsubsection{Closed-form representation and its properties}

Let $N_a := \sum_{i=1}^n \mathbbm{1}\{A_i = a\}$ denote the number of units in arm $a$, let $\mathbf 1_a$ denote the $N_a$-vector of ones, and let $\mathbf X_a \in \mathbb{R}^{N_a \times (q+1)}$ denote the covariate matrix in arm $a$, including an intercept column. Let $\bar X_n := \frac{1}{n}\sum_{i=1}^n X_i$ denote the full-sample mean of the covariate vector. The unconstrained, or closed-form, SBWs in arm $a$ can be written as
\begin{align}
\tilde w_a^{\mathrm{cf}}
&=
\mathbf 1_a-\mathbf X_a \tilde\beta_a^{\,\mathrm{cf}}\ \ \textnormal{ with }\ \ \tilde\beta_a^{\,\mathrm{cf}}
:=
(\mathbf X_a^\top \mathbf X_a)^{-1}
\bigl(\mathbf X_a^\top \mathbf 1_a - N_a \bar X_n\bigr).
\label{eq:unconstrained_weights_appendix}
\end{align}
Equivalently, if we define $ \tilde h_a^{\mathrm{cf}}(x) := 1 - x^\top \tilde\beta_a^{\,\mathrm{cf}}$, then for each unit $i$ in arm $a$, $\tilde w_{a,i}^{\mathrm{cf}} = \tilde h_a^{\mathrm{cf}}(X_i).$

Because $\tilde w_{a,i}^{\mathrm{cf}}=1-X_i^\top\tilde\beta_a^{\,\mathrm{cf}}$, the data are identically distributed, and $X_i$ is a.s. bounded, positivity of all $n$ closed-form weights follows once $\tilde\beta_a^{\,\mathrm{cf}}$ is small. 

\begin{lemma}[Convergence of the closed-form coefficient]
\label{lemma:beta_to_0}
Assume that $B:=\mathbb E_P[XX^\top]$ is nonsingular. Then $\tilde\beta_a^{\,\mathrm{cf}} \xrightarrow{P} 0$.
\end{lemma}
\begin{proof}
Write
\[
M_{n,a}:=\frac{1}{N_a}\mathbf X_a^\top\mathbf X_a,
\qquad
D_{n,a}:=\frac{1}{N_a}\mathbf X_a^\top\mathbf 1_a-\bar X_n,
\]
so that $\tilde\beta_a^{\,\mathrm{cf}}=M_{n,a}^{-1}D_{n,a}$. By the law of large numbers and
the continuous mapping theorem,
\[
M_{n,a}
=
\frac{n^{-1}\sum_{i=1}^n \mathbbm 1\{A_i=a\}X_iX_i^\top}
     {n^{-1}\sum_{i=1}^n \mathbbm 1\{A_i=a\}}
\xrightarrow{P}
\frac{\mathbb E_P[\mathbbm 1\{A=a\}XX^\top]}{\pi_a}
=
\mathbb E_P[XX^\top]
=
B,
\]
where the penultimate equality uses randomization. Since $B$ is nonsingular,
$M_{n,a}^{-1}\xrightarrow{P}B^{-1}$. Similarly,
\[
\frac{1}{N_a}\mathbf X_a^\top \mathbf 1_a
=
\frac{n^{-1}\sum_{i=1}^n\mathbbm 1\{A_i=a\}X_i}
     {n^{-1}\sum_{i=1}^n\mathbbm 1\{A_i=a\}}
\xrightarrow{P}
\frac{\mathbb E_P[\mathbbm 1\{A=a\}X]}{\pi_a}
=
\mathbb E_P[X].
\]
Also, $\bar X_n\xrightarrow{P}\mathbb E_P[X]$, so
$D_{n,a}\xrightarrow{P}0$. Slutsky's theorem gives $\tilde\beta_a^{\,\mathrm{cf}}=M_{n,a}^{-1}D_{n,a}\xrightarrow{P}0.$
\end{proof}
The next lemma converts this vector convergence into a uniform bound over the support of $X$, which will allow us to control the minimum weight.

\begin{lemma}[Asymptotic positivity of closed-form weights]
\label{lemma:uniform_and_positivity}
Under the conditions of Lemma~\ref{lemma:beta_to_0},
\[
P\!\left(\min_{1\le i \le N_a} \tilde w_{a,i}^{\mathrm{cf}} > 0\right) \to 1.
\]
Consequently, the nonnegativity constraint in \eqref{eq:sbw_primal} is asymptotically inactive.
\end{lemma}
\begin{proof}
Let $C$ be a finite constant such that $P(\|X\|_2 \le C)=1$, which necessarily exists since $\mathcal X$ was assumed to be bounded throughout the paper.
For any $1\le j\le N_a$,
\begin{align*}
    1&= \tilde h_a^{\mathrm{cf}}(X_j) + [1-\tilde h_a^{\mathrm{cf}}(X_j)]\le \tilde h_a^{\mathrm{cf}}(X_j) + \max_{1\le i\le N_a} [1-\tilde h_a^{\mathrm{cf}}(X_i)].
\end{align*}
This is true, in particular, for $j=\argmin_{1\le i\le N_a} \tilde h_a^{\mathrm{cf}}(X_i)$, and so
\begin{align*}
    \min_{1\le i\le N_a} \tilde h_a^{\mathrm{cf}}(X_i)\ge 1 - \max_{1\le i\le N_a} [1-\tilde h_a^{\mathrm{cf}}(X_i)].
\end{align*}
Finally, note that
\begin{align*}
    \max_{1\le i\le N_a}[1-\tilde h_a^{\mathrm{cf}}(X_i)]
=
\max_{1\le i\le N_a}[X_i^\top\tilde\beta_a^{\,\mathrm{cf}}]
\le
\max_{1\le i\le N_a}\|X_i\|_2\|\tilde\beta_a^{\,\mathrm{cf}}\|_2
\le
C\|\tilde\beta_a^{\,\mathrm{cf}}\|_2
\xrightarrow{P}0.
\end{align*}
Combining the preceding two displays gives the result.
\end{proof}

\subsubsection{Equivalence with the quadratic-program estimator}

We now relate the closed-form weights to the constrained quadratic-program solution. When the closed-form weights are strictly positive, they are feasible for the constrained problem; since they already minimize the objective over the larger equality-constrained set, they also solve the constrained problem.

For $a\in\{0,1\}$, define the objective $Q_a(w) := \|w - \mathbf 1_a\|_2^2$ and the feasible sets $\mathcal W_a^{\mathrm{cf}}
:=
\{w \in \mathbb R^{N_a} : \mathbf X_a^\top w = N_a \bar X_n\}$ and $\mathcal W_a^{\mathrm{sbw}}
:=
\{w \in \mathbb R^{N_a} : \mathbf X_a^\top w = N_a \bar X_n,\; w \ge 0\}.$ Note that $\mathcal W_a^{\mathrm{sbw}} \subseteq \mathcal W_a^{\mathrm{cf}}$.

Let $\tilde w_a^{\mathrm{cf}}$ denote the unique minimizer of $Q_a$ over $\mathcal W_a^{\mathrm{cf}}$, given by the closed-form expression \eqref{eq:unconstrained_weights_appendix}, and let $\hat w_a^{\mathrm{sbw}}$ denote the unique minimizer of $Q_a$ over $\mathcal W_a^{\mathrm{sbw}}$, i.e., the solution to the quadratic program \eqref{eq:sbw_primal}. Our final result tells us that the practical SBW estimator computed from the quadratic program \eqref{eq:sbw_primal} coincides with the plug-in estimator $\Phi(P_n)$ with probability tending to one.

\begin{proposition}[Equivalence of constrained and closed-form SBW estimators]
\label{prop:equivalence-constrained-sbw}
Under the assumptions of Lemma~\ref{lemma:uniform_and_positivity}, for each $a\in\{0,1\}$, the $N_a$-dimensional weight vectors $\hat w_a^{\mathrm{sbw}}$ and $\tilde w_a^{\mathrm{cf}}$ coincide with probability tending to $1$ as $n\to\infty$. Consequently, 
\[
\widehat{\psi}^{\,\mathrm{sbw}}=\tilde\psi_n^{\mathrm{\,cf}}\ \ \textnormal{ with probability tending to $1$.}
\]
\end{proposition}
\begin{proof}
Fix $a\in\{0,1\}$. By Lemma~\ref{lemma:uniform_and_positivity}, $ P^n\!\left(\min_{1\le i\le N_a}\tilde w_{a,i}^{\mathrm{cf}}>0\right)\to1. $ On this event, $\tilde w_a^{\mathrm{cf}}\in\mathcal W_a^{\mathrm{sbw}}$. Since $\mathcal W_a^{\mathrm{sbw}}\subseteq\mathcal W_a^{\mathrm{cf}}$ and $\tilde w_a^{\mathrm{cf}}$ minimizes $Q_a$ over $\mathcal W_a^{\mathrm{cf}}$, it also minimizes $Q_a$ over $\mathcal W_a^{\mathrm{sbw}}$. By uniqueness of the constrained minimizer, $\hat w_a^{\mathrm{sbw}}=\tilde w_a^{\mathrm{cf}}$ on this event. Applying the same argument to both treatment arms, we have
\[
P^n(\hat w_0^{\mathrm{sbw}}=\tilde w_0^{\mathrm{cf}},\ \hat w_1^{\mathrm{sbw}}=\tilde w_1^{\mathrm{cf}})\to1.
\]
Since $\Phi(P_n)=\tilde\psi_n^{\mathrm{\,cf}}$ uses the closed-form weights
$\tilde w_a^{\mathrm{cf}}$, whereas $\widehat{\psi}^{\,\mathrm{sbw}}$ uses the constrained
weights $\hat w_a^{\mathrm{sbw}}$, it follows that
\[
P^n(\widehat{\psi}^{\,\mathrm{sbw}}=\tilde\psi_n^{\mathrm{\,cf}})\to1.
\]
\end{proof}

\subsubsection{Proofs of the main asymptotic results}

We first transfer the asymptotic distribution and bootstrap consistency from the
closed-form plug-in estimator analyzed above to the practical estimator computed
by the quadratic program. 

\begin{theorem}[Multivariate asymptotic normality and bootstrap consistency]
\label{thm:asymp_boot_multivariate}
Let $\widehat{\psi}^{\,\mathrm{sbw}}$ denote the $d$-dimensional version of the SBW estimator computed from Algorithm~\ref{alg:gsbw}, and let $\psi=\Psi(P_0,P_1)\in\mathbb R^d$. Under the regularity conditions of Theorem~\ref{thm:Phi_asymptotic_distribution},
there exists a finite covariance matrix $\Sigma_{\mathrm{sbw}}$ such that
\[
\sqrt n(\widehat{\psi}^{\,\mathrm{sbw}}-\psi)
\rightsquigarrow
N(0,\Sigma_{\mathrm{sbw}}).
\]
Furthermore, the nonparametric bootstrap consistently estimates this limiting
distribution.
\end{theorem}
\begin{proof}
By construction, $\Phi(P)=\psi$ and
$\Phi(P_n)=\tilde\psi_n^{\mathrm{\,cf}}$, where
$\tilde\psi_n^{\mathrm{\,cf}}$ denotes the closed-form SBW plug-in estimator.
Therefore, Theorem~\ref{thm:Phi_asymptotic_distribution} gives
\[
\sqrt n(\tilde\psi_n^{\mathrm{\,cf}}-\psi)
\rightsquigarrow
\dot\Phi_P(\mathbb G_P)
\sim N(0,\Sigma_{\mathrm{sbw}}),
\qquad
\Sigma_{\mathrm{sbw}}
:=
\Var\{\dot\Phi_P(\mathbb G_P)\}.
\]
Since $\widehat{\psi}^{\,\mathrm{sbw}}=\tilde\psi_n^{\mathrm{\,cf}}$ with probability tending
to one by Proposition~\ref{prop:equivalence-constrained-sbw}, the same weak
limit holds for $\widehat{\psi}^{\,\mathrm{sbw}}$.

Similarly, Theorem~\ref{thm:bootstrap_phi} gives bootstrap consistency for the closed-form plug-in estimator $\tilde\psi_n^{\mathrm{\,cf}}=\Phi(P_n)$. The same arguments used in Appendix~\ref{sec:equivalence-constrained-sbw} to show equivalence of the constrained and closed-form estimators apply conditionally to the bootstrap versions, since the bootstrap empirical means satisfy the corresponding conditional convergence statements. Thus the bootstrap constrained estimator and bootstrap closed-form estimator agree with conditional probability tending to one, and bootstrap consistency transfers to the practical estimator $\widehat{\psi}^{\,\mathrm{sbw}}$.
\end{proof}
\begin{proof}[Proof of Theorem~\ref{thm:asymp_boot}]
Theorem~\ref{thm:asymp_boot} is the scalar case $d=1$ of
Theorem~\ref{thm:asymp_boot_multivariate}, with
$\sigma_{\mathrm{sbw}}^2=\Sigma_{\mathrm{sbw}}$.
\end{proof}

\bigskip
We next prove the asymptotic variance reduction statement in the main text.

\begin{proof}[Proof of Theorem~\ref{thm:variance_reduction}]
Proposition~\ref{prop:equivalence-constrained-sbw} implies that the practical SBW estimator computed from Algorithm~\ref{alg:gsbw} and the closed-form SBW plug-in estimator have the same first-order asymptotic distribution and the same asymptotic variance. For the closed-form SBW plug-in estimator, this variance is $\Var\{\theta_{\mathrm{sbw}}(Z)\}$. Similarly, the unadjusted plug-in estimator has asymptotic variance $\Var\{\theta_{\mathrm{unadj}}(Z)\}$. Under Assumption~\ref{assumption:arm_specific_derivative_decomposition}, Lemma~\ref{lem:derivative-level-variance-reduction} gives $\sigma_{\mathrm{sbw}}^2\le\sigma_{\mathrm{unadj}}^2$, as claimed.
\end{proof}
For vector-valued estimands, the same argument applies to any fixed linear combination of the components. Therefore, the multivariate version of the result can be stated as a covariance-matrix comparison: SBW weakly reduces the large-sample variance of every fixed linear contrast of the estimand vector. The coordinatewise variance comparisons are obtained as special cases.

\end{document}